\documentclass[12pt,reqno]{amsart}
\usepackage{graphicx}
\usepackage{amssymb,amsmath}
\usepackage{amsthm}
\usepackage{color}
\usepackage[pdf]{pstricks}
\usepackage{hyperref}

\numberwithin{equation}{section}
\numberwithin{figure}{section}

\newcommand{\beq}{\begin{equation}}
\newcommand{\eeq}{\end{equation}}

\newtheorem{theorem}{Theorem}[section]

\newtheorem{rem}{Remark}[section]
\newtheorem{lemma}{Lemma}[section]
\newtheorem{assumption}{Assumption}[section]

\newtheorem{proposition}{Proposition}[section]
\newtheorem*{rem*}{Remark}
\newtheorem*{lemma*}{Lemma}

\newcommand{\norm}[1]{\left\lVert#1\right\rVert}

\begin{document}

\title[Ground state in a harmonic potential]{\bf Ground state in the energy super-critical Gross-Pitaevskii equation with a harmonic potential}

\author{Piotr Bizon}
\address[P. Bizon]{Institute of Theoretical Physics, Jagiellonian University,
	Krak\'{o}w, Poland}
\email{bizon@th.if.uj.edu.pl}

\author{Filip Ficek}
\address[F. Ficek]{Institute of Theoretical Physics, Jagiellonian University,
	Krak\'{o}w, Poland}
\email{filip.ficek@doctoral.uj.edu.pl}

\author{Dmitry E. Pelinovsky}
\address[D.E. Pelinovsky]{Department of Mathematics and Statistics, McMaster University,
	Hamilton, Ontario, Canada, L8S 4K1}
\email{dmpeli@math.mcmaster.ca}

\author{Szymon Sobieszek}
\address[S. Sobieszek]{Department of Mathematics and Statistics, McMaster University,	Hamilton, Ontario, Canada, L8S 4K1}
\email{sobieszs@mcmaster.ca}

\keywords{Gross--Pitaevskii equation, ground states, Emden--Fowler transformation, shooting method, solution curve}

\begin{abstract}
The energy super-critical Gross--Pitaevskii equation with a harmonic potential is revisited 
in the particular case of cubic focusing nonlinearity and dimension $d \geq 5$. In order to prove 
the existence of a ground state (a positive, radially symmetric solution in the energy space), 
we develop the shooting method and deal with a one-parameter family of classical solutions to an initial-value problem for the stationary equation. We prove that the 
solution curve (the graph of the eigenvalue parameter versus the supremum norm) is oscillatory for $d \leq 12$ and 
monotone for $d \geq 13$. Compared to the existing literature, rigorous asymptotics are derived by constructing three families of solutions 
to the stationary equation with functional-analytic rather than geometric methods.
\end{abstract}

\date{\today}
\maketitle

\section{Introduction}

The main subject of this study is the focusing nonlinear Schr\"odinger (NLS) equation with an isotropic harmonic potential given by 
\begin{equation}\label{nls}
i \partial_t w =-\Delta w +|x|^2 w - |w|^{2p} w\,,
\end{equation}
where $w(t,x) : \mathbb{R}\times \mathbb{R}^d \to \mathbb{C}$ and $p > 0$.
In the physically relevant dimensions $d = 1,2,3$ and for cubic or quintic powers $p = 1,2$, the NLS equation \eqref{nls}, also known as the Gross-Pitaevskii (GP) equation, is used to describe the macroscopic behaviour of the Bose-Einstein condensate with attractive atomic interactions in  a harmonic trap. From a mathematical viewpoint, it is a prototype model of dynamics of nonlinear dispersive waves subject to a confining potential and from this perspective it is interesting to consider higher dimensions $d$. 

We recall the mass and energy which are formally conserved in time evolution:
\begin{equation}
\label{mass} 
M(w) = \int_{\mathbb{R}^d} |w|^2 dx
\end{equation}
and
\begin{equation}
\label{energy-E} 
E(w) = \int_{\mathbb{R}^d} \left( |\nabla w|^2 + |x|^2 |w|^2 - \frac{1}{p+1} |w|^{2p+2} \right) dx.
\end{equation}
In the case of the NLS equation without the potential, 
the scaling transformation 
\begin{equation}
w(t,x) \mapsto w_L(t,x)=L^{\frac{1}{p}} w(L^2 t, Lx), \quad L > 0, 
\end{equation}
leaves the NLS equation invariant of $L$ but changes the mass and energy as follows:
\begin{equation}
\label{scaling} 
M(w_L) = L^{\frac{2}{p}-d} M(w), \quad 
E(w_L) = L^{\frac{2}{p}+2-d} E(w).
\end{equation}
The mass-critical case is $p = \frac{2}{d}$ and the energy-critical case is $p = \frac{2}{d-2}$ (if $d \geq 3$). This work addresses the energy-supercritical case $p > \frac{2}{d-2}$ (if $d \geq 3$). In order to simplify the technical details of presentation, we fix the cubic nonlinearity power $(p=1)$ and consider the energy-supercritical case $d\geq 5$, where $d$ is always taken as a natural number.

Substituting the ansatz $w(t,x)=e^{-i\lambda t} u(x)$ into equation \eqref{nls} with $p = 1$, where $\lambda$ is a real parameter, one obtains the stationary GP equation
\begin{equation}
\label{statGP}
-\Delta u +|x|^2 u - |u|^2 u = \lambda u\,.
\end{equation}
We are interested in the existence and properties of the ground states of the stationary equation \eqref{statGP} which are defined as classical solutions that are positive and decaying to zero at infinity. By the well-known moving planes argument \cite{LiNi}, such solutions must be radially symmetric and monotonically decreasing, hence the stationary equation \eqref{statGP} reduces to the radial differential equation for $u(r) : \mathbb{R}^+ \mapsto \mathbb{R}$, where $r=|x|$ and $\Delta u = u''(r) + \frac{d-1}{r} u'(r)$.
Combining all requirements together, one defines 
the following boundary-value problem for the ground states denoted by 
$\mathfrak{u}$:
\begin{equation}
\label{statGPrad}
\left\{ \begin{array}{ll}
\mathfrak{u}''(r) + \frac{d-1}{r} \mathfrak{u}'(r) - r^2 \mathfrak{u}(r) + \lambda \mathfrak{u}(r) + \mathfrak{u}(r)^3 = 0, \quad & r > 0, \\
\mathfrak{u}(r) > 0, \qquad \qquad \mathfrak{u}'(r) < 0, \quad & \\
\lim\limits_{r \to 0} \mathfrak{u}(r) < \infty, \quad 
\lim\limits_{r \to \infty} \mathfrak{u}(r) = 0. & \end{array} \right.
\end{equation}
Weak solutions to the boundary-value problem (\ref{statGPrad}) are defined in the energy space 
\begin{equation}
\label{energy-space}
\mathcal{E} := \left\{ u \in L^2_r(\mathbb{R}^+) : \quad u' \in L^2_r(\mathbb{R}^+), \quad r u \in L^2_r(\mathbb{R}^+), \quad u \in L^4_r(\mathbb{R}^+) \right\}.
\end{equation}
Both $M(u)$ and $E(u)$ are well-defined for $u \in \mathcal{E}$. 

The linear part of the boundary-value problem (\ref{statGPrad}) 
has the ground state 
given by the Gaussian function $\mathfrak{u}_0(r) = e^{-r^2/2}$. It exists when the eigenvalue parameter $\lambda$ is given by $\lambda_0 = d$. By standard local bifurcation theory \cite{N}, there exist the nonlinear ground state $\mathfrak{u}_b(r)$ with small amplitude $b$  that bifurcates from the linear ground state $\mathfrak{u}_0$ as $b \rightarrow 0$. 
The nonlinear ground state corresponds to $\lambda = \lambda(b)$ satisfying 
the limit $\lambda(b) \to \lambda_0$ as $b \to 0$. An elementary calculation gives  
\begin{equation}
\label{continuation}
\mathfrak{u}_b(r) = b \mathfrak{u}_0(r) + \mathcal{O}(b^3), \qquad 
\lambda(b) = d-2^{-\frac{d}{2}} b^2+\mathcal{O}(b^4).  
\end{equation}
In the energy-subcritical dimensions $1 \leq d \leq 3$, the global behavior of the solution curve can be analyzed by using variational methods and the global bifurcation theory \cite{R} due to the fact that the Sobolev embedding $H^1(\mathbb{R}^d) \cap L^{2,1}(\mathbb{R}^d) \subset L^4(\mathbb{R}^d)$ is compact.
It was shown in \cite{KW} (see also \cite{Fuk} and \cite{Selem2011}) that for each $\lambda<d$, there exists the ground state $\mathfrak{u} \in \mathcal{E}$ satisfying the boundary-value problem (\ref{statGPrad}) with $1 \leq d \leq 3$. Uniqueness 
of the ground state in $\mathcal{E}$ was proven in \cite{H} for $d = 3$ and in \cite{HO} for $d  = 1,2$.  

In the energy-critical case $d=4$, the Sobolev embedding is not compact, nonetheless the existence and uniqueness of a ground state for 
$\lambda \in (0,d)$ was shown in \cite{Selem2011} by modification of the corresponding variational methods for bounded domains from \cite{BN} 
(see also Theorem 6 in \cite{SW2013}). In the energy-supercritical case $d \geq 5$, it was proven in \cite{SK2012} that the solution curve for $\lambda$ is located in a subset of $(0,d)$ and is unbounded in the sense that $\mathfrak{u}(0)$ diverges along the solution curve.

{\em The main goal of this paper} is to analyze the global behaviour of the solution curve in the energy-supercritical case $d \geq 5$. We implement the shooting method pioneered in \cite{JL} in the context of the Liouville--Bratu--Gerlfand problem \cite{JS}. Therefore, we define a solution to the following initial-value problem:
\begin{equation}
\label{statGPshoot}
\left\{ \begin{array}{ll}
f''(r) + \frac{d-1}{r} f'(r) - r^2 f(r) + \lambda f(r) + f(r)^3 =0,\quad & r > 0, \\
f(0)=b, \quad f'(0)=0, & \end{array} \right.
\end{equation}
where $b \in \mathbb{R}$ is a free parameter (assumed to be positive without loss of generality). 

We first prove that for each $b > 0$ and each $\lambda \in \mathbb{R}$, there exists the unique global classical solution to the initial-value problem (\ref{statGPshoot}); moreover, there exists $\lambda = \lambda(b) \in (d-4,d)$ such that the corresponding solution $f$ decays to zero at infinity, so that it gives the ground state $\mathfrak{u} = \mathfrak{u}_b$ of the boundary-value problem (\ref{statGPrad}). The following theorem presents this result.

\begin{theorem}
	\label{theorem-1}
	Fix $d \geq 4$. For every $b > 0$, there exists $\lambda \in (d-4,d)$, labeled as $\lambda(b)$, such that the unique classical solution $f \in C^2(0,\infty)$ 
	to the initial-value problem (\ref{statGPshoot}) with $\lambda = \lambda(b)$ is a solution $\mathfrak{u} = \mathfrak{u}_b \in \mathcal{E}$ to the boundary-value problem (\ref{statGPrad}).
\end{theorem}

\begin{rem}
		\label{rem-unique-1}
	Uniqueness of $\lambda$ in Theorem \ref{theorem-1} for each given $b > 0$ is an open problem.	
\end{rem}

\begin{rem}
	We believe that the shooting argument used to prove 
	Theorem \ref{theorem-1} can be generalized to prove the existence of 
	the $n$-th excited state with $n$ nodes on $\mathbb{R}^+$ 
	for some $\lambda \in (\lambda_n-4,\lambda_n)$, where $\lambda_n := d+4n$ is the $n^{\rm th}$ eigenvalue of the linear problem, $n \in \mathbb{N}$. 
	Such solutions were also considered in \cite{Selem2011}.
\end{rem}

Next, we analyze the behavior of the family of ground states 
parameterized by $b := \mathfrak{u}_b(0)$ as $b \rightarrow \infty$. The existence of the limiting singular solution $f_{\infty}$ for a unique value of $\lambda = \lambda_{\infty}$ such that $\mathfrak{u}_b \to f_{\infty}$ in $\mathcal{E}$ and $\lambda(b) \to \lambda_{\infty}$ as $b \to \infty$, 
was established in \cite{Selem2013}. The limiting singular solution 
$f_{\infty}$ is defined by the following divergent behavior:
\begin{equation}
\label{rate-divergence}
f_{\infty}(r) = \frac{\sqrt{d-3}}{r} \left[ 1 + \mathcal{O}(r^2) \right] \qquad \mbox{\rm as} \quad r \to 0.
\end{equation}
If $f_{\infty} \in C^2(0,\infty)$ and $f_{\infty}$ decays to zero at infinity fast enough, then $f_{\infty} \in \mathcal{E}$ for $d \geq 5$. Convergence $\mathfrak{u}_b \to f_{\infty}$ in $\mathcal{E}$ and $\lambda(b) \to \lambda_{\infty}$ as $b \to \infty$ was shown in \cite{Selem2013} similarly 
to \cite{MP} where the analogous problem was analyzed for the stationary focusing nonlinear Schr\"odinger equation in a ball and without a harmonic potential. 

The limiting singular solution $f_{\infty}$ can be introduced by the change of variables $f(r) = r^{-1} F(r)$, where $F(r)$ is defined as a solution to the following initial value problem:
\begin{equation}
\label{statGPsingular}
\left\{ \begin{array}{ll}
F''(r) + \frac{d-3}{r} F'(r) - \frac{d-3}{r^2} F(r) - r^2 F(r) + \lambda F(r) + \frac{1}{r^2} F(r)^3 = 0,\quad & r > 0, \\
F(0) = \sqrt{d-3}, \quad F'(0)=0. & \end{array} \right.
\end{equation}
For each $\lambda \in \mathbb{R}$, there exists the unique global classical
solution to the initial value problem (\ref{statGPsingular}), moreover, 
there exists a value of $\lambda$ denoted as $\lambda_{\infty}$ such that the corresponding 
solution $F$ decays to zero at infinity. This decaying solution $F$ gives the limiting 
singular solution $f_{\infty}$ after the transformation 
$f(r) = r^{-1} F(r)$. The following theorem was proven in \cite{Selem2013}.

\begin{theorem}
	\label{theorem-Selem}
	Fix $d \geq 5$. There exists a value of $\lambda \in (0,d)$, labeled as $\lambda_{\infty}$, such that the unique classical 
	solution $F \in C^2(0,\infty)$ to the initial-value problem (\ref{statGPsingular}) with $\lambda = \lambda_{\infty}$ satisfies $F(r) > 0$ and $F'(r) < 0$ for every $r > 0$ and $F(r) \to 0$ as $r \to \infty$ such that $f_{\infty}(r) = r^{-1} F(r)$ belongs to $\mathcal{E}$.
\end{theorem}

\begin{rem} 
		\label{rem-unique-2}
	Uniqueness of the value	of $\lambda_{\infty}$ in Theorem \ref{theorem-Selem} was claimed in \cite[Section 4]{Selem2013} by analyzing the behavior of the quotient between two hypothetical solutions of (\ref{statGPsingular}) for two different values of $\lambda$. However, we believe the proof is incorrect, see Remark \ref{remark-error} below.
\end{rem}

\begin{rem}
	The proof of Theorem \ref{theorem-1} is similar to the proof of Theorem  \ref{theorem-Selem} in \cite{Selem2013} but we have to work with the different initial-value problem (\ref{statGPshoot}) compared to (\ref{statGPsingular}). We also prove the fast decay to zero at infinity  and this allows us to simplify some arguments from \cite{Selem2013}. For the reader's convenience, we also provide a simpler proof of Theorem \ref{theorem-Selem} by using our technique.
\end{rem}

Finally, we consider the convergence of $\lambda(b) \to \lambda_{\infty}$ as $b \to \infty$ depending on the dimension $d \geq 5$, which was not explored in \cite{SK2012,Selem2013}. We show under a technical non-degeneracy assumption that the solution curve has an oscillatory (snaking) behavior for $5\leq d\leq 12$ and a monotone behavior for $d \geq 13$. The following theorem presents the corresponding result.

\begin{theorem}
	\label{theorem-2}
	Assume that $\lambda_{\infty}$ is given by Theorem \ref{theorem-Selem} 
	and Assumptions \ref{assumption-1} and \ref{assumption-2} are satisfied.
	Then, there exists $b_0\in [0,\infty)$ such that for every $b>b_0$ the value of $\lambda$ in Theorem  \ref{theorem-1}, denoted by $\lambda(b)$, is uniquely defined near $\lambda_{\infty}$ such that $\lim\limits_{b \to \infty} \lambda(b) = \lambda_{\infty}$.
	Moreover, for $5 \leq d \leq 12$, there exist constants $A_{\infty} > 0$ and $\delta_{\infty} \in \mathbb{R}$ such that
	\begin{equation}
	\label{snake}
	\lambda(b) - \lambda_{\infty} \sim A_{\infty} b^{-\beta} \sin(\alpha \ln b + \delta_{\infty}) \quad \mbox{\rm as} \quad b \to \infty,
	\end{equation}
	where
	\begin{equation}
	\label{alpha-beta}
	\alpha = \frac{\sqrt{-d^2 + 16 d - 40}}{2}, \quad \beta = \frac{d-4}{2},
	\end{equation}
	whereas for $d \geq 13$, there exists $B_{\infty} > 0$ such that 
		\begin{equation}
	\label{monotone}
	\lambda(b) - \lambda_{\infty} \sim B_{\infty} b^{\kappa_+} \quad \mbox{\rm as} \quad b \to \infty,
	\end{equation}
	where
	\begin{equation}
	\label{kappa-plus}
	\kappa_+ = -\frac{d-4}{2} + \frac{\sqrt{d^2 - 16 d + 40}}{2}.
	\end{equation}
\end{theorem}

\begin{rem}
	In (\ref{snake}) and (\ref{monotone}), $f(b) \sim g(b)$ denotes
	the asymptotic correspondence in the sense $g(b) \to 0$ as $b \to \infty$ and $\lim\limits_{b \to \infty}\frac{|f(b)-g(b)|}{|g(b)|} = 0$.
	Moreover, the asymptotic correspondence $f(b) \sim g(b)$ can be differentiated term by term.
\end{rem}

\begin{rem}
	\label{rem-unique-3}
	If the value of $\lambda_{\infty}$ in Theorem \ref{theorem-Selem} is not unique, then for each $\lambda_{\infty}$, which is isolated under Assumptions \ref{assumption-1} and \ref{assumption-2}, there exists the solution curve of Theorem \ref{theorem-2} with the oscillatory or monotone behavior. Our numerical results indicate that $\lambda_{\infty}$ in Theorem \ref{theorem-Selem} is unique; moreover, $\lambda(b)$ in Theorem \ref{theorem-1} is unique for every $b > 0$.
\end{rem}

\begin{rem}
	The oscillatory behavior similar to the one in (\ref{snake}) was obtained in \cite{Budd_Norbury1987,Budd1989,DF} for the stationary focusing nonlinear Schr\"odinger equation in a ball and without a harmonic potential. The similarity is explained by the same linearization of the stationary equation near the origin after the Emden--Fowler transformation \cite{F}. While the previous works explore geometric methods, the main approach we undertake to prove Theorem \ref{theorem-2} is based on the functional-analytical methods. In particular, we construct three families of solutions to the same differential equation: one family extends the solution of the initial value problem (\ref{statGPshoot}) in new variables, the other family extends the solution of the initial value problem (\ref{statGPsingular}), and the third family describes solution decaying to zero at infinity. By using our methods, we see necessity of adding technical non-degeneracy assumptions (Assumptions \ref{assumption-1} and \ref{assumption-2}), which were not mentioned previously.
\end{rem}

Figure \ref{fig:blambda} illustrates the result of Theorem \ref{theorem-2} and shows the numerically computed solution curve (the graph of $\lambda$ as a function of $b$) for $d = 5$ (left) and $d = 13$ (right). In agreement with 
Theorem \ref{theorem-2}, we confirm the oscillatory behavior in the former case 
and the monotone behavior in the latter case. We also note that 
the unique value of $\lambda = \lambda(b)$ is found for every $b > 0$ in both cases (see Remarks \ref{rem-unique-1}, \ref{rem-unique-2}, and \ref{rem-unique-3}).

\begin{figure}[htp!]
	\centering
	\includegraphics[width=0.45\textwidth]{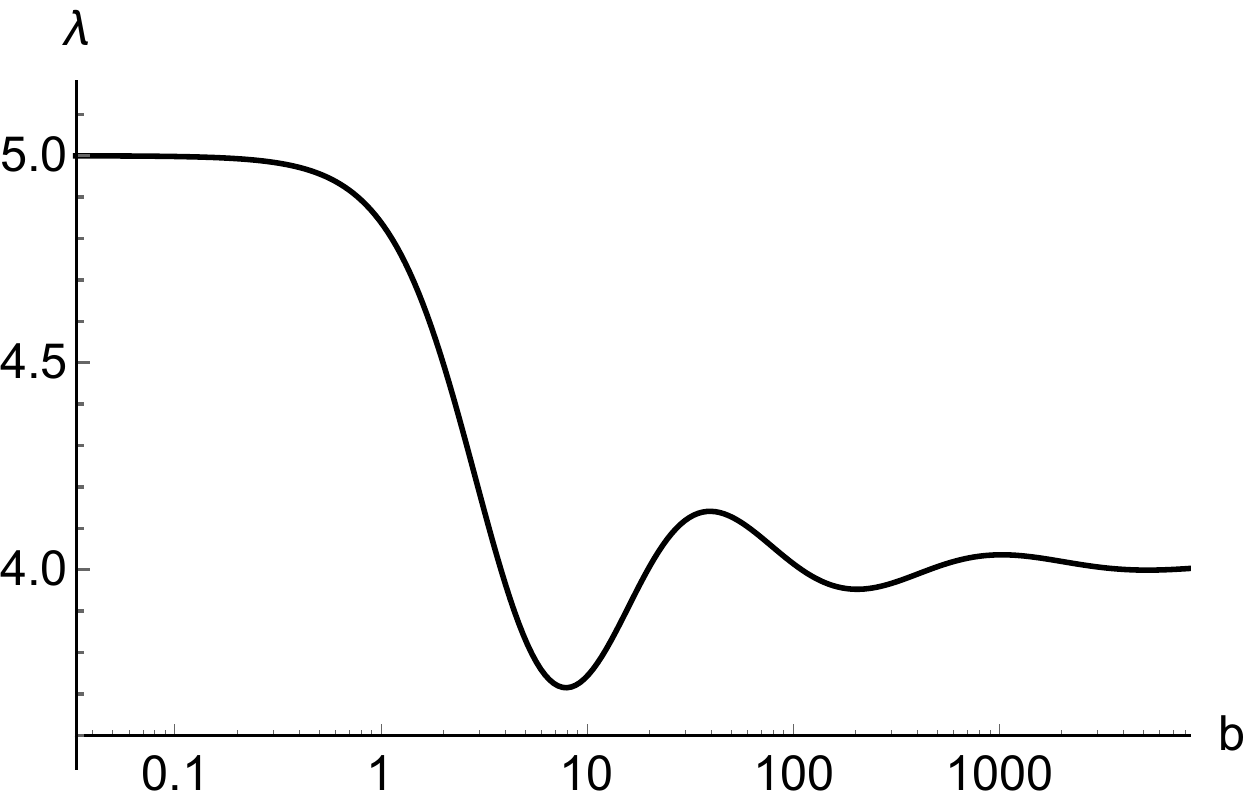}\qquad
	\includegraphics[width=0.45\textwidth]{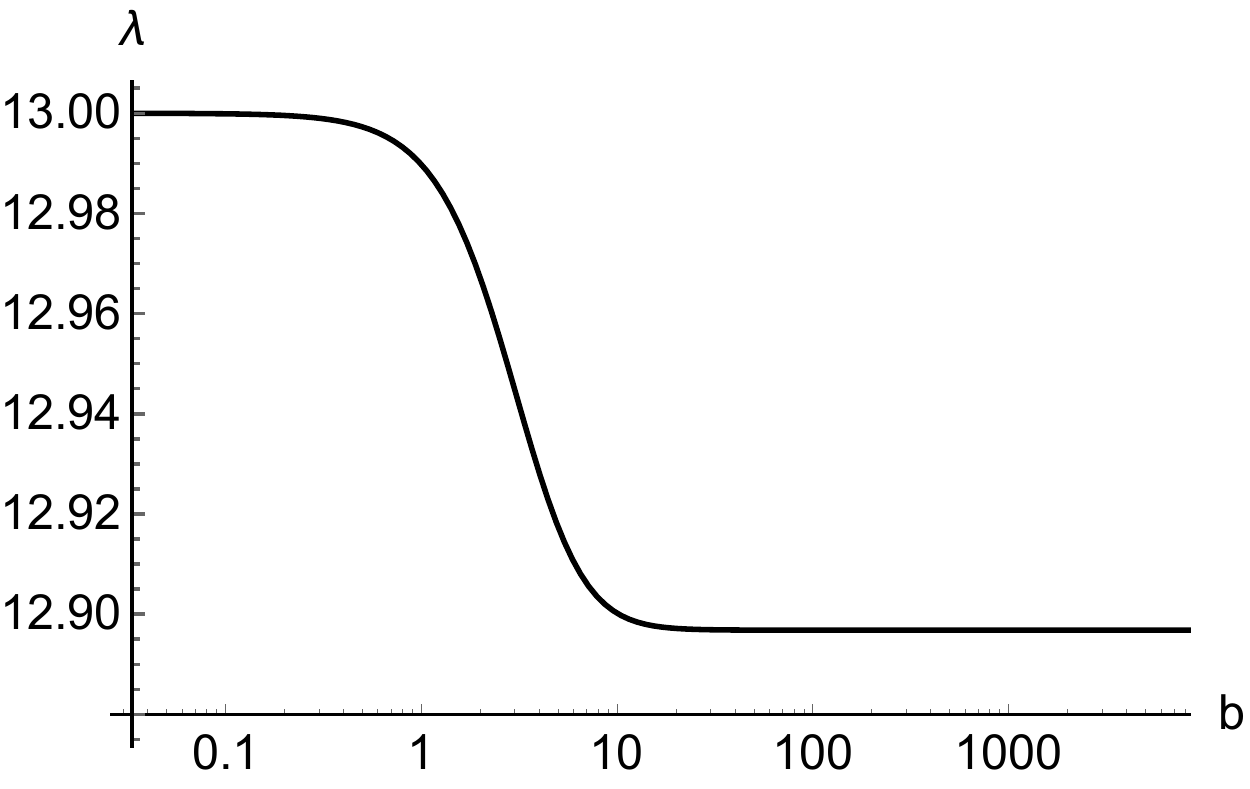}
	\caption{Graph of $\lambda$ as a function of $b$ for the ground state of the boundary-value problem (\ref{statGPrad}) for $d=5$ (left) and $d=13$ (right).}
	\label{fig:blambda}
\end{figure}

Table \ref{tab:notations} lists differential equations, their solutions, 
their asymptotic behaviors, and the relations between the solutions. 
This table helps the readers to get oriented between the differential equations 
(\ref{statGPshoot}) and (\ref{statGPsingular}) as well as 
their analogues (\ref{eq-psi}) and (\ref{eq-psi-singular}) arising 
after the Emden--Fowler transformation.

\vspace{0.2cm}

{\bf Notations.} We denote $A = \mathcal{O}(\varepsilon)$ as $\varepsilon \to 0$ 
if there exists an $\varepsilon_0 > 0$ and an $\epsilon$-independent constant $C > 0$ such that $|A| \leq C \varepsilon$ if $\varepsilon \in (0,\varepsilon_0)$. 

We denote the space of square integrable distributions by $L^2_r$ and equip it with the inner product
$\langle \cdot, \cdot \rangle_{L^2_r}$ and the induced norm $\| \cdot \|_{L^2_r}$, where 
\[
\| u \|_{L^2_r} := \left( \int_0^{\infty} |u(r)|^2 r^{d-1} dr \right)^{1/2}.
\]
The Lebesgue space $L^p_r$ is introduced similarly for $1 \leq p < \infty$, whereas  the space of bounded functions $L^{\infty}_r$ is equipped with the standard supremum norm 
$$\| u \|_{L^{\infty}_r} = \sup\limits_{r \in [0,\infty)} |u(r)|.$$ 
Finally, $H^1(\mathbb{R}^d)$ and $L^{2,1}(\mathbb{R}^d)$ are standard 
Sobolev and weighted Lebesgue spaces of functions in $\mathbb{R}^d$ given by 
$H^1(\mathbb{R}^d) := \{ f \in L^2(\mathbb{R}^d) : f' \in L^2(\mathbb{R}^d)\}$ and $L^{2,1}(\mathbb{R}^d) : \{ f \in L^2(\mathbb{R}^d) : xf \in L^2(\mathbb{R}^d)\}$.

\vspace{0.1cm}

\begin{table}[h]
	\centering
	\begin{tabular}{|c|c|l|l|}
		\hline
		Equation  & Solution & Asymptotic behavior & Relations between the solutions\\
		\hline
		(\ref{statGPshoot}) & $f(r)$ & $f(r) = b + \mathcal{O}(r^2) \;\; \mbox{\rm as} \;\; r \to 0$ & $f \in \mathcal{E}$ if $\lambda = \lambda(b)$:  $f(r) = \mathfrak{u}(r)$\\
		\hline
		(\ref{statGPsingular}) & $F(r)$ & $F(r) = \sqrt{d-3} + \mathcal{O}(r^2) \;\; \mbox{\rm as} \;\; r \to 0$& $f_{\infty} \in \mathcal{E}$ if $\lambda = \lambda_{\infty}$: $f_{\infty}(r) = r^{-1} F(r)$\\
		\hline
		(\ref{eq-psi}) & $\psi(t)$ & $\psi(t) = b + \mathcal{O}(e^{2t}) \;\; \mbox{\rm as} \;\; t \to -\infty$& $\psi(t) = f(e^t)$\\
		\hline
		(\ref{eq-psi-singular})  & $\Psi_b(t)$ & $\Psi_b(t) = b e^t + \mathcal{O}(e^{3t}) \;\; \mbox{\rm as} \;\; t \to -\infty$& $\Psi_b(t) = e^t \psi(t) = e^t f(e^t)$ \\
		\hline
		(\ref{eq-psi-singular})  & $\Psi_\lambda(t)$ & $\Psi_{\lambda}(t) = \sqrt{d-3} + \mathcal{O}(e^{2t}) \;\; \mbox{\rm as} \;\; t \to -\infty$& $\Psi_{\lambda}(t) = F(e^{t})$\\
		\hline  (\ref{eq-psi-singular}) & $\Psi_C(t)$ & $\Psi_C(t) \sim C e^{\frac{\lambda - d + 2}{2} t} e^{-\frac{1}{2} e^{2t}} \;\; \mbox{\rm as} \;\; t \to +\infty$& $\Psi_C(t) = \Psi_{\lambda}(t)$ if $\lambda = \lambda_{\infty}$, $C = C_{\infty}$\\
		\hline
	\end{tabular}
	\vspace{0.1cm} 
	\caption{Table of differential equations and solutions used in this paper.}
	\label{tab:notations}
\end{table}

{\bf Organization of the paper.}  Section \ref{sec-preliminary} describes well-known results about the boundary-value problem (\ref{statGPrad}). Section \ref{sec-ivp} reports analysis of solutions to the initial-value problem (\ref{statGPshoot}). Sections \ref{sec-proof-1} and \ref{sec-proof-2} contain the proof of Theorems \ref{theorem-1} and \ref{theorem-Selem} respectively. The proof of Theorem \ref{theorem-2} is developed in Section \ref{sec-proof-3}. Section \ref{sec-conclusion} 
gives a summary and describes open problems.

\section{Preliminary results}
\label{sec-preliminary}

Here we collect together three well-known results regarding existence of 
nontrivial solutions to the boundary-value problem (\ref{statGPrad}) in the energy space $\mathcal{E}$.

\begin{proposition}
\label{prop-1}
For every $d\geq 1$ and $\lambda \in [d,\infty)$,
no solutions of the boundary-value problem (\ref{statGPrad}) exist in $\mathcal{E}$.
\end{proposition}

\begin{proof}
It is well known (see, e.g., \cite{He}) that the operator $L_0 := -\Delta + |x|^2$ is self-adjoint in $L^2(\mathbb{R}^d)$. The ground state of $L_0$ is given up to a normalization by the Gaussian function $\mathfrak{u}_0(r) = e^{-\frac{1}{2}r^2}$ and corresponds to the smallest eigenvalue $\lambda_0 = d$. The linear ground state $\mathfrak{u}_0$ satisfies the following boundary value problem:
\begin{equation}
\label{eig-L0}
\left\{ \begin{array}{ll}
\mathfrak{u}_0''(r) + \frac{d-1}{r} \mathfrak{u}_0'(r) - r^2 \mathfrak{u}_0(r) = -d \mathfrak{u}_0(r), \quad & r > 0, \\
\mathfrak{u}_0(r) > 0, \qquad \qquad e'_0(r) < 0, & r > 0, \\
\lim\limits_{r \to 0} \mathfrak{u}_0(r) < \infty, \quad  
\lim\limits_{r \to \infty} \mathfrak{u}_0(r) = 0. & \end{array} \right.
\end{equation}
By projecting \eqref{statGPrad} to $\mathfrak{u}_0$ and integrating by parts with the use of (\ref{eig-L0}), we obtain:
\begin{equation*}
    -d\langle \mathfrak{u}_0, \mathfrak{u} \rangle _{L^2_r} + \lambda\langle \mathfrak{u}_0, \mathfrak{u} \rangle_{L^2_r} + \langle \mathfrak{u}_0, \mathfrak{u}^3 \rangle_{L^2_r} = 0
\end{equation*}
which implies 
\begin{equation*}
    d-\lambda = \frac{\langle \mathfrak{u}_0, \mathfrak{u}^3 \rangle_{L^2_r}}{\langle \mathfrak{u}_0, \mathfrak{u} \rangle_{L^2_r}}.
\end{equation*}
Since $\langle \mathfrak{u}_0, \mathfrak{u} \rangle_{L^2_r} > 0$ and $\langle \mathfrak{u}_0, \mathfrak{u}^3 \rangle_{L^2_r} > 0$, we must have $\lambda<d$ for every solution $\mathfrak{u} \in \mathcal{E}$ of the boundary-value problem \eqref{statGPrad}.
\end{proof}

\begin{proposition}
\label{prop-2}
For every $d \geq 4$ and $\lambda \in (-\infty,d-4]$,
no solutions of the boundary-value problem (\ref{statGPrad}) exist in $\mathcal{E}$.
\end{proposition}

\begin{proof}
It follows from multiplication of (\ref{statGPrad}) by $r^{d-1} \mathfrak{u}$ that if $\mathfrak{u} \in \mathcal{E}$, then
\begin{equation}\label{eq:poh_id1}
    \norm{\mathfrak{u}'}^2_{L^2_r} + \norm{r \mathfrak{u}}^2_{L^2_r}-\lambda\norm{\mathfrak{u}}^2_{L^2_r}-\norm{\mathfrak{u}}^4_{L^4_r} = 0.
\end{equation}
Similarly, it follows from multiplication of (\ref{statGPrad}) by $r^d \mathfrak{u}'(r)$ and integration by parts that 
\begin{equation}\label{eq:poh_id2}
    (d-2)\norm{\mathfrak{u}'}^2_{L^2_r}+(d+2)\norm{r \mathfrak{u}}^2_{L^2_r}-\lambda d \norm{\mathfrak{u}}^2_{L^2_r} -\frac{1}{2}d\norm{\mathfrak{u}}^4_{L^4_r} = 0.
\end{equation}
Combining \eqref{eq:poh_id1} and \eqref{eq:poh_id2} yields the Pohozaev identity \cite{P}:
\begin{equation}
\label{pohozaev}
    4\norm{r \mathfrak{u}}^2_{L^2_r}-2\lambda\norm{\mathfrak{u}}^2_{L^2_r}+\frac{1}{2}\left(d-4\right)\norm{\mathfrak{u}}^4_{L^4_r} = 0.
\end{equation}
Hence, no nonzero solution $\mathfrak{u} \in \mathcal{E}$ exists if $\lambda\leq 0$ and $d\geq 4$.

Furthermore, since $d$ is the lowest eigenvalue of $L_0 = -\Delta + r^2$, we obtain similarly to \cite{BN}:
\begin{equation}\label{rayleigh}
d \| \mathfrak{u} \|_{L^2}^2 \leq \| \mathfrak{u}'\|_{L^2}^2+\|r \mathfrak{u}\|_{L_2}^2 = \lambda \| \mathfrak{u} \|_{L^2}^2 + \| \mathfrak{u} \|_{L^4}^4.
\end{equation}
If $d \neq 4$, then $\| \mathfrak{u} \|_{L^4}^4$ can be expressed by using  (\ref{pohozaev}), after which inequality (\ref{rayleigh})  yields
\begin{equation}\label{bound1}
\lambda \geq d-4 +\frac{8}{d} \frac{\|r \mathfrak{u}\|_{L_2}^2}{\| \mathfrak{u}\|_{L_2}^2},
\end{equation}
hence no nonzero solution $\mathfrak{u} \in \mathcal{E}$ exists if $\lambda \leq d-4$.
\end{proof}

\begin{proposition}
\label{prop-3}
For every $d\geq 1$, there exists a unique solution of the boundary-value problem (\ref{statGPrad}) in $\mathcal{E} \cap L^{\infty}$
for $\lambda \in (d - \delta, d)$ with sufficiently small $\delta$ such that $\| \mathfrak{u} \|_{L^{\infty}_r} \to 0$ as $\lambda \to d$.
\end{proposition}

\begin{proof}
For $1 \leq d \leq 4$, the proof follows by the standard Lyapunov--Schmidt theory (see Theorem 2.1 in \cite{Selem2011} and
references therein). For $d \geq 5$, the proof follows by the compactification of the nonlinear term
for the standard Lyapunov--Schmidt theory and by the Moser's iteration argument to control the $L^{\infty}$-norm of
the bifurcating solution and thus the nonlinear term (see Theorem 5 in \cite{SK2012} and references therein).
\end{proof}

\section{Existence of solutions to the initial-value problem (\ref{statGPshoot})}
\label{sec-ivp} 

Here we consider the differential equation 
\begin{equation}\label{eq}
f''(r) + \frac{d-1}{r} f'(r) -r^2 f(r) + \lambda f(r) + f(r)^3 = 0, \quad r > 0,
\end{equation}
and prove several results regarding existence of classical solutions to this 
differential equation.

The first result shows that the additional condition
$f'(0) = 0$ does not over-determine the initial-value problem (\ref{statGPshoot}) at the singularity point $r = 0$ as long as 
the classical solution $f(r)$ to the differential equation (\ref{eq}) 
is bounded as $r \to 0$.

\begin{lemma}
\label{prop-4}
For every $d\geq 1$ and every $\lambda \in \mathbb{R}$, assume that there exists a classical solution $f \in C^2(0,r_0)$, $r_0 > 0$ to the differential equation
(\ref{eq}) such that $\displaystyle f(0) := \lim_{r \to 0} f(r) < \infty$. Then, $$
\displaystyle f'(0) := \lim_{r \to 0} f'(r) = 0.
$$
\end{lemma}

\begin{proof}
One can rewrite the differential equation \eqref{eq} in the self-adjoint form:
\begin{equation}\label{eq:self_adj}
\frac{d}{dr}\left[r^{d-1}f'(r)\right] = 
r^{d-1}\left[ r^2 f(r) -\lambda f(r) - f(r)^3 \right].
\end{equation}
The right-hand side of \eqref{eq:self_adj} is integrable as $r\to 0$ if $f$ is bounded near $r=0$. Then $\displaystyle \lim_{r\to 0}r^{d-1}f'(r)=0$ and integration of \eqref{eq:self_adj} on $[0,r]$ yields
\begin{equation}
    f'(r) = \frac{1}{r^{d-1}}\int_0^r s^{d-1}\left[s^2f(s)-\lambda f(s) - f(s)^3\right]ds.
\end{equation}
Since $\displaystyle\lim_{r\to 0}f'(r)$ is an indeterminate form $\left[\frac{0}{0}\right]$, we can apply L'Hospital's rule and obtain 
\begin{equation*}
\lim_{r\to 0}f'(r) = \lim_{r\to 0} \frac{r^{d-1}\left[r^2f(r)-\lambda f(r) - f(r)^3\right]}{\left(d-1\right)r^{d-2}} = \frac{1}{d-1}\lim_{r\to 0}r\left[r^2f(r)-\lambda f(r) - f(r)^3\right] = 0,
\end{equation*}
since $f$ is bounded near $r=0$. Hence, $\displaystyle f'(0):=\lim_{r\to 0 }f'(r)=0$.
\end{proof}

\begin{rem}
	A similar result but for $d \geq 5$ can be stated about the initial-value problem (\ref{statGPsingular}). If $F(0) = \sqrt{d-3}$ for a classical solution $F \in C^2(0,r_0)$ with $r_0 > 0$, then $F'(0) = 0$. Indeed, the differential equation in the initial-value problem (\ref{statGPsingular}) can be written in the self-adjoint form:
	$$
	\frac{d}{dr} \left[ r^{d-3} F'(r) \right] = r^{d-5} \left[ (d-3) F(r) - F(r)^3 - \lambda r^2 F(r) + r^4 F(r) \right]
	$$
	The right-hand side is integrable for $d \geq 5$ so that integration gives 
	$$
	F'(r) = \frac{1}{r^{d-3}} \int_0^r s^{d-5} \left[ (d-3) F(s) - F(s)^3 - \lambda s^2 F(s) + s^4 F(s) \right] ds
	$$
	By using the L'Hospital's rule twice, we get if $F(0) = \sqrt{d-3}$:
\begin{eqnarray*}
	\lim_{r \to 0} F'(r) & = & \lim_{r\to 0} \frac{(d-3) F(r) - F(r)^3 - \lambda r^2 F(r) + r^4 F(r)}{(d-3) r} \\
	& = & \lim_{r\to 0} \frac{(d-3) - 3 F(r)^2}{(d-3)} F'(r) \\
	& = & -2 \lim_{r \to 0} F'(r),
\end{eqnarray*}
	so that $\lim\limits_{r\to 0} F'(r) = 0$.
\end{rem}

Singularity at $r = 0$ of the differential equation (\ref{eq}) is unfolded 
using the following Emden--Fowler transformation \cite{F}: 
\begin{equation}
\label{transformation-r-t}
r = e^t, \qquad f(r) = \psi(t), \qquad f'(r) = e^{-t} \psi'(t).
\end{equation} 
By chain rule, the second-order differential equation (\ref{eq}) for $f(r)$ becomes
\begin{equation}
\label{eq-psi}
\psi''(t) + (d-2) \psi'(t) = -e^{2t} \left( \lambda \psi(t) + \psi(t)^3 \right) + e^{4t} \psi(t), \quad t \in \mathbb{R}.
\end{equation} 
The next result guarantees that there exists a unique local classical solution to the initial-value problem (\ref{statGPshoot}). The proof is developed from analysis of the existence of the bounded solutions of the differential equation (\ref{eq-psi}) as $t \to -\infty$.

\begin{lemma}
\label{lemma-1}
For every $d \geq 3$, $\lambda \in \mathbb{R}$, and $b > 0$, there exists $r_0 > 0$
and a unique classical solution $f \in C^2(0,r_0)$ to the initial-value problem (\ref{statGPshoot})
such that $f(r) > 0$ and $f'(r) < 0$ for $r \in (0,r_0)$.
\end{lemma}

\begin{proof}
By using the Emden--Fowler transformation (\ref{transformation-r-t}), the
initial conditions $f(0) = b$ and $f'(0) = 0$ in the initial-value 
problem (\ref{statGPshoot}) become the following boundary conditions
\begin{equation}
\label{bc}
\left\{ \begin{array}{l} \psi(t) \to b, \\ \psi'(t) \to 0, \end{array} \right. \quad \mbox{\rm as } \;\; t \to -\infty.
\end{equation}
By the method of variation of parameters, we rewrite the differential equation (\ref{eq-psi}) with the boundary conditions
(\ref{bc}) as the following Volterra's integral equation:
\begin{equation}
\label{volterra-psi}
\psi(t) = A(\psi)(t) := b + \frac{1}{d-2} \int_{-\infty}^t \left[ 1 - e^{-(d-2) (t - t')} \right] F(\psi(t'),t') d t',
\end{equation}
where $F(\psi,t) := -e^{2t} \left( \lambda \psi + \psi^3 \right) + e^{4t} \psi$. The integral operator $A$ is considered on $\psi$ in
the Banach space $L^{\infty}(-\infty,t_0)$, where $-\infty < t_0 \ll -1$. It follows from (\ref{volterra-psi}) that 
$$
\| A(\psi) \|_{L^{\infty}} \leq b +
\left[ \frac{1}{2d} \left(|\lambda| + \| \psi \|_{L^{\infty}}^2 \right) + \frac{1}{4(d+2)} e^{2t_0}
\right] \| \psi \|_{L^{\infty}} e^{2t_0},
$$
and
\begin{eqnarray*}
\| A(\psi) - A(\phi) \|_{L^{\infty}} & \leq & \left[ \frac{1}{2d} \left(|\lambda| + (\| \psi \|_{L^{\infty}}
+ \| \phi \|_{L^{\infty}})^2 \right) + \frac{1}{4(d+2)} e^{2t_0} \right] \| \psi - \phi \|_{L^{\infty}} e^{2 t_0}.
\end{eqnarray*}
If $t_0$ is a sufficiently large negative number, then
$A : B_{2b} \to B_{2b}$ is a contraction operator in the ball $B_{2b} \subset L^{\infty}(-\infty,t_0)$
of a fixed radius $2b > 0$. By Banach's fixed-point theorem, there exists the unique solution
$\psi \in B_{2b} \subset L^{\infty}(-\infty,t_0)$ to the integral equation (\ref{volterra-psi}).

Since $F(\psi(\cdot),\cdot) \in  L^1(-\infty,t_0)$ if $\psi \in L^{\infty}(-\infty,t_0)$, the fixed point of the integral equation \eqref{volterra-psi} is in $C^0(-\infty, t_0)$. Since $F(\psi(\cdot),\cdot) \in  C^0(-\infty,t_0)$ if $\psi \in C^0(-\infty, t_0)$, the fixed point of the integral equation \eqref{volterra-psi} is in $C^1(-\infty,t_0)$, so that differentiation of (\ref{volterra-psi}) yields
\begin{eqnarray}
\label{psi-prime}
\psi'(t) = \int_{-\infty}^t e^{-(d-2) (t - t')} F(\psi(t'),t') d t'.
\end{eqnarray}
Finally, since $F(\psi(\cdot),\cdot) \in C^1(-\infty,t_0)$ if 
$\psi \in C^1(-\infty, t_0)$, the fixed point of the integral equation \eqref{volterra-psi} is in
$C^2(-\infty,t_0)$. By the chain rule, this implies that $f \in C^2(0,r_0)$ for small $r_0 > 0$.

By continuity of the solution, we have $\psi(t) > 0$ for $t \in (-\infty,t_0)$ if $t_0$ is a sufficiently large negative number. The transformation formula $f(r) = \psi(t)$ yields $f(r) > 0$ for $r \in (0,r_0)$ with small positive $r_0$. Furthermore, thanks to the bound
$$
\| \psi - b \|_{L^{\infty}(-\infty,t_0)} \leq C_1 e^{2 t_0}
$$
with some $C_1 > 0$, it follows from (\ref{psi-prime}) that
\begin{eqnarray*}
\| \psi' + (\lambda b + b^3) d^{-1} e^{2t} \|_{L^{\infty}(-\infty,t_0)} \leq C_2 e^{4 t_0},  
\end{eqnarray*}
for some $C_2 > 0$.
Hence $\psi'(t) < 0$ for $t \in (-\infty,t_0)$ if $t_0$ is a large negative number. By the transformation formula $f'(r) = e^{-t} \psi'(t)$, this yields $f'(r) < 0$ for $r \in (0,r_0)$ with small $r_0 > 0$. 
\end{proof}

\begin{rem}
The solution $\psi \in C^2(-\infty,t_0)$ in Lemma \ref{lemma-1} satisfies the asymptotic expansion 
\begin{equation}
\label{asymptotic-psi-prime}
\psi(t) = b - \frac{\lambda b + b^3}{2d} e^{2t} + \mathcal{O}(e^{4t}) \quad \mbox{\rm as} \quad t \to -\infty.
\end{equation}
This expansion implies that
\begin{equation}
\label{limit-f-1}
\lim_{r \to 0} f'(r) = \lim_{t \to -\infty} e^{-t} \psi'(t) = 0
\end{equation}
and
\begin{equation}
\label{limit-f-2}
\lim_{r \to 0} f''(r) = \lim_{t \to -\infty} e^{-2t} \left[ \psi''(t) - \psi'(t) \right] = -(\lambda b + b^3) d^{-1}.
\end{equation}
where the first limit is in agreement with Lemma \ref{prop-4}.
\end{rem}

\begin{rem}
	The proof of Lemma \ref{lemma-1} is based on classical fixed-point arguments, which is the main technical tool used in the rest of this paper.
\end{rem}

Another solution to the same differential equation (\ref{eq}) can be constructed 
from the condition that $f(r), f'(r) \to 0$ as $r \to \infty$. In order to construct such decaying solutions, we reformulate the second-order equation (\ref{eq-psi}) as the following three-dimensional dynamical system:
\begin{equation}
\label{3d-system}
\left\{ \begin{array}{l}
x' = 2x, \\
\psi' = \varphi, \\
\varphi' = (2-d) \varphi - x (\lambda \psi + \psi^3) + x^2 \psi,
\end{array} \right.
\end{equation}
where $x(t) := e^{2t}$ and the prime stands for the derivative in $t$.
The following lemma identifies the admissible behavior of classical 
solutions to the differential equation (\ref{eq}) such that 
$f(r), f'(r) \to 0$ as $r \to \infty$.

\begin{lemma}
\label{lemma-0}
For every $d \geq 1$ and every $\lambda \in \mathbb{R}$, there exists $r_0 > 0$ and a one-parameter family of classical solutions $f \in C^2(r_0,\infty)$ to the differential equation (\ref{eq}) such that $f(r), f'(r) \to 0$ as $r \to \infty$. Moreover, 
\begin{equation}
    \label{asymptotics-infinity}
f(r) \sim C r^{\frac{\lambda - d}{2}} e^{-\frac{1}{2} r^2} \quad \mbox{\rm as} \quad r \to \infty,
\end{equation}
for some $C \in \mathbb{R}$, where $f(r) \sim g(r)$ is the asymptotic correspondence which can be differentiated.
\end{lemma}

\begin{proof}
The limit $r \to \infty$ corresponds to the limit $t \to +\infty$
due to the transformation (\ref{transformation-r-t}). 
If $f(r), f'(r) \to 0$ as $r \to \infty$, then 
$x(t) \to \infty$, $\psi(t) \to 0$, and $\varphi(t)/\sqrt{x(t)} \to 0$ as $t \to +\infty$.
We introduce the following transformation of variables:
\begin{equation}
    \label{3d-transformation}
x(t) = \frac{1}{y(\tau)}, \quad \psi(t) = \psi(\tau), \quad \varphi(t) = \frac{\phi(\tau)}{y(\tau)},
\end{equation}
where $\tau$ is the new time variable defined by the chain rule $dt = y(\tau) d\tau$.
For convenience, we do not change the notation for $\psi$ that now depends on $\tau$.
By integrating $dt = y(\tau) d \tau$ or equivalently, $d \tau = x(t) dt$ with the initial condition $\tau = 0$ at $t = 0$, we obtain
\begin{equation}
    \label{tau-t}
\tau = \frac{1}{2} \left( e^{2t} - 1 \right).
\end{equation}
Substitution of (\ref{3d-transformation}) into (\ref{3d-system}) yields the dynamical system
\begin{equation}
    \label{3d-system-infinity}
\left\{ \begin{array}{l}
\dot{y} = -2y^2, \\
\dot{\psi} = \phi, \\
\dot{\phi} = \psi - y \left( d \phi + \lambda \psi + \psi^3\right),
\end{array} \right.
\end{equation}
where the dot denotes the derivative in $\tau$.

The only equilibrium point of system (\ref{3d-system-infinity}) is $(y,\psi,\phi) = (0,0,0)$.
Linearization of system (\ref{3d-system-infinity}) at $(0,0,0)$ yields eigenvalues $\{-1,0,1\}$,
which implies that the orbits approaching $(0,0,0)$ as $\tau \to \infty$ belongs
to the two-dimensional stable-center manifold. Moreover, since
$x(t) = e^{2t}$, the transformation (\ref{3d-transformation}) suggests that
\begin{equation}
    \label{y-explicit}
y(\tau) = \frac{1}{1 + 2\tau} = \frac{1}{2\tau} + Y(\tau),
\end{equation}
where $Y \in L^1(\tau_0,\infty)$ for any $\tau_0 > 0$. Since $y(0) = 1$ is uniquely determined,
we are only looking for a unique orbit on the two-dimensional stable-center manifold that
approaches $(0,0,0)$ as $\tau \to \infty$. By the theorem on invariant manifolds,
this manifold is tangential to the stable-center manifold of the linearized system.
Therefore, we study the analytical representation of solutions to the linearized system.
With the help of (\ref{y-explicit}), the linearized system is written in the form:
\begin{equation}
\label{eq:3d-system-infinity-linearized}
\frac{d}{d\tau}
    \begin{pmatrix} \psi \\  \phi \end{pmatrix}
= \left[A + V\left(\tau\right) + R\left(\tau\right)\right]
\begin{pmatrix}    \psi \\    \phi \end{pmatrix},
\end{equation}
where
\begin{equation*}
    A =
    \begin{pmatrix}
        0 & 1 \\
        1 & 0
    \end{pmatrix}, \quad
    V\left(\tau\right)= -\frac{1}{2\tau}
    \begin{pmatrix}
        0 & 0 \\
        \lambda & d
    \end{pmatrix}, \quad
    R\left(\tau\right)= - Y(\tau)
    \begin{pmatrix}
        0 & 0 \\
        \lambda & d
    \end{pmatrix}.
\end{equation*}
The eigenvalues of $A$ are $\mu_{\pm}=\pm 1$ with the eigenvectors $(1, \pm 1)$. Solving the characteristic equation
for $A + V(\tau)$, we obtain the eigenvalues of $A + V(\tau)$ denoted by $\nu_{\pm}(\tau)$ in the form:
\begin{equation}
\label{eigenvales-A-V}
\nu_{\pm}(\tau) = -\frac{d}{4 \tau} \pm \sqrt{1 - \frac{\lambda}{2\tau} + \frac{d^2}{16\tau^2}}
= \mu_{\pm} - \frac{d \pm \lambda}{4 \tau} + \nu^{(R)}_{\pm}(\tau),
\end{equation}
where $\nu_{\pm}^{(R)} \in L^1(\tau_0,\infty)$ for any $\tau_0 > 0$. By Theorem 8.1 on p.92 in \cite{CL},
for which the assumptions $V', R \in L^1\left(\tau_0,\infty\right)$ are satisfied, there exist
two linearly independent classical solutions $(\psi_{\pm},\phi_{\pm})$ of the linearized system
\eqref{eq:3d-system-infinity-linearized} satisfying the limit
\begin{equation}
    \lim_{\tau \to \infty}
    \begin{pmatrix}
        \psi_{\pm} \\
        \phi_{\pm}
\end{pmatrix}
e^{-\int_{\tau_0}^\tau \nu_{\pm}\left(\tau'\right)d\tau'} =
\begin{pmatrix}1 \\ \pm 1\end{pmatrix}.
\end{equation}
Thanks to the leading order of the eigenvalues in (\ref{eigenvales-A-V}),
the upper sign corresponds to the unstable solution and the lower sign corresponds
to the stable solution. Since we are looking for the stable solution,
we adopt the decomposition of $(\psi,\phi)$ over the eigenvectors of $A$ together with
the time-dependent factor which follows from the integration
\begin{equation*}
    e^{\int_{\tau_0}^\tau \nu_-\left(\tau'\right)d\tau'} = C(\tau_0) \tau^{\frac{\lambda-d}{4}} e^{-\tau} \left[ 1 + \mathcal{O}(\tau^{-1}) \right] \quad \mbox{\rm as} \quad \tau \to \infty,
\end{equation*}
where the positive constant $C(\tau_0)$ depends on $\tau_0$. Hence we write
\begin{equation}
    \label{3d-asymptotics}
    \psi(\tau) = \tau^{\frac{\lambda - d}{4}} e^{-\tau} \left[ \psi_+(\tau) + \psi_-(\tau) \right], \quad
    \phi(\tau) = \tau^{\frac{\lambda - d}{4}} e^{-\tau} \left[ \psi_+(\tau) - \psi_-(\tau) \right],
\end{equation}
where $(\psi_+,\psi_-)$ are new variables satisfying the following system of equations:
\begin{equation}
    \label{F-plus-minus}
\left\{ \begin{array}{l}
\dot{\psi}_+ = 2 \psi_+ -
 \lambda(2 \tau)^{-1} \psi_+
+ (d-\lambda) (4 \tau)^{-1} \psi_- - H(\psi_+,\psi_-,\tau), \\
\dot{\psi}_- =
(d + \lambda) (4 \tau)^{-1} \psi_+ + H,
\end{array} \right.
\end{equation}
where
\[
H(\psi_+,\psi_-,\tau) := \frac{1}{2} Y(\tau) \left[ (\lambda + d) \psi_+ + (\lambda - d) \psi_- \right]
+ \frac{1}{2} y(\tau) \tau^{\frac{\lambda-d}{2}} e^{-2\tau} \left(\psi_++\psi_-\right)^3.
\]
If $\psi_+, \psi_- \in L^{\infty}(\tau_0,\infty)$ for $\tau_0 > 0$, then
$H(\psi_+(\cdot),\psi_-(\cdot),\cdot) \in L^1(\tau_0,\infty)$ due to $Y \in L^1(\tau_0,\infty)$. This suggests that
the remainder terms in the $H$-function remain small along the solution
satisfying $(\psi,\phi) \to (0,0)$ as $\tau \to \infty$. In order to make this analysis precise,
we integrate the first equation of system (\ref{F-plus-minus}) subject to the boundary condition
$\lim\limits_{\tau \to \infty} e^{-2\tau} \psi_+(\tau) = 0$ and obtain the integral equation:
\begin{equation}
    \label{volterra-plus}
\psi_+(\tau) = \int_{\tau}^{\infty} e^{-2(\tau'-\tau)}
\left[ \frac{\lambda}{2 \tau'} \psi_+(\tau') + \frac{\lambda - d}{4 \tau'} \psi_-(\tau') +
H(\psi_+(\tau'),\psi_-(\tau'),\tau') \right] d \tau'.
\end{equation}
On the other hand, integrating the second equation of system (\ref{F-plus-minus})
subject to the boundary condition $\lim\limits_{\tau \to \infty} \psi_-(\tau) = c$ for
an arbitrary constant $c \in \mathbb{R}$ yields another integral equation:
\begin{equation}
    \label{volterra-minus}
\psi_-(\tau) = c - \int_{\tau}^{\infty} \left[
\frac{\lambda + d}{4 \tau'} \psi_+(\tau') + H(\psi_+(\tau'),\psi_-(\tau'),\tau') \right] d \tau'.
\end{equation}
It is clear from the integral equation (\ref{volterra-minus}) 
that $\psi_+ \in L^{\infty}(\tau_0,\infty)$ is not sufficient for
$\psi_- \in L^{\infty}(\tau_0,\infty)$. Therefore, we consider the Banach space $L^1(\tau_0,\infty) \cap L^{\infty}(\tau_0,\infty)$ for $\tau^{-1} \psi_+(\tau)$
and $L^{\infty}(\tau_0,\infty)$ for $\psi_-(\tau)$, where  $1 \ll \tau_0 < \infty$.
This suggest that one can obtain $\tilde{\psi}_+(\tau) := \tau^{-1} \psi_+(\tau)$ and $\psi_-(\tau)$ 
from solutions to the system of fixed-point equations:
\begin{equation}
\label{sys-fixed-points}
\tilde{\psi}_+ = A_+(\tilde{\psi}_+,\psi_-), \qquad 
\psi_- = A_-(\tilde{\psi}_+,\psi_-),
\end{equation}
where
$$
A_+(\tilde{\psi}_+,\psi_-)(\tau) := \frac{1}{\tau}
\int_{\tau}^{\infty} e^{-2(\tau'-\tau)}
\left[ \frac{\lambda}{2} \tilde{\psi}_+ + \frac{\lambda - d}{4 \tau'} \psi_- +  H(\tilde{\psi}_+,\psi_-,\tau') \right] d \tau'
$$
and
$$
A_-(\tilde{\psi}_+, \psi_-)(\tau) := c - \int_{\tau}^{\infty} \left[
\frac{\lambda + d}{4} \tilde{\psi}_+ + H(\tilde{\psi}_+,\psi_-,\tau') \right] d \tau',
$$
with $H(\tilde{\psi}_+,\psi_-,\tau)$ being redefined in new variables by
$$
H = \frac{1}{2} Y(\tau) \left[ (\lambda + d) \tau \tilde{\psi}_+ + (\lambda - d) \psi_- \right]
+ \frac{1}{2} y(\tau) \tau^{\frac{\lambda-d}{2}} e^{-2\tau} \left( \tau \tilde{\psi}_+ + \psi_- \right)^3.
$$

We proceed with fixed-point estimates similarly to the proof of Lemma \ref{lemma-1}. By using the Young inequality
for convolution integrals, we estimate the first and third term in $A_+$ as follows:
\begin{eqnarray*}
& \| \frac{1}{\tau}
\int_{\tau}^{\infty} e^{-2(\tau'-\tau)}
 [ \frac{\lambda}{2} \tilde{\psi}_+ +  H(\tilde{\psi}_+,\psi_-,\tau') ] d \tau' \|_{L^1 \cap L^{\infty}} \\
& \leq
\| \tau^{-1} \|_{L^{\infty}} \| e^{-2\tau} \|_{L^1(0,\infty)}
\left[ \frac{|\lambda|}{2} \| \tilde{\psi}_+ \|_{L^1 \cap L^{\infty}} + \| H(\tilde{\psi}_+,\psi_-,\cdot) \|_{L^1 \cap L^{\infty}} \right],
\end{eqnarray*}
where all norms are defined on $(\tau_0,\infty)$ with $\tau_0 \gg 1$ except for $\| e^{-2\tau} \|_{L^1(0,\infty)} = \frac{1}{2}$. In addition, we estimate
\begin{eqnarray*}
\| H(\tilde{\psi}_+,\psi_-,\cdot) \|_{L^1 \cap L^{\infty}} & \leq &
\frac{|\lambda| + d}{2} \left( \| \tau Y(\tau) \|_{L^{\infty}} \| \tilde{\psi}_+ \|_{L^1 \cap L^{\infty}}
+ \| Y \|_{L^1 \cap L^{\infty}} \| \psi_- \|_{L^{\infty}} \right) \\
&& + \frac{1}{2} \| y(\tau) \tau^{\frac{\lambda-d}{2}} e^{-2\tau} (\tau \tilde{\psi}_+ + \psi_-)^3 \|_{L^1 \cap L^{\infty}},
\end{eqnarray*}
with $Y \in L^1 \cap L^{\infty}$, $\tau Y \in L^{\infty}$ from (\ref{y-explicit})
and $y(\tau) \tau^{\frac{\lambda-d}{2}} e^{-2\tau}$ being exponentially small on $(\tau_0,\infty)$ with $\tau_0 \gg 1$.
For the second term in $A_+$, we use both the Young and Cauchy--Schwarz inequalities in order to obtain:
\begin{eqnarray*}
\| \frac{1}{\tau}
\int_{\tau}^{\infty} e^{-2(\tau'-\tau)}
\frac{(\lambda - d)}{4 \tau'} \psi_-  d \tau' \|_{L^1 \cap L^{\infty}} \leq \frac{|\lambda| + d}{4} \| \tau^{-1} \|_{L^2 \cap L^{\infty}}
\| e^{-2\tau} \|_{L^1(0,\infty)} \| \tau^{-1} \|_{L^2 \cap L^{\infty}} \| \psi_- \|_{L^{\infty}}.
\end{eqnarray*}
Finally, we estimate $A_-$ as follows:
\begin{eqnarray*}
\| A_-(\tilde{\psi}_+, \psi_-) \|_{L^{\infty}} & \leq & |c| +  \frac{|\lambda| + d}{4} \| \tilde{\psi}_+ \|_{L^1} +
\| H(\tilde{\psi}_+,\psi_-,\cdot) \|_{L^1}.
\end{eqnarray*}
If $\tau_0$ is a sufficiently large positive number and if 
\begin{equation}
\label{set-G-F}
\| \tilde{\psi}_+ \|_{L^1 \cap L^{\infty}} + \| \psi_- \|_{L^{\infty}} \leq 2 |c|
\end{equation}
then the previous bounds imply that 
$$
\| A(\tilde{\psi}_+,\psi_-) \|_{L^1 \cap L^{\infty}} + \| A_-(\tilde{\psi}_+,\psi_-) \|_{L^{\infty}} \leq 2 |c|,
$$
due to smallness of $\| \tau Y \|_{L^{\infty}}$, $\| Y \|_{L^1 \cap L^{\infty}}$, $\| \tau^{-1}\|_{L^2 \cap L^{\infty}}$, 
and $\| y(\tau) \tau^{\frac{\lambda-d}{2}} e^{-2\tau} \|_{L^1 \cap L^{\infty}}$ if $\tau_0 \gg 1$.
In addition, by similar estimates, it is easy to prove that $(A_+,A_-)$ is a contraction operator
in the set (\ref{set-G-F}) if $\tau_0 \gg 1$. By Banach's fixed-point theorem, there exists the unique solution for
$\tilde{\psi}_+ \in L^1(\tau_0,\infty) \cap L^{\infty}(\tau_0,\infty)$ and $\psi_- \in L^{\infty}(\tau_0,\infty)$
to the system of integral equations (\ref{sys-fixed-points}) in the set (\ref{set-G-F}). From $\tilde{\psi}_+$, we obtain $\psi_+$ by 
$\psi_+(\tau) = \tau \tilde{\psi}_+(\tau)$. Furthermore, bootstrapping arguments similar to those in the proof of Lemma \ref{lemma-1} gives smoothness of $\psi_+$ and $\psi_-$ on $(\tau_0,\infty)$. Thanks to the integrability of $\tau^{-1} \psi_+$ and continuity of $\psi_+$, we have
$\psi_+(\tau) \to 0$ as $\tau \to \infty$. 

By unfolding the transformations (\ref{transformation-r-t}), (\ref{3d-transformation}), and (\ref{3d-asymptotics}),
we obtain that if $f(r), f'(r) \to 0$, then $f(r)$ satisfies the asymptotic behavior (\ref{asymptotics-infinity}), where
$$
C := 2^{-\frac{\lambda - d}{4}} e^{\frac{1}{2}} \; c
$$
and $c := \lim\limits_{\tau \to \infty} \psi_-(\tau)$ is defined in the integral equation (\ref{volterra-minus}).
\end{proof}

The following lemma guarantees global continuation of classical solutions to the differential equation (\ref{eq}) from $r = 0$ to $r \to \infty$ and 
from $r \to \infty$ to $r = 0$.

\begin{lemma}
	\label{lemma-2}
	For every $d \geq 1$ and $\lambda \in \mathbb{R}$, if $f \in C^2(0,r_0)$ is a solution of Lemma \ref{lemma-1} for some  $r_0 \in (0,\infty)$, then $f \in C^2(0,\infty)$ and if $f \in C^2(r_0,\infty)$ is a solution of Lemma \ref{lemma-0} for some $r_0 \in (0,\infty)$, then $f \in C^2(0,\infty)$.
\end{lemma}

\begin{proof}
	Let us introduce the Lyapunov function in the form:
	\begin{equation}
	\label{Lyapunov-again}
	\Lambda(f,f',r) := \frac{1}{2} (f')^2 + \frac{1}{2} (\lambda - r^2) f^2 + \frac{1}{4} f^4.
	\end{equation}
It follows from (\ref{eq}) and (\ref{Lyapunov-again}) that
	\begin{equation}
	\label{Lyapunov-again-rate}
	\frac{d}{dr} \Lambda(f,f',r) = - \frac{d-1}{r} (f')^2 - r f^2 < 0,
	\end{equation}
	hence the map $r \mapsto \Lambda(f(r),f'(r),r)$ is strictly monotonically decreasing along the classical solution to the differential equation (\ref{eq}).	It follows from (\ref{Lyapunov-again}) that 
	\begin{equation}
	\label{Lyapunov-again-bound}
	\frac{1}{2} (f')^2 + \frac{1}{4} (f^2 + \lambda - r^2)^2 \leq 
	\Lambda(f,f',r) + \frac{1}{4} (\lambda - r^2)^2.
	\end{equation}
	
	Let $f \in C^2(0,r_0)$ be a solution of Lemma \ref{lemma-1} for some $r_0 > 0$ and assume that the solution blows up at a finite $R < \infty$. Since 
the map $r \mapsto \Lambda(f(r),f'(r),r)$ is decreasing, we obtain a contradiction from the bound (\ref{Lyapunov-again-bound}):
\begin{equation*}
	\frac{1}{2} (f')^2 + \frac{1}{4} (f^2 + \lambda - r^2)^2 \leq 
\Lambda(f(r_0),f'(r_0),r_0) + \frac{1}{4} (\lambda - R^2)^2 < \infty, \quad r \in [r_0,R].
\end{equation*}	
Hence, no finite $R$ exists and the classical solution continues on $(0,\infty)$.
	
Let $f \in C^2(r_0,\infty)$ be a solution of Lemma \ref{lemma-0} for some $r_0 > 0$. It follows from the fast decay of $f(r),f'(r) \to 0$ as $r \to \infty$ that 
$\Lambda(f(r),f'(r),r) \to 0$ as $r \to \infty$. It follows from (\ref{Lyapunov-again-rate}) and (\ref{Lyapunov-again-bound}) that there exist positive constants $A_0$ and $B_0$ such that 
	$$
	r \frac{d}{dr} \Lambda(f,f',r) \geq -A_0 \Lambda(f,f',r) - B_0, \quad 
	r \in (0,r_0],
	$$
	or equivalently, 
	$$
	r \frac{d}{dr} r^{A_0} \Lambda(f,f',r) \geq - B_0 r^{A_0}, \quad 
	r \in (0,r_0].
	$$
	Integration on $[r,r_0]$ yields
	$$
	\Lambda(f(r),f'(r),r) \leq \left(\frac{r_0}{r} \right)^{A_0} \left[ \Lambda(f(r_0),f'(r_0),r_0) + \frac{B_0}{A_0} \right] < \infty, \quad r \in (0,r_0].
	$$
	It follows from the bound (\ref{Lyapunov-again-bound}) that the classical solution continues on $(0,\infty)$.
\end{proof}

\section{Proof of Theorem \ref{theorem-1}}
\label{sec-proof-1}

Here we develop the shooting method for the proof of Theorem \ref{theorem-1}.

The unique global solution $f \in C^2(0,\infty)$ of the initial value problem (\ref{statGPshoot}) is given by Lemmas \ref{lemma-1} and \ref{lemma-2}. We define the following three sets:
\begin{equation}
    \label{I-plus}
    I_+ := \left\{ \lambda \in \mathbb{R} : \; \exists r_0 \in (0,\infty) : \; f(r_0) = 0, \; \mbox{\rm while} \; f(r) > 0, \;\; f'(r) < 0, \;\; r \in (0,r_0) \right\},
\end{equation}
\begin{equation}
    \label{I-minus}
    I_- := \left\{ \lambda \in \mathbb{R} : \; \exists r_0 \in (0,\infty) : \; f'(r_0) = 0, \; \mbox{\rm while} \; f(r) > 0, \;\; f'(r) < 0, \;\; r \in (0,r_0) \right\},
\end{equation}
and
\begin{equation}
    \label{I-zero}
    I_0 := \left\{ \lambda \in \mathbb{R} : \; f(r) > 0, \;\; f'(r) < 0, \;\; r \in (0,\infty) \right\}.
\end{equation}
The sets $I_+$, $I_-$, and $I_0$ depend on parameters $b$ and $d$, which are not written. We make the following partition of $\mathbb{R}$ for parameter $\lambda$:
\begin{equation}
    \label{partition}
    \mathbb{R}  = I_+ \cup I_0 \cup I_-.
\end{equation}
By uniqueness of solutions to differential equations, if $f(r_0) = f'(r_0) = 0$ for some $r_0 \in (0,\infty)$, then $f(r) = 0$ for every $r \in (0,\infty)$, hence $I_+ \cap I_- = \emptyset$. By construction, it is also true that $I_+ \cap I_0 = \emptyset$ and $I_- \cap I_0 = \emptyset$, hence the three sets are disjoint. 

In the following two lemmas, we prove that the sets $I_+$ and $I_-$ are open and non-empty. These results imply that $I_0$ in the partition (\ref{partition}) is closed and non-empty.

\begin{lemma}
\label{lemma-plus}
For every $d \geq 1$, $I_+$ is open and, moreover, $[d,\infty) \subset I_+$.
\end{lemma}

\begin{proof}
The unique solution $f \in C^2(0,\infty)$ depends smoothly on the parameter $\lambda$ since the differential equation (\ref{eq}) is smooth in $f$ and $\lambda$. Let $f_{\lambda}$ denotes the unique $\lambda$-dependent solution 
and $r_0$ be a root of $f_{\lambda_0}$ for a fixed $\lambda_0 \in I_+$. By uniqueness of the zero
solution, if $f_{\lambda_0}(r_0) = 0$, then $f_{\lambda_0}'(r_0) \neq 0$. Since $f_{\lambda}$ is smooth in $\lambda$,
it follows from the implicit function theorem that for every $\lambda$ in an open neighborhood of $\lambda_0$
there exists $r_{\lambda}$ near $r_0$ such that $f_{\lambda}(r_{\lambda}) = 0$. Hence, the set $I_+$ is open.
It remains to prove that such $r_{\lambda} \in (0,\infty)$ exists for every $\lambda \in [d,\infty)$.

Let  $g(r)=e^{\frac{1}{2} r^2} f(r)$. Then, $g(r)$ satisfies the differential equation:
\begin{equation}\label{eqg}
  g''(r) +\left[\frac{d-1}{r}-2r\right] g'(r) + e^{-r^2} g(r)^3 + (\lambda - d) g(r) = 0,
\end{equation}
subject to the initial conditions $g(0)= b$ and $g'(0) = 0$.
By using the transformation (\ref{transformation-r-t}) and the asymptotic expansion (\ref{asymptotic-psi-prime}), we obtain with the chain rule
\begin{eqnarray*}
e^{-\frac{1}{2} r^2} g'(r) & = & f'(r) + r f(r) \\
& = & e^{-t} \psi'(t) + e^t \psi(t) \\
& = & \frac{b}{d} (d-\lambda - b^2) e^t + \mathcal{O}(e^{3t}) \quad \mbox{\rm as} \quad t \to -\infty.
\end{eqnarray*}
Since $\lambda \geq d$, we have $g'(r) < 0$ for some small $r > 0$.

Let $r_0 := \inf \{ r > 0 : \; g(r) = 0 \}$. We need to show that $r_0 < \infty$. First, we show that $g'(r) < 0$ for all
$r \in (0,r_0)$. Indeed, if there exists $r_1 \in (0,r_0)$ such that $g'(r_1) = 0$ and $g'(r) < 0$ for $r \in (0,r_1)$, then the differential equation (\ref{eqg})
with $\lambda \geq d$ implies that $g''(r_1) < 0$, which is impossible.
Hence, $g'(r) < 0$ for all $r \in (0,r_0)$.

It follows from (\ref{eqg}) that
$$
g''(r) \leq \left[2 r - \frac{d-1}{r}\right] g'(r), \quad r \in (0,r_0).
$$
If $r_0 \leq R := \frac{\sqrt{d-1}}{\sqrt{2}}$, we are done. Assume that $r_0 > R$ and define $G(r) := -g'(r)$. Then,
$$
G'(r) \geq \left[2 r - \frac{d-1}{r}\right] G(r), \quad r \in (R,r_0).
$$
Since $G(r) > 0$ for $r \in [R,r_0)$, we have $G(r) \geq G(R)$ for $r \in [R,r_0)$, or alternatively, $g'(r) \leq g'(R) < 0$. The case $r_0 = \infty$ is impossible since $g(r)$ must hit zero for a finite $r$. Thus, $r_0 < \infty$ for every $\lambda \in [d,\infty)$.
\end{proof}

\begin{lemma}
\label{lemma-minus}
For every $d \geq 4$, $I_-$ is open and, moreover, $(-\infty,0] \subset I_-$.
\end{lemma}

\begin{proof}
In order to prove that $I_-$ is open, we extend the proof of Lemma \ref{lemma-plus} based on the implicit function
theorem. Let $f_{\lambda}$ denote the $\lambda$-dependent unique solution and $r_0$ be a root of $f_{\lambda_0}'$ for a fixed $\lambda_0 \in I_-$. Then, the differential
equation (\ref{eq}) implies that either $f_{\lambda_0}''(r_0) \neq 0$ 
or $f_{\lambda_0}''(r_0) = 0$ and $r_0^2 = \lambda_0 + f_{\lambda_0}(r_0)^2$.
In the latter case, since $f$ is smooth,
the derivative of the differential equation (\ref{eq}) at the point $r_0$ for which $f_{\lambda_0}'(r_0) = 0$
and $f_{\lambda_0}''(r_0) = 0$ gives $f_{\lambda_0}'''(r_0) = 2 r_0 f_{\lambda_0}(r_0) > 0$,
which is impossible if $f_{\lambda_0}'(r) < 0$ for $r \in (0,r_0)$. 
This implies that if $f'_{\lambda_0}(r_0) = 0$, then $f_{\lambda_0}''(r_0) \neq 0$. Since $f_{\lambda}$ is smooth in $\lambda$,
it follows from the implicit function theorem that for every $\lambda$ in an open neighborhood of $\lambda_0$
there exists $r_{\lambda}$ near $r_0$ such that $f_{\lambda}'(r_{\lambda}) = 0$. Hence, the set $I_-$ is open.
It remains to prove that such $r_{\lambda} \in (0,\infty)$ exists for every $\lambda \in (-\infty,0]$.

First, we show that $f(r) > 0$ for $r > 0$ if $\lambda \leq 0$. It follows from (\ref{eq}) that
$$
r^d f'(r) f''(r) + (d-1) r^{d-1} [f'(r)]^2 - r^{d+2} f(r) f'(r) + r^d f(r)^3 f'(r) + \lambda r^d f(r) f'(r) = 0.
$$
Assuming $f(R) = 0$ for some $R \in (0,\infty)$ and integrating on $[0,R]$ yields
\begin{eqnarray*}
\frac{1}{2} R^d [f'(R)]^2 + \frac{d-2}{2} \int_0^{R} r^{d-1} [f'(r)]^2 dr + \frac{d+2}{2} \int_0^{R} r^{d+1}  f(r)^2 dr \\
- \frac{d}{4} \int_0^{R} r^{d-1} f(r)^4 dr - \frac{d \lambda}{2} \int_0^{R} r^{d-1} f(r)^2 dr = 0.
\end{eqnarray*}
Similarly, integrating equation
$$
r^{d-1} f(r) f''(r) + (d-1) r^{d-2} f(r) f'(r) - r^{d+1} f(r)^2 + r^{d-1} f(r)^4 + \lambda r^{d-1} f(r)^2 = 0
$$
on $[0,R]$ with $f(R) = 0$ yields
$$
-\int_0^{R} r^{d-1} [f'(r)]^2 dr - \int_0^{R} r^{d+1}  f(r)^2 dr + \int_0^{R} r^{d-1} f(r)^4 dr + \lambda \int_0^{R} r^{d-1} f(r)^2 dr = 0.
$$
Eliminating $\int_0^{R} r^{d-1} [f'(r)]^2 dr$ from these two equations yields the constraint:
$$
\frac{1}{2} R^d [f'(R)]^2 +
2 \int_0^{R} r^{d+1}  f(r)^2 dr + \frac{d-4}{4} \int_0^{R} r^{d-1} f(r)^4 dr - \lambda \int_0^{R} r^{d-1} f(r)^2 dr = 0.
$$
If $d \geq 4$ and $\lambda \leq 0$, this constraint is never satisfied, hence no $R \in (0,\infty)$ exists and $f(r) > 0$ for every $r > 0$. Moreover, if $f \in C^2(0,\infty)$ and $f(r), f'(r) \to 0$ as $r \to \infty$,
then the fast asymptotic decay (\ref{asymptotics-infinity}) in Lemma \ref{lemma-0} implies
that $f \in L^2_r(\mathbb{R}^+)$, which is impossible if $\lambda \in (-\infty,0]$ by Proposition \ref{prop-2}.
Hence, there exists a constant $c > 0$ such that $f(r) \geq c$ for $r > 0$.

Next, we show that there exists $r_0 \in (0,\infty)$ such that $f'(r_0) = 0$. To do so, we integrate the differential equation
\begin{equation}
    \label{eq-self-adjoint-form}
\frac{d}{dr} \left[ r^{d-1} f'(r) \right] = r^{d+1} f(r) - r^{d-1} f(r)^3 - \lambda r^{d-1} f(r)
\end{equation}
on $[0,R]$ and obtain the estimate:
\begin{eqnarray*}
R^{d-1} f'(R) & = & \int_0^R r^{d+1} f(r) dr - \int_0^R r^{d-1} f(r)^3 dr - \lambda \int_0^R r^{d-1} f(r) dr \\
& \geq & \frac{c}{d+2} R^{d+2} - \frac{b^3}{d} R^d,
\end{eqnarray*}
where we have used that $\lambda \leq 0$ and $c \leq f(r) \leq b$ as long as $f'(r) < 0$. Hence for
$$
R > \left(\frac{b^3(d+2)}{d c} \right)^{1/2},
$$
we must have $f'(R) > 0$ so that there exists $r_0 \in (0,\infty)$ such that $f'(r_0) = 0$ if $\lambda \in (-\infty,0]$.
\end{proof}

It follows from Lemmas \ref{lemma-plus} and \ref{lemma-minus} that $I_0$ is closed and non-empty.
The following lemma states that the set $I_0$ in the partition (\ref{partition}) contains all values of $\lambda$ for which the unique solution $f$
to the initial-value problem (\ref{statGPshoot}) is a solution $\mathfrak{u} \in \mathcal{E}$ to the boundary-value problem (\ref{statGPrad}). 

\begin{lemma}
\label{lemma-decay} If $\lambda \in I_0$, then $f(r), f'(r) \to 0$ as $r \to \infty$ and $f \in \mathcal{E} \subset L^2_r(\mathbb{R}^+)$.
\end{lemma}

\begin{proof}
If $f \in C^2(0,\infty)$ satisfies $f(r) > 0$ and $f'(r) < 0$ for $r \in (0,\infty)$, then necessarily $f'(r) \to 0$ as $r \to \infty$ because $[0,b] \ni f$ is compact. Assume that $f(r) \to c$ as $r \to \infty$ with some $c \in (0,b)$. Then, integrating (\ref{eq-self-adjoint-form}) on $[0,R]$ similarly to the proof of Lemma \ref{lemma-minus} yields
\begin{eqnarray*}
R^{d-1} f'(R) & = & \int_0^R r^{d+1} f(r) dr - \int_0^R r^{d-1} f(r)^3 dr - \lambda \int_0^R r^{d-1} f(r) dr \\
& \geq & \frac{c}{d+2} R^{d+2} - \frac{b(b^2 + \lambda)}{d} R^d,
\end{eqnarray*}
where $\lambda \in (0,d)$ if $\lambda \in I_0$. Hence for
$$
R > \left(\frac{b (b^2 + \lambda) (d+2)}{d c} \right)^{1/2},
$$
we must have $f'(R) > 0$ which is a contradiction. This implies that $c = 0$, that is, $f(r) \to 0$ as $r \to \infty$.
Since $f(r), f'(r) \to 0$ as $r \to \infty$, Lemma \ref{lemma-0} implies that  $f(r)$ satisfies the fast asymptotic decay (\ref{asymptotics-infinity}) so that $f \in \mathcal{E} \subset L^2_r(\mathbb{R}^+)$ for this
$\lambda \in I_0$.
\end{proof}

We collect all individual results together as the proof of Theorem \ref{theorem-1}. \\

{\em Proof of Theorem \ref{theorem-1}.} Fix $d \geq 4$ so that all previous results can be equally applied.

By Lemmas \ref{lemma-1} and \ref{lemma-2}, there exists the unique global classical solution $f \in C^2(0,\infty)$ to the initial-value problem (\ref{statGPshoot}) for $\lambda \in \mathbb{R}$. The line $\mathbb{R}$ for the parameter $\lambda$ in the differential equation (\ref{eq}) can be partitioned into the union of three disjoint sets $I_+$, $I_-$, and $I_0$ given by (\ref{I-plus}), (\ref{I-minus}), and (\ref{I-zero}) respectively. Suitable solutions to the boundary-value problem (\ref{statGPrad})
in the function space $\mathcal{E} \subset L^2_r(\mathbb{R}^+)$ may only exist for $\lambda \in I_0$.

By Lemmas \ref{lemma-plus} and \ref{lemma-minus}, the sets $I_+$ and $I_-$ are open and non-empty, so that the set $I_0$ in the partition (\ref{partition}) is closed and non-empty. By Lemma \ref{lemma-decay}, we proved that if $\lambda \in I_0$, then the corresponding function $f \in C^2(0,\infty)$ is a solution $\mathfrak{u} \in \mathcal{E}$ to the boundary-value problem (\ref{statGPrad}). 
It follows by Propositions \ref{prop-1} and \ref{prop-2} that $I_0 \subset (d-4,d)$. \hspace{5.8cm} $\Box$

\vspace{0.25cm}

Figure \ref{fig:fsol} illustrates the shooting method used in the proof of Theorem \ref{theorem-1}. For $d = 5$ and $b = 10$, we compute numerically the unique classical solution to the initial-value problem (\ref{statGPshoot}) for three different values of $\lambda$. For a special value of $\lambda$ denoted as $\lambda(b)$, the solution gives the ground state of the boundary-value problem (\ref{statGPrad}), which implies that $\lambda(b) \in I_0$. For another value of $\lambda < \lambda(b)$ the solution does not cross the zero level but grows with some oscillations as $r \to \infty$. Therefore, there is $r_0 \in (0,\infty)$ such that $f'(r_0) = 0$ 
and this $\lambda \in I_-$. For yet another value of $\lambda > \lambda(b)$, the solution crosses the zero level (and becomes large negative with some oscillations) so that there is $r_0 \in (0,\infty)$ such that $f(r_0) = 0$ and this $\lambda \in I_+$. We have confirmed numerically 
that the value of $\lambda(b) \in I_0$ is unique for every $b > 0$ 
as is stated in Remark \ref{rem-unique-3}. 

\begin{figure}[htp!]
	\centering
	\includegraphics[width=0.45\textwidth]{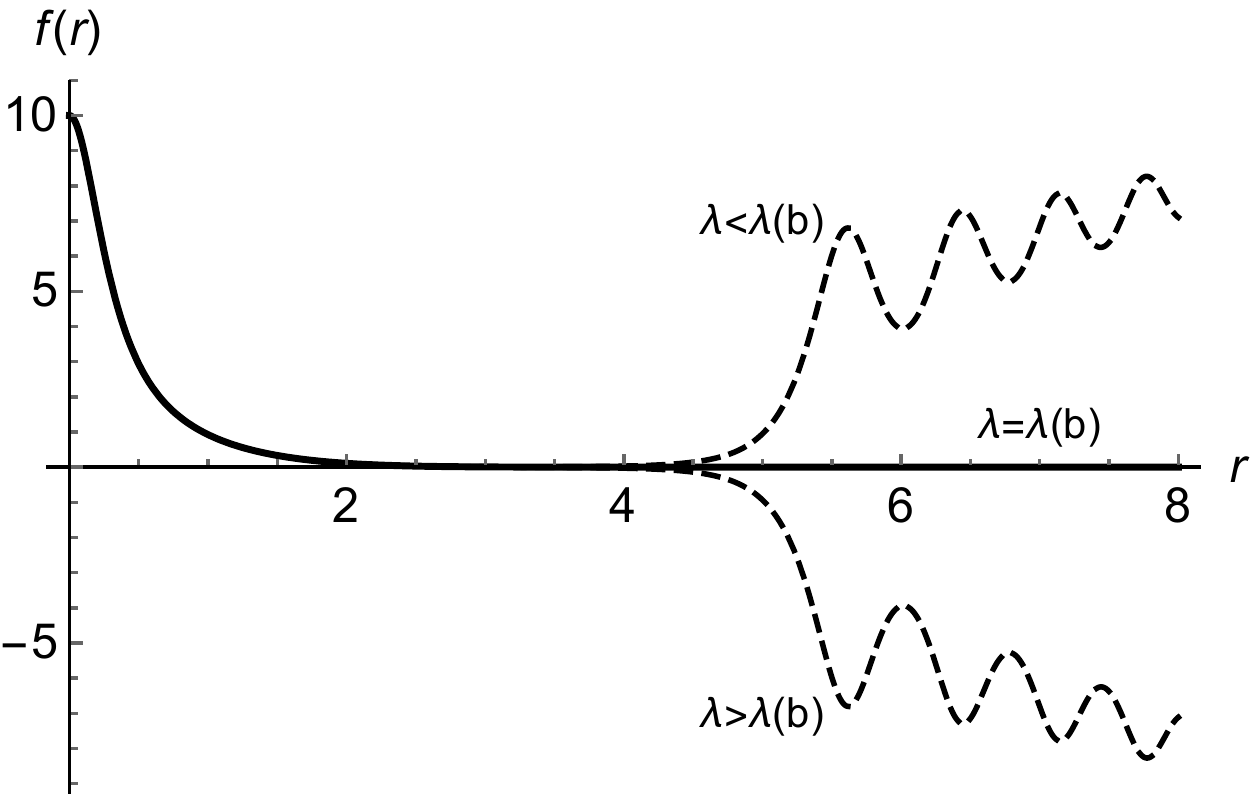}
	\caption{Plot of the unique solution $f$ satisfying the initial-value problem (\ref{statGPshoot}) for $d=5$, $b=10$, and three values of $\lambda$. For $\lambda = \lambda(b)$, the solution $f$ satisfies the boundary-value problem (\ref{statGPrad}).}
	\label{fig:fsol}
\end{figure}

\section{Proof of Theorem \ref{theorem-Selem}}
\label{sec-proof-2}

Here we explain how the shooting method can be applied to the proof of Theorem \ref{theorem-Selem}. Our arguments basically reproduce the approach in \cite{Selem2013} with some important modifications. 

By using the transformation 
\begin{equation}
\label{transformaton-singular}
r = e^t, \quad F(r) = \Psi(t) \quad F'(r) = e^{-t} \Psi'(t),
\end{equation}
one can obtain solutions to the initial-value problem (\ref{statGPsingular}) from the second-order differential equation
\begin{equation}
    \label{eq-psi-singular}
    \Psi''(t) + (d-4) \Psi'(t) + (3-d) \Psi(t) + \Psi(t)^3
    = -\lambda e^{2t} \Psi(t) + e^{4t} \Psi(t), \quad t \in \mathbb{R}
\end{equation}
completed with the boundary conditions
\begin{equation}
\label{bc-singular}
\left\{ \begin{array}{l} \Psi(t) \to \sqrt{d-3}, \\ 
\Psi'(t) \to 0, \end{array} \right. \quad \mbox{\rm as } \;\; t \to -\infty.
\end{equation}
We denote the unique solution of the second-order equation (\ref{eq-psi-singular}) satisfying the boundary conditions (\ref{bc-singular}) by $\Psi_{\lambda}(t)$. One can prove by a simple extension of Lemma \ref{lemma-1} that this solution satisfies the asymptotic behavior:
\begin{equation}
\label{lambda-solution-behavior}
\Psi_{\lambda}(t) = \sqrt{d-3} \left[1 - \frac{\lambda}{4d - 10} e^{2t} + \mathcal{O}(e^{4t})  \right] \quad \mbox{\rm as} \;\; t \to -\infty.
\end{equation}
The solution of Theorem \ref{theorem-Selem} arises for $\lambda = \lambda_{\infty}$, 
for which $\Psi_{\infty} := \Psi_{\lambda = \lambda_{\infty}}$ decays to zero as $t \to +\infty$. The following list contains relevant details how the shooting method is modified for the proof of Theorem \ref{theorem-Selem}.

\vspace{0.25cm}

\begin{itemize}
    \item $(\sqrt{d-3},0,0)$ is an equilibrium point of the three-dimensional dynamical system 
  \begin{equation}
  \label{3d-system-capital}
  \left\{ \begin{array}{l}
  x' = 2x, \\
  \Psi' = \Phi, \\
  \Phi' = (4-d) \Phi + (d-3) \Psi - \Psi^3 - \lambda x \Psi + x^2 \Psi,
  \end{array} \right.
  \end{equation}
  where $x(t) := e^{2t}$ and the prime stands for the derivative in $t$.  
  If $d \geq 5$, the equilibrium point $(\sqrt{d-3},0,0)$ admits a one-dimensional unstable manifold and a two-dimensional stable manifold. 
  The unique local classical solution $\Psi_{\lambda}$ satisfying the differential equation (\ref{eq-psi-singular}) and the boundary conditions (\ref{bc-singular}) corresponds to the one-dimensional unstable manifold of the dynamical system (\ref{3d-system-capital}) with uniquely defined $x(t) = e^{2t}$. The existence and uniqueness of $\Psi_{\lambda}$ follows by the unstable manifold theorem. In an analogue with Lemma \ref{lemma-1}, this gives the unique solution $F \in C^2(0,r_0)$ with $F(r) > 0$ and $F'(r) < 0$ for $r \in (0,r_0)$ to the initial-value problem (\ref{statGPsingular}) for $d \geq 5$.

\vspace{0.25cm}

    \item The proof of Lemma \ref{lemma-0} does not depend on the behavior of $f(r)$ near $r = 0$ as long as $f(r), f'(r) \to 0$ as $r \to \infty$. By the transformation $F(r) = r f(r)$, if $F(r), F'(r) \to 0$ as $r \to \infty$, then $f(r), f'(r) \to 0$ as $r \to \infty$. By Lemma \ref{lemma-0}, there exists $C \in \mathbb{R}$ such that
    \begin{equation}
    \label{asymptotics-infinity-F}
F(r) \sim C r^{\frac{\lambda - d + 2}{2}} e^{-\frac{1}{2}r^2} \quad \mbox{\rm as} \quad r \to \infty.
\end{equation}

\vspace{0.25cm}

\item The proof of Lemma \ref{lemma-2} is extended to $f(r) = r^{-1} F(r)$ verbatim.

\vspace{0.25cm}

    \item For the set $I_+$ in (\ref{partition}) defined by zeros of $F$, 
    openness of $I_+$ follows from uniqueness of the zero solutions in (\ref{eq-psi-singular}) which implies that if $F(r_0) = 0$, then $F'(r_0) \neq 0$. In order to show that $[d,\infty) \subset I_+$, we define 
    $$
    e^{-\frac{1}{2}r^2} g(r) = f(r) = r^{-1} F(r)
    $$
    and 
    $$
    e^{-\frac{1}{2} r^2} g'(r) = r^{-1} F'(r) - \frac{1-r^2}{r^2} F(r),
    $$
    hence $g'(r) < 0$ for small $r > 0$. The rest of the proof of Lemma \ref{lemma-plus} applies verbatim.
    
    \vspace{0.25cm}
    
      \item For the set $I_-$ in (\ref{partition}) defined by zeros of $F'$, 
      a special care should be taken to prove that the set is open. In the special case when $F_{\lambda_0}'(r_0) = F_{\lambda_0}''(r_0) = 0$, for which 
      $d-3 - F_{\lambda_0}(r_0)^2 = (\lambda_0 - r_0^2) r_0^2 > 0$, we obtain by differentiation in $r$:
      $$
      F_{\lambda_0}'''(r_0) = 2 (2r_0 - \lambda_0 r_0^{-1}) F_{\lambda_0}(r_0),
      $$
      hence the contradiction with $F_{\lambda_0}'''(r_0) > 0$ only holds if $\lambda_0 \in (r_0^2, 2 r_0^2)$. If $\lambda_0 = 2 r_0^2$ so that $F_{\lambda_0}'''(r_0) = 0$, then we obtain by another differentiation in $r$:
      $$
      F_{\lambda_0}''''(r_0) = (6 \lambda_0 r_0^{-2} - 4) F_{\lambda_0}(r_0) = 8 F_{\lambda_0}(r_0) > 0,
      $$
      so that the minimum of $F_{\lambda}$ persists near $r_0$ when the solution is continued with respect to $\lambda$ near $\lambda_0$. If $\lambda_0 > 2 r_0^2$ and $F_{\lambda_0}'''(r_0) < 0$, 
      then $F_{\lambda_0}'(r) \leq 0$ near $r = r_0$, so that if no other extremal points exist, then $F_{\lambda_0}(r) \in [0,\sqrt{d-3}]$ and $F_{\lambda_0}'(r) \leq 0$ for all $r > 0$. However, $F_{\lambda_0}(r) \to c$ as $r \to \infty$ is impossible for $c \neq 0$ (see the next item), hence $F_{\lambda_0}(r) \to 0$ as $r \to \infty$. However, if $F_{\lambda_0} \in C^2(0,\infty)$, $F_{\lambda_0}(r) > 0$ for $r > 0$, and $F_{\lambda_0}(r) \to 0$ as $r \to \infty$, then $F_{\lambda_0}'(r) < 0$ for $r > 0$ by the arguments from \cite{LiNi}, which is a contradiction with $F_{\lambda_0}'(r_0) = 0$. Thus, either $F_{\lambda_0}'(r_0) = 0$ and $F_{\lambda_0}''(r_0) \neq 0$ 
      or $F_{\lambda_0}'(r_0) = F_{\lambda_0}''(r_0) = F_{\lambda_0}'''(r_0) = 0$ and $F_{\lambda_0}''''(r_0) > 0$, in both cases 
      the minimum of $F_{\lambda}$ persists near $r = r_0$ in $\lambda$ near $\lambda_0$. 
      
      In order to show that $(-\infty,0] \subset I_-$, we apply the proof of Lemma \ref{lemma-minus} to $f(r) = r^{-1} F(r)$, which holds due to the fast decay 
      $$
      r^{d} [f'(r)]^2 \to 0, \quad 
      r^d [f(r)]^4 \to 0, \quad 
      r^{d-1} f(r) f'(r) \to 0 \quad \mbox{\rm as} \quad r \to 0
      $$ 
      if $d \geq 5$ (no decay holds if $d = 4$). Integrating (\ref{eq-self-adjoint-form}) on $[0,R]$ with  $f(r) = r^{-1} F(r)$ for $c \leq F(r) \leq \sqrt{d-3}$ and $\lambda \leq 0$ yields
      $$
      R^{d-2} F'(R) \geq \frac{c}{d+1} R^{d+1} + (c - \sqrt{d-3}) R^{d-3},
      $$ 
      due to the fast decay $r^{d-1} f'(r) \to 0$ as $r \to 0$.
      The lower bound implies that $F'(r) > 0$ for sufficiently large $r$. Hence, $(-\infty,0] \in I_-$ if $d \geq 5$.
      
      \vspace{0.25cm}
      
    \item The proof of Lemma \ref{lemma-decay} also applies to $f(r) = r^{-1} F(r)$.  Integrating (\ref{eq-self-adjoint-form}) on $[0,R]$ with  $f(r) = r^{-1} F(r)$ for $c \leq F(r) \leq \sqrt{d-3}$ yields
    $$
    R^{d-2} F'(R) \geq \frac{c}{d+1} R^{d+1} + (c - \sqrt{d-3}) R^{d-3} - \frac{\lambda \sqrt{d-3}}{d-2} R^{d-2},
    $$ 
    which is a contradiction with $F'(r) < 0$ for sufficiently large $r$. Hence $c = 0$ and $F(r),F'(r) \to 0$ as $r \to \infty$.
\end{itemize}

\begin{rem}
	\label{remark-error}
	Uniqueness of the solution in Theorem \ref{theorem-Selem} was claimed in Section 4 of \cite{Selem2013}, however, we believe that the proof was incorrect. Indeed, assuming two solutions $\Psi_{\lambda_1}(r)$ and $\Psi_{\lambda_2}(r)$ for two values $\lambda_1$ and $\lambda_2$ in Theorem \ref{theorem-Selem}, we construct a quotient 
	$$
	\rho(t) = \frac{\Psi_1(t)}{\Psi_2(t)},
	$$
	which satisfies the differential equation
\begin{equation}
\label{rho-eq}
	\rho''(t) + \left[ d - 4 + \frac{2 \Psi_2'(t)}{\Psi_2(t)} \right] \rho'(t) + \Psi_2(t)^2 \rho(t) [\rho(t)^2 - 1] + (\lambda_1 - \lambda_2) e^{2t} \rho(t) = 0.
\end{equation}
	It follows from (\ref{lambda-solution-behavior}) that 
\begin{equation}
\label{rho-solution}
	\rho(t) = 1 - \frac{\lambda_1 - \lambda_2}{4d - 10} e^{2t} + \mathcal{O}(e^{4t}) \quad \mbox{\rm as} \quad t \to -\infty.
\end{equation}
	A rescaling of time was applied in the arguments of \cite{Selem2013} to make the last term in (\ref{rho-eq}) small but was not applied to the second term of the expansion (\ref{rho-solution}). As a result, the differential equation (\ref{rho-eq}) was replaced by a differential inequality which led to a contradiction in \cite{Selem2013}. With the proper scaling of time in both (\ref{rho-eq}) and (\ref{rho-solution}), transformation of the differential equation to a differential inequality cannot be justified.
\end{rem}

Figure \ref{fig:Fsol} illustrates the shooting method used in the proof of Theorem \ref{theorem-Selem}. The left panel shows $F(r)$ as the unique classical solution to the initial-value problem (\ref{statGPsingular}), whereas the right panel shows $\Psi(t)$ as a solution to the differential equation (\ref{eq-psi-singular}) with the boundary conditions (\ref{bc-singular}). For $d = 5$, we compute numerically the solutions for three different values of $\lambda$. For a special value of $\lambda = \lambda_{\infty}$, the solution $F$ gives the limiting singular solution 
$f_{\infty} \in \mathcal{E}$ after the transformation $f(r) = r^{-1} F(r)$. For values of $\lambda$ above (below) $\lambda_{\infty}$, the solution crosses the zero level and diverges to negative infinity (attains a minimum and diverges to positive infinity). We have found numerically that the value of $\lambda_{\infty}$ is unique (see Remark \ref{rem-unique-3}).

\begin{figure}[htp!]
	\centering
	\includegraphics[width=0.45\textwidth]{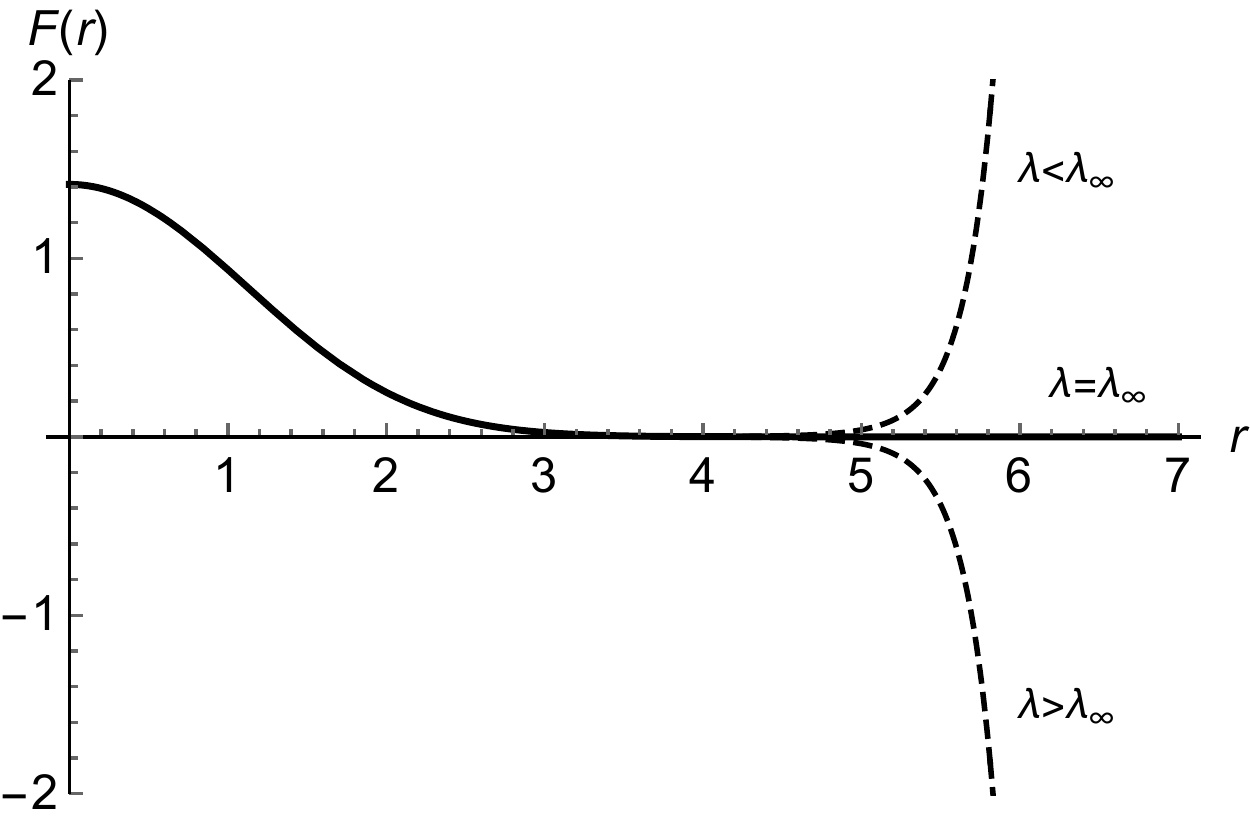} 
	\includegraphics[width=0.45\textwidth]{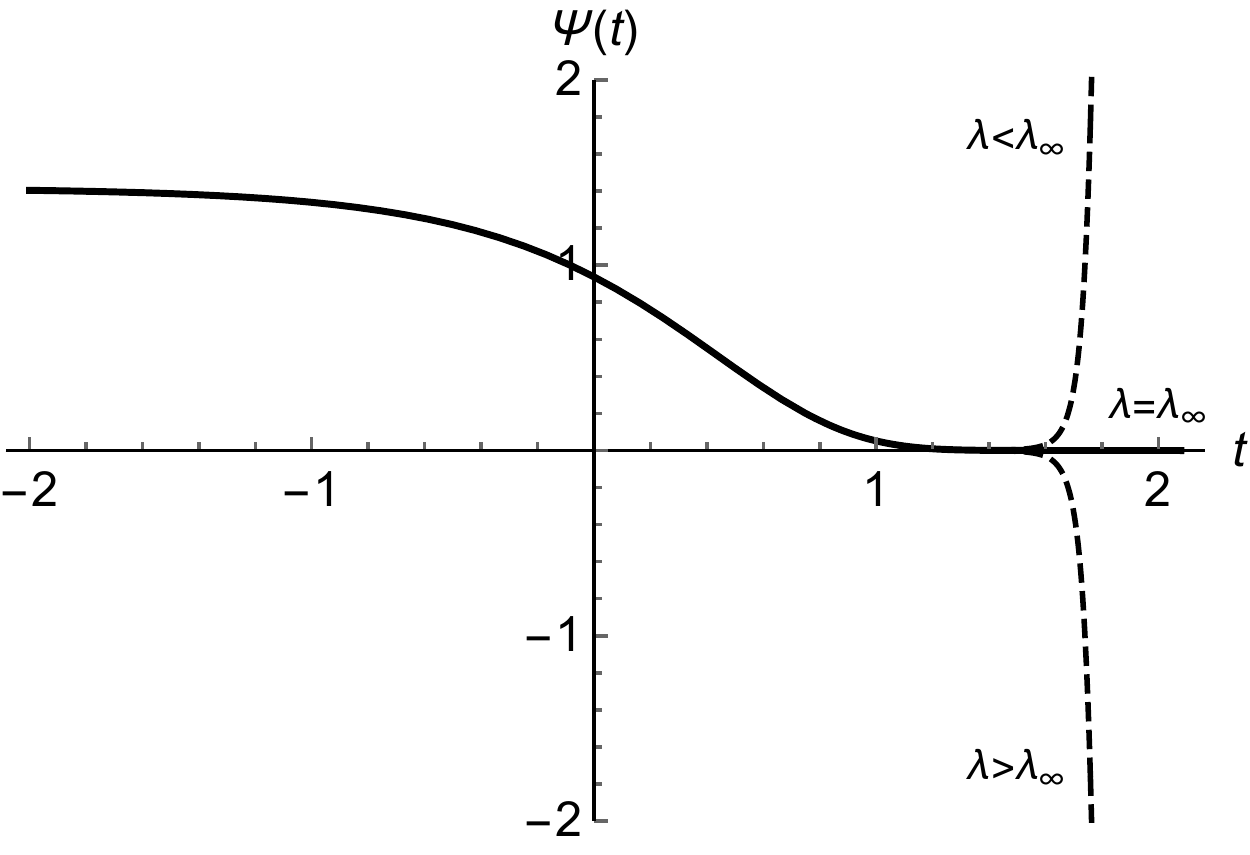}
	\caption{Plot of the unique solution $F(r)$ satisfying the initial-value problem (\ref{statGPsingular}), together with $\Psi(t)$ defined by the transformation (\ref{transformaton-singular}), for $d=5$  and three values of $\lambda$. For $\lambda = \lambda_{\infty}$, the solution gives the limiting singular solution $f_{\infty} \in \mathcal{E}$ after the transformation $f(r) = r^{-1} F(r)$. The dashed lines show solutions for values of $\lambda$ slightly deviating from $\lambda_\infty$.}
	\label{fig:Fsol}
\end{figure}

\begin{rem}
	The limiting value $\lambda_{\infty}$ can be computed semi-analytically. 
	Near the origin, the limiting singular solution $f_{\infty}$ has the form
	$$
	f_{\infty}(r)=\frac{\sqrt{d-3}}{r}\left(1+\sum_{n=1}^{\infty} c_n r^{2n}\right),
	$$
	where the coefficients of the Taylor series are explicit polynomials in $\lambda$.  The radius of convergence of this series is not big enough to guarantee the decay of $f_{\infty}$ to zero at infinity but this problem can be resolved by making the Pad\'e approximation. The results for different values of $d \geq 5$ are collected in Table \ref{tab:my_label}.
\end{rem}

\begin{table}[h]
	\centering
	\begin{tabular}{c|r|c|r|c|r|c|r}
		$d$ & $\lambda_\infty$ & $d$ & $\lambda_\infty$ & $d$ & $\lambda_\infty$ & $d$ & $\lambda_\infty$\\ \hline
		5 & 4.01036 & 9 & 8.68938 & 13 & 12.89681 & 17 & 16.96618\\
		6 & 5.27039 & 10 & 9.76437 & 14 & 13.92174 & 18 & 17.97452\\
		7 & 6.45486 & 11 & 10.82105 & 15 & 14.94073 & 19 & 18.98085\\
		8 & 7.58946 & 12 & 11.86408 & 16 & 15.95519 & 20 & 19.98563\\
	\end{tabular}
\vspace{0.25cm} 
	\caption{Approximate values of $\lambda_\infty$ for different values of $d$.}
	\label{tab:my_label}
\end{table}

\section{Proof of Theorem \ref{theorem-2}}
\label{sec-proof-3}

Here we study three particular solutions to the differential equation (\ref{eq-psi-singular}) in order to prove Theorem \ref{theorem-2}. 

One solution to (\ref{eq-psi-singular}) is defined from the boundary conditions (\ref{bc-singular}) and is denoted by $\Psi_{\lambda}$. It satisfies the asymptotic behavior (\ref{lambda-solution-behavior}) as $t \to -\infty$.

Another solution to (\ref{eq-psi-singular}) is obtained from the unique solution constructed in Lemma \ref{lemma-1} after the scaling transformation $\Psi(t) = e^t \psi(t)$, where $\psi(t)$ satisfies the differential equation (\ref{eq-psi}) and the boundary conditions (\ref{bc}). In order to distinguish this solution from $\Psi_{\lambda}$, we denote it by $\Psi_b$. It follows from (\ref{asymptotic-psi-prime}) that $\Psi_b$ satisfies the asymptotic behavior
\begin{equation}
\label{asym-b-solution}
\Psi_b(t) = b e^t - (\lambda b + b^3) (2d)^{-1} e^{3t} + \mathcal{O}(e^{5t}) \quad \mbox{\rm as} \quad t \to -\infty.
\end{equation}
The truncated (autonomous) version to the second-order equation (\ref{eq-psi-singular}) is given by 
\begin{equation}
    \label{eq-psi-truncated}
    \Theta''(t) + (d-4) \Theta'(t) + (3-d) \Theta(t) + \Theta(t)^3 = 0.
\end{equation}
With an elementary exercise, we have the following lemma.

\begin{lemma}
\label{lemma-Psi-inf}
Fix $d \geq 5$. There exists a unique orbit of the truncated equation (\ref{eq-psi-truncated}) on the phase plane
$(\Theta,\Theta')$ that connects the equilibrium points $(0,0)$ and $(\sqrt{d-3},0)$.
\end{lemma}

\begin{proof}
The equilibrium point $(0,0)$ is a saddle point with two roots $\kappa_1 = 1$ and $\kappa_2 = 3 - d < 0$ of the characteristic equation
\begin{equation}
\label{kappa-negative-infinity}
\kappa^2 +  (d-4) \kappa + (3-d) = 0.
\end{equation}
By the unstable curve theorem, there exists a unique unstable curve on the 
plane $(\Theta,\Theta')$ tangential to the direction $(1,1)$, along which two orbits exist satisfying $\Theta(t) \to 0$ as $t \to -\infty$. One orbit is connected to $(0,0)$ in the first quadrant of the $(\Theta,\Theta')$-plane and the other orbit is connected to $(0,0)$ in the third quadrant. Because the stable curve is connected to $(0,0)$ in the second and fourth quadrants and the orbits of the planar system do not intersect away from the equilibrium points, the unstable orbit connected to $(0,0)$ in the first quadrant stays in the right half-plane with positive $\Theta$ and the unstable orbit connected to $(0,0)$ in the third quadrant stays in the left half-plane with negative $\Theta$. For the proof of the lemma, we only consider the former unstable orbit and introduce the energy function
$$
V(\Theta,\Theta') := \frac{1}{2} (\Theta')^2 +
\frac{1}{2} (3-d) \Theta^2 + \frac{1}{4} \Theta^4.
$$
If $\Theta(t) \in C^2(\mathbb{R})$ is a solution
to the second-order equation (\ref{eq-psi-truncated}), then
\begin{equation}
    \label{monotonicity-E}
\frac{d}{dt} V(\Theta,\Theta') = (4 - d) (\Theta')^2 \leq 0.
\end{equation}
Since $V(\Theta,\Theta')$ is bounded from below,
and its value is monotonically decreasing, the unstable orbit stays in a compact region of the right-half of the phase plane $(\Theta,\Theta')$. No periodic orbits exist in this compact region, because if $\Theta(t+T) = \Theta(t)$ is periodic with the minimal period $T > 0$, then we get contradiction with (\ref{monotonicity-E}):
$$
0 = V(\Theta,\Theta') |_{t = T} - V(\Theta,\Theta') |_{t = 0} = (4 - d) \int_0^T \left( \Theta' \right)^2 dt < 0.
$$
Hence, the unstable curve from $(0,0)$ has the limit set at the stable equilibrium point. The only stable equilibrium point in the right-half of the phase plane $(\Theta,\Theta')$ is the point $(\sqrt{d-3},0)$, hence $(0,0)$ and $(\sqrt{d-3},0)$ are connected by the unique heteroclinic orbit.
\end{proof}

Since the truncated equation (\ref{eq-psi-truncated}) is autonomous, the unique orbit of Lemma \ref{lemma-Psi-inf} can be parameterized by the time translation such that
\begin{equation}
\label{asym-Theta-solution}
\Theta(t+t_0) = e^{t+t_0} - (2d)^{-1} e^{3(t + t_0)} + \mathcal{O}(e^{5(t+t_0)}) \quad \mbox{\rm as} \quad t \to -\infty,
\end{equation}
where $t_0 \in \mathbb{R}$ is arbitrary
and $\Theta(t)$ is uniquely defined. It follows by comparing the asymptotic behaviors (\ref{asym-b-solution}) and (\ref{asym-Theta-solution}) that 
the parameter $b$ plays the same role as the translation parameter $t_0$ with the correspondence $t_0 = \log b$. The following lemma states that, when $b$ is sufficiently large, the solution $\Psi_b$ translated by $\log b$ converges to $\Theta$ on the negative half-line.

\begin{lemma}
\label{lemma-Psi}
Fix $d \geq 5$ and $\lambda \in \mathbb{R}$.
There exist $b_0 > 0$ (sufficiently large) and $C_0 > 0$ such that
the unique solution $\Psi_b$ to the second-order equation (\ref{eq-psi-singular})
with the asymptotic behavior (\ref{asym-b-solution}) satisfies
\begin{equation}
\label{bound-solution-b}
    \sup_{t \in (-\infty,0]} |\Psi_b(t-\log b) - \Theta(t)| 
    +     \sup_{t \in (-\infty,0]} |\Psi_b'(t-\log b) - \Theta'(t)| \leq C_0 b^{-2}, \quad b \geq b_0,
\end{equation}
where $\Theta$ is the uniquely defined solution to the truncated equation (\ref{eq-psi-truncated})
with the asymptotic behavior (\ref{asym-Theta-solution}) in Lemma \ref{lemma-Psi-inf}.
\end{lemma}

\begin{proof}
By translating $t$, we rewrite (\ref{eq-psi-singular}) for $\Psi_b(t)$ in the form:
\begin{equation}
    \label{eq-psi-singular-translated}
    \Psi''(t) + (d-4) \Psi'(t) + (3-d) \Psi(t) + \Psi(t)^3
    = -\lambda b^{-2} e^{2(t + \log b)} \Psi(t) + b^{-4} e^{4(t + \log b)} \Psi(t).
\end{equation}
The solution $\Psi_b(t-\log b)$ is decomposed near the uniquely defined solution $\Theta(t)$
to the truncated equation (\ref{eq-psi-truncated}) by using $\Psi_b(t-\log b) = \Theta(t) + \Upsilon(t)$,
where $\Upsilon(t)$ satisfies the persistence problem
\begin{equation}
\label{persistence-varphi}
L \Upsilon = f_b (\Theta + \Upsilon) + N(\Theta,\Upsilon),
\end{equation}
where
\begin{eqnarray*}
(L \Upsilon)(t) & = & \Upsilon''(t) + (d-4) \Upsilon'(t) + (3-d) \Upsilon(t) + 3 \Theta(t)^2 \Upsilon(t), \\
f_b(t) & = & -\lambda b^{-2} e^{2t} + b^{-4} e^{4t}, \\
N(\Theta, \Upsilon) & = & -3 \Theta \Upsilon^2 - \Upsilon^3.
\end{eqnarray*}
There exist two linearly independent solutions $\Theta'(t)$ and $\Xi(t)$ of the homogeneous equation $L \Upsilon = 0$,
where $\Theta'(t)$ is due to translation of the truncated equation (\ref{eq-psi-truncated}) and
$\Xi(t)$ is the linearly independent solution satisfying the Wronskian relation 
from Liouville's theorem:
\begin{equation}
\label{Wronskian-relation}
W(\Theta',\Xi)(t) := \Theta'(t) \Xi'(t) - \Theta''(t) \Xi(t) = W_{\infty} e^{(4-d) t},
\end{equation}
where $W_{\infty}$ is an arbitrary nonzero constant. For unique normalization
of $\Xi(t)$, we can just set $W_{\infty} = 1$. Since $\Theta'(t)$ decays to zero as $t \to -\infty$ according to 
$$
\Theta'(t) = e^t + \mathcal{O}(e^{3t}) \quad \mbox{\rm as} \quad t \to -\infty,
$$
it follows from the integration of (\ref{Wronskian-relation}) 
with $W_{\infty} = 1$ that $\Xi(t)$ grows as $t \to -\infty$ according to
$$
\Xi(t) = (2-d)^{-1} e^{(3-d)t} + \mathcal{O}(e^{(5-d)t}) \quad \mbox{\rm as} \quad t \to -\infty.
$$
By solving the second-order differential equation (\ref{persistence-varphi})
with the variation of parameters, we obtain the integral equation for $\Upsilon(t)$:
\begin{equation}
\label{persistence-varphi-integral}
\Upsilon(t) = - \int_{-\infty}^t e^{(d-4) t'} \left[ \Xi(t') \Theta'(t) - \Xi(t) \Theta'(t') \right]
\left[ f_b(t') (\Theta(t') + \Upsilon(t')) + N(\Theta(t'),\Upsilon(t')) \right] dt',
\end{equation}
where the choice of integration from $-\infty$ to $t$ ensures that $\Upsilon(t)$ does not grow as $t \to -\infty$
along the solution $\Xi(t)$ and does not introduce the additional translation in time along the solution
$\Theta'(t)$. In order to prove existence of small solutions to the integral equation
(\ref{persistence-varphi-integral}) on $(-\infty,0]$ for large $b$, we introduce $\tilde{\Upsilon}(t) := e^{-t} \Upsilon(t)$,
so that
\begin{equation}
\label{bound-on-bound}
\sup_{t \in (-\infty,0]} |\Upsilon(t)| \leq \sup_{t \in (-\infty,0]} |\tilde{\Upsilon}(t)|.
\end{equation}
Then, $\tilde{\Upsilon}(t)$ is found from the integral equation
\begin{equation}
\label{persistence-varphi-tilde}
\tilde{\Upsilon}(t) = - \int_{-\infty}^t K(t,t')
\left[ f_b(t') \left( e^{-t'} \Theta(t') + \tilde{\Upsilon}(t') \right) +
e^{2 t'} N(e^{-t'} \Theta(t'),\tilde{\Upsilon}(t')) \right] dt',
\end{equation}
where
\begin{eqnarray*}
K(t,t') = \left[ e^{(d-3) t'} \Xi(t') \right] \left[ e^{-t} \Theta'(t) \right] - e^{(d-2)(t'-t)}
\left[ e^{(d-3) t} \Xi(t) \right] \left[ e^{-t'} \Theta'(t') \right].
\end{eqnarray*}
Thanks to the exponential rates of $\Theta'(t)$ and $\Xi(t)$ as $t \to -\infty$,
there exists a positive constant $A_0$ such that
$$
\sup_{t \in (-\infty,0], t' \in (-\infty,0]} |K(t',t)| \leq A_0.
$$
It is also clear that 
$$
\sup_{t \in (-\infty,0]} |f_b(t)| \leq (|\lambda| + 1) b^{-2}, \quad b \geq 1,
$$
so that the inhomogeneous term of the integral equation (\ref{persistence-varphi-tilde})
is small if $b$ is large. By the same fixed-point iterations as in the proof of Lemma \ref{lemma-1},
it follows that there exists a sufficiently large $b_0$ such that for every $b \geq b_0$
there exists the unique solution $\tilde{\Upsilon}$ to the integral equation (\ref{persistence-varphi-tilde})
in a closed subset of the Banach space $L^{\infty}(-\infty,0)$ 
satisfying the bound
\begin{equation}
\label{bound-on-bound-too}
\sup_{t \in (-\infty,0]} |\tilde{\Upsilon}(t)| \leq C_0 b^{-2},
\end{equation}
where $C_0 > 0$ is a suitable chosen constant. Bounds (\ref{bound-on-bound}) and (\ref{bound-on-bound-too}) yield the first bound in (\ref{bound-solution-b}). Since $\tilde{\Upsilon} \in C^1(-\infty,0)$ 
by bootstrapping arguments similar to those in the proof of Lemma \ref{lemma-1}, the second bound in (\ref{bound-solution-b}) follows by differentiating (\ref{persistence-varphi-tilde}) in $t$ 
and using bound (\ref{bound-on-bound-too}).
\end{proof}

\begin{rem}
It follows from the integral equation (\ref{persistence-varphi-tilde}) with the account of exponential
rates of $\Xi(t)$ and $\Theta'(t)$ that $\tilde{\Upsilon}(t) = \lambda b^{-2} (2d)^{-1} e^{2t} + \mathcal{O}(e^{4t})$
as $t \to -\infty$, in agreement with the asymptotic expansions (\ref{asym-b-solution}) and (\ref{asym-Theta-solution})
for $\tilde{\Upsilon}(t) = e^{-t} [\Psi_b(t-\log b) - \Theta(t)]$.
\end{rem}

In addition to the solutions $\Psi_{\lambda}$ and $\Psi_b$ to the differential equation (\ref{eq-psi-singular}), which are defined from the behavior as $t \to - \infty$, we define the third solution to (\ref{eq-psi-singular}) 
from the decaying behavior as $t \to +\infty$. This solution to (\ref{eq-psi-singular}) is denoted by $\Psi_C(t)$. Its existence follows from a modification of the result of Lemma \ref{lemma-0}. 

\begin{lemma}
\label{lemma-Psi-asymptotics}
Fix $d \geq 1$ and $\lambda \in \mathbb{R}$. There exists a one-parameter family of solutions to the second-order equation (\ref{eq-psi-singular}) denoted by $\Psi_C(t)$ 
such that $\Psi_C(t),\Psi_C'(t) \to 0$ as $t \to +\infty$ and
    \begin{equation}
    \label{asymptotics-infinity-psi}
\Psi_C(t) \sim C e^{\frac{\lambda - d + 2}{2} t} e^{-\frac{1}{2} e^{2t}} \quad \mbox{\rm as} \quad t \to +\infty,
\end{equation}
for some $C \in \mathbb{R}$, where the asymptotic correspondence can be differentiated. Moreover, the solution $\Psi_C(t)$ is extended globally for every $t \in \mathbb{R}$.
\end{lemma}

\begin{proof}
By the chain rule in (\ref{transformaton-singular}), 
if $\Psi(t),\Psi'(t) \to 0$ as $t \to +\infty$, then 
$F(r), F'(r) \to 0$ as $r \to \infty$. Since $f(r) = r^{-1} F(r)$, 
this decay implies $f(r),f'(r) \to 0$ as $r \to \infty$. The precise asymptotic correspondence (\ref{asymptotics-infinity-psi})
is obtained from (\ref{asymptotics-infinity}) and $\Psi_C(t) = e^t f(e^t)$.
Global continuation of $\Psi_C$ on $\mathbb{R}$ follows by Lemma \ref{lemma-2} from the global continuation of the solution $f \in C^2(r_0,\infty)$ for some $r_0 \in (0,\infty)$ to $f \in C^2(0,\infty)$.
\end{proof}

Finally, we discuss linearization of the truncated equation (\ref{eq-psi-truncated}) at the equilibrium point $(\sqrt{d-3},0)$. 
The characteristic equation
\begin{equation}
\label{char-eq-roots}
\kappa^2 +  (d-4) \kappa + 2(d-3) = 0
\end{equation}
admits the following two roots
\begin{equation}
\label{linearization-nonzero}
\kappa_{\pm} := -\frac{1}{2}(d-4) \pm \frac{1}{2}\sqrt{d^2 - 16 d + 40}
\end{equation}
If $5 \leq d \leq 12$, the roots are complex-conjugate and can be written as 
\begin{equation}
\label{complex-roots}
\kappa_{\pm} = - \beta \pm i \alpha,
\end{equation}
where real and positive $\alpha$ and $\beta$ are given by (\ref{alpha-beta}). If $d \geq 13$, then the roots are real and negative with ordering 
\begin{equation}
\label{real-roots}
\kappa_- < \kappa_+ < 0. 
\end{equation}
The difference between 
the two cases can be observed when the solution $\psi$ of Theorem \ref{theorem-1} is transformed to the variable $\Psi$, in which case 
it becomes the intersection of the second solution $\Psi_b$ defined as $t \to -\infty$ and the third solution $\Psi_C$ for some $C = C(b) > 0$ defined as $t \to +\infty$. Both solutions satisfy the differential equation (\ref{eq-psi-singular}) for the particular value of $\lambda = \lambda(b)$.

Figure \ref{fig:psid5b6} shows both components $\psi(t)$ and $\Psi(t)$ for the solution of Theorem \ref{theorem-1} for $d = 5$ (top) and $d = 13$ (bottom) that corresponds to $b = 14000$.
Compared to the component $\psi(t)$ which is monotonically decreasing on $\mathbb{R}$, the component $\Psi(t)$ decays to zero at both infinities. 
In addition, $\Psi(t)$ develops oscillations at the intermediate range of $t$ for $d = 5$ and no oscillations for $d = 13$.

\begin{figure}[htp!]
	\centering
	\includegraphics[width=0.45\textwidth]{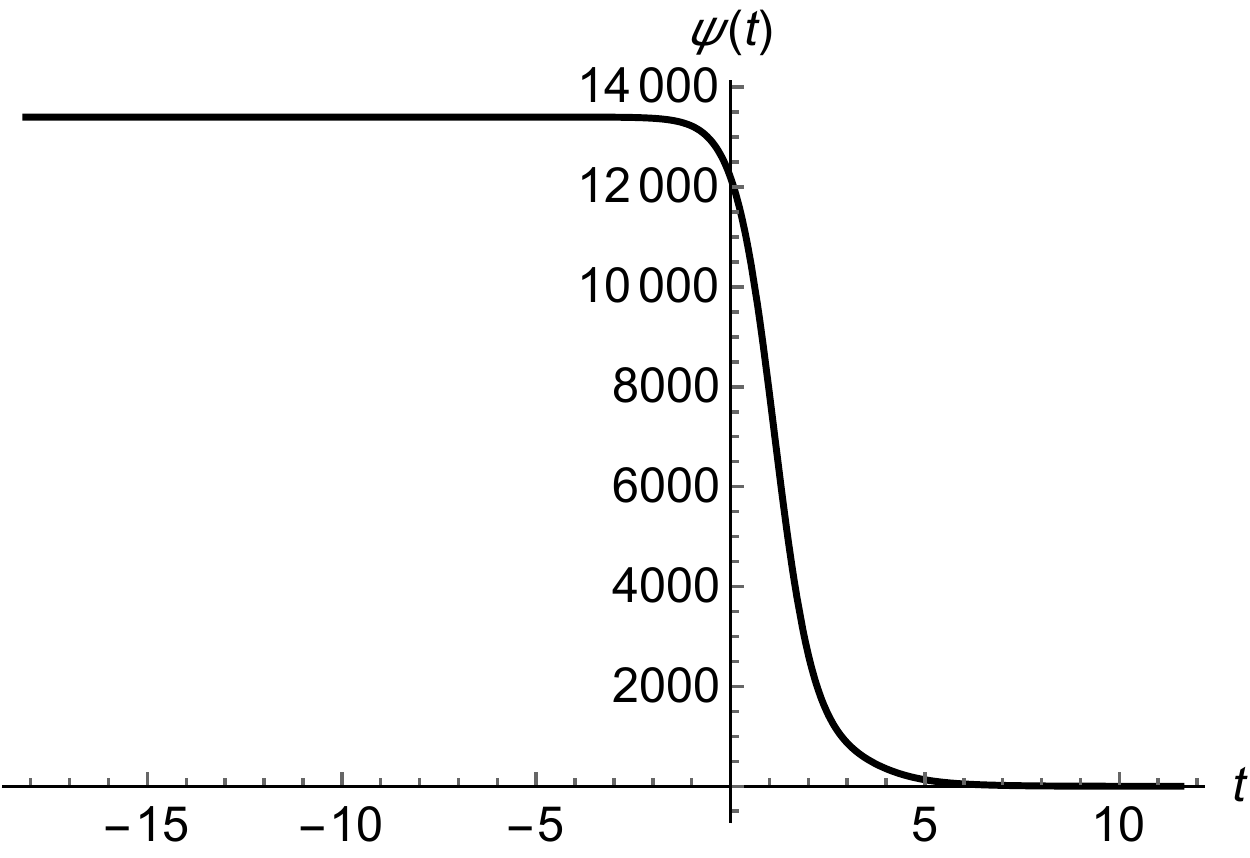}\qquad
	\includegraphics[width=0.45\textwidth]{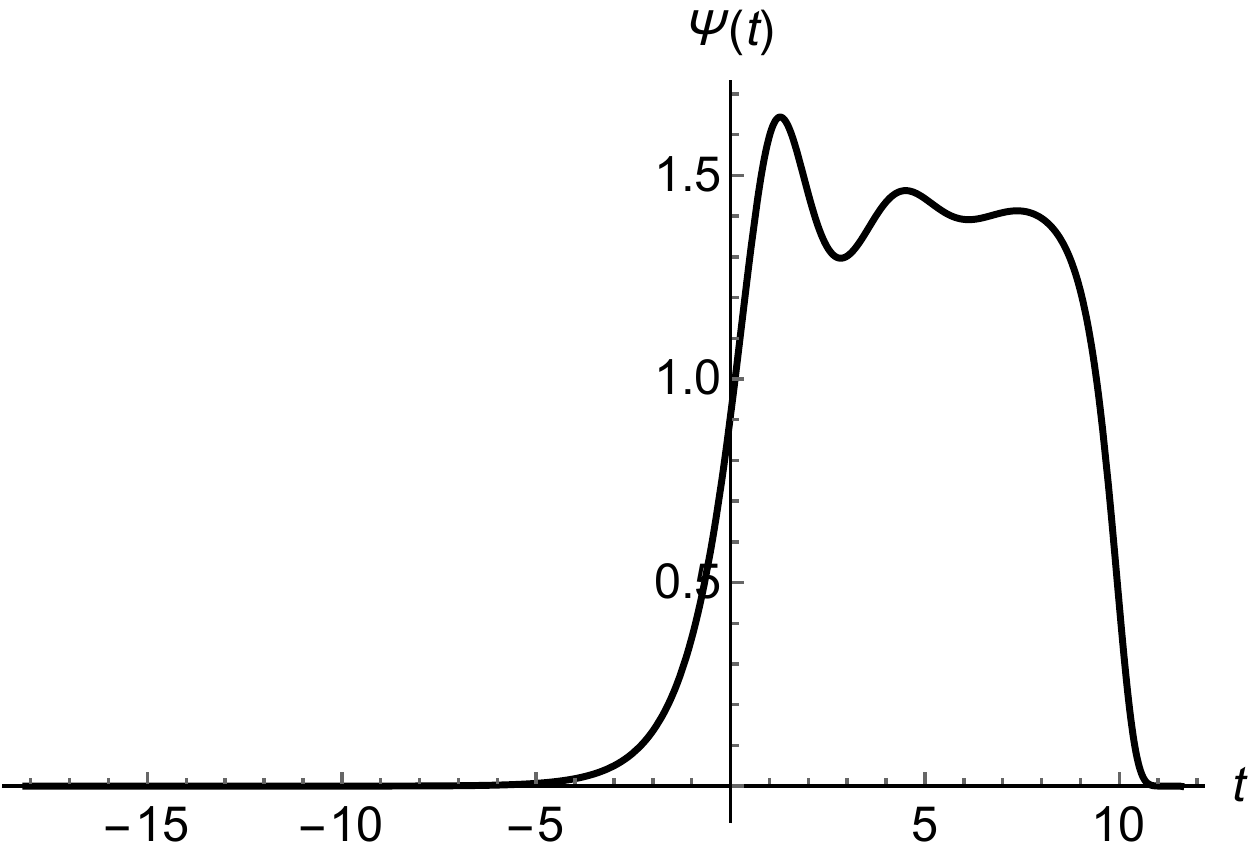}\\
	\includegraphics[width=0.45\textwidth]{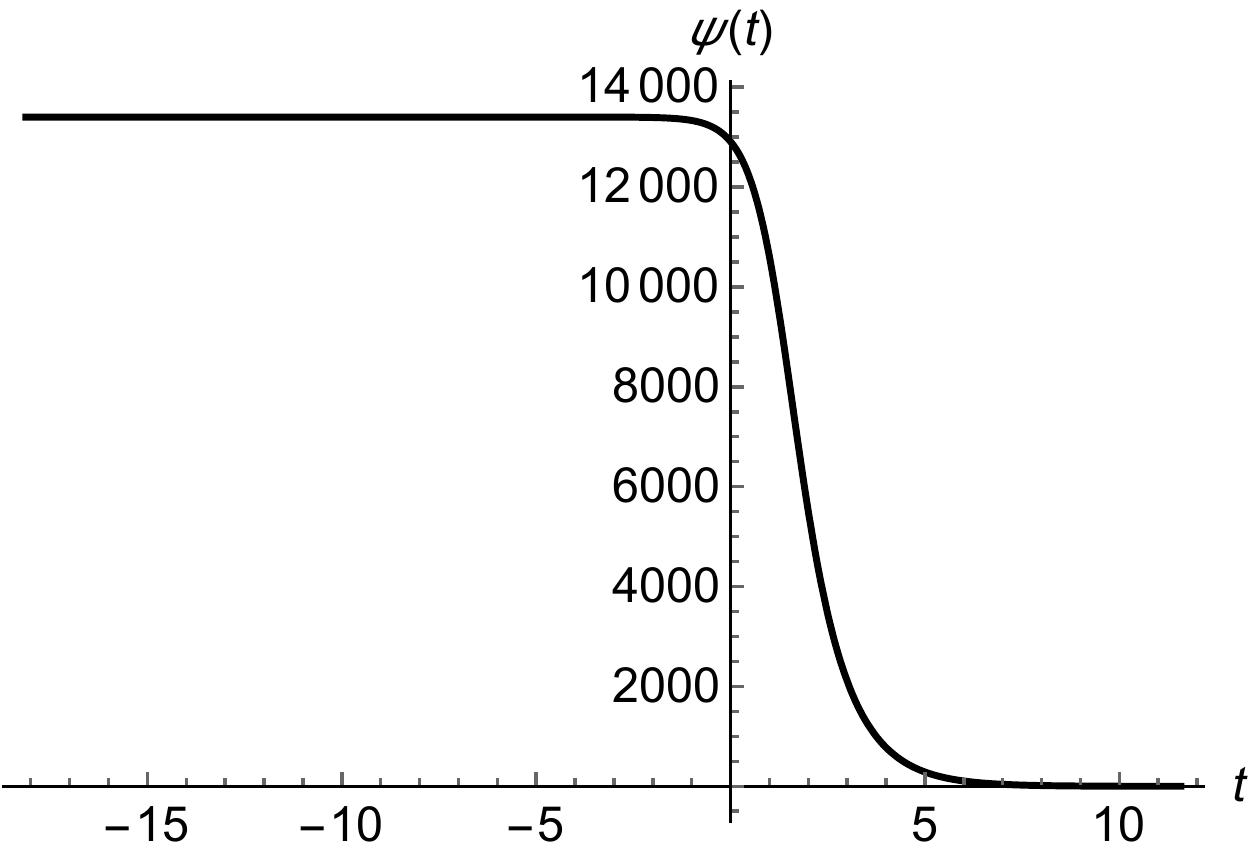}\qquad
	\includegraphics[width=0.45\textwidth]{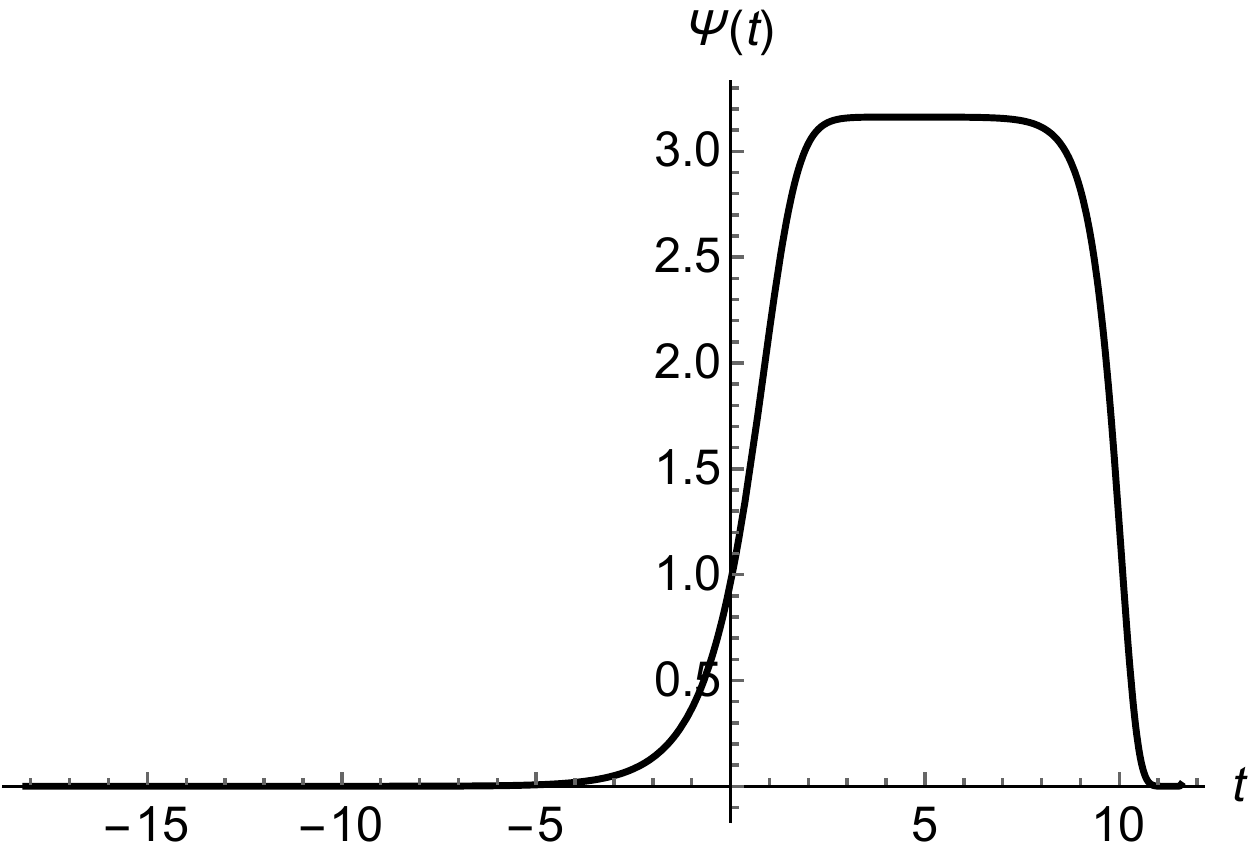}
	\caption{Components $\psi$ (left) 
		and $\Psi$ (right) for the solution of Theorem \ref{theorem-1} plotted versus $t$ for $d=5$ (top) and $d=13$ (bottom) with $b = 14000$.}
	\label{fig:psid5b6}
\end{figure}
 
Because of the difference between the oscillatory and monotone behavior of the solutions of Theorem \ref{theorem-1} in variable $\Psi$, the proof of Theorem \ref{theorem-2} 
is developed separately for $5 \leq d \leq 12$ and $d \geq 13$.

\subsection{Oscillatory behavior for $5 \leq d \leq 12$}

By Lemma \ref{lemma-Psi-inf}, there exists the unique solution $\Theta$ to the truncated equation (\ref{eq-psi-truncated})
with the asymptotic behavior (\ref{asym-Theta-solution}). The following lemma described the oscillatory behavior of $\Theta(t)$ as $t \to +\infty$.

\begin{lemma}
\label{lemma-spiral-Theta}
Fix $5 \leq d \leq 12$. There exist $t_0 > 0$ (sufficiently large), $A_0 > 0$, $\delta_0 \in [0,2\pi)$, and $C_0 > 0$
such that the unique solution $\Theta$ of Lemma \ref{lemma-Psi-inf} satisfies the following oscillatory behavior:
\begin{equation}
\label{behavior-spiral}
\sup_{t \in [t_0,\infty)} |\Theta(t) - \sqrt{d-3} - A_0 e^{-\beta t} \sin(\alpha t + \delta_0) | \leq C_0 e^{-2\beta t_0}.
\end{equation}
where $\alpha$ and $\beta$ are given by (\ref{alpha-beta}).
\end{lemma}

\begin{proof}
The equilibrium point $(\sqrt{d-3},0)$ is a stable spiral point of the truncated equation (\ref{eq-psi-truncated}) for $5 \leq d \leq 12$ due to the roots (\ref{complex-roots}) of the characteristic equation (\ref{char-eq-roots}). Quadratic terms beyond the linearization at $(\sqrt{d-3},0)$ can be removed by a near-identity transformation if $\kappa_{\pm} = - \beta \pm i \alpha$. By the Hartman--Grobman theorem, there exists a $C^2$-diffeomorphism, under which the dynamics of the truncated equation (\ref{eq-psi-truncated})
near $(\sqrt{d-3},0)$ is conjugate to the dynamics of the linearized equation.
The asymptotic behavior (\ref{behavior-spiral}) follows from the solution of the linearized equation
and the existence of the $C^2$-diffeomorphism.
\end{proof}

\begin{rem}
	It follows from the dynamical system theory that the bound (\ref{behavior-spiral}) can be extended to $\Theta'(t)$  as follows: 
	\begin{equation}
	\label{behavior-spiral-derivative}
	\sup_{t \in [t_0,\infty)} |\Theta'(t) - \alpha A_0 e^{-\beta t} \cos(\alpha t + \delta_0) + \beta A_0 e^{-\beta t} \sin(\alpha t + \delta_0)  | \leq C_0 e^{-2\beta t_0}.
	\end{equation}
For simplicity of writing, we will not write henceforth the explicit bounds on the derivatives.
\end{rem}

By extending Lemma \ref{lemma-Psi} and using Lemma \ref{lemma-spiral-Theta}, we prove the oscillatory behavior of the solution $\Psi_b(t)$ 
at the intermediate values of $t$ as $b \to \infty$.

\begin{lemma}
\label{lemma-spiral-Psi}
Fix $5 \leq d \leq 12$ and $\lambda \in \mathbb{R}$. For fixed $T > 0$ and $a \in \left(0,\frac{4}{d}\right)$,
there exist $b_{T,a} > 0$ and $C_{T,a} > 0$ such that the unique solution $\Psi_b$ to the second-order equation (\ref{eq-psi-singular})
with the asymptotic behavior (\ref{asym-b-solution}) satisfies
\begin{equation}
\label{bound-solution-b-extended}
    \sup_{t \in [0,T + a \log b ]} |\Psi_b(t-\log b) - \Theta(t)| \leq C_{T,a} b^{-2(1-a)}, \quad b \geq b_{T,a}.
\end{equation}
Consequently, it follows that 
\begin{eqnarray}
\nonumber
&& |\Psi_b(T + (a-1) \log b)   - \sqrt{d-3} - A_0 b^{-a \beta} e^{-\beta T} \sin(\alpha T + \delta_0 + a \alpha \log b) | \\
&& \qquad \qquad \leq C_{T,a} \max\{b^{-2a \beta}, b^{-2(1-a)}\}, \quad 
b \geq b_{T,a},
\label{behavior-Psi-spiral}
\end{eqnarray}
where $(\alpha,\beta)$ are given by (\ref{alpha-beta}), $(A_0,\delta_0)$ are defined in (\ref{behavior-spiral}), and $(b_{T,a},C_{T,a})$ are adjusted appropriately.
\end{lemma}

\begin{proof}
We start with the proof of the bound (\ref{bound-solution-b-extended}). This can be done by
rewriting the integral equation (\ref{persistence-varphi-integral}) in an equivalent
form which is useful for $t \in [0,T + a \log b]$. To do so, we solve
the second-order equation (\ref{persistence-varphi})
with the variation of parameters from $t = 0$ towards $t = T + a \log b > 0$.
This gives us the integral equation in the form:
\begin{eqnarray}
\nonumber
&& \Upsilon(t) = \Upsilon(0) \left[ \Xi'(0) \Theta'(t) - \Xi(t) \Theta''(0) \right]
+ \Upsilon'(0) \left[ \Xi(t) \Theta'(0) - \Xi(0) \Theta'(t) \right]  \\
&& 
- \int_{0}^t e^{(d-4) t'} \left[ \Xi(t') \Theta'(t) - \Xi(t) \Theta'(t') \right]
\left[ f_b(t') (\Theta(t') + \Upsilon(t')) + N(\Theta(t'),\Upsilon(t')) \right] dt'. 
\label{volterra-again}
\end{eqnarray}
By the bound (\ref{bound-solution-b}), there exist $b_0 > 0$ and $C_0 > 0$ such that
$$
|\Upsilon(0)| + |\Upsilon'(0)| \leq C_0 b^{-2}, \quad b \geq b_0.
$$
It follows from the definition of $f_b$ that 
$$
\sup_{t \in [0,T + a \log b]} |f_b(t)| \leq (|\lambda| + 1) b^{-2(1-a)} e^{4 T}, \quad b \geq 1,
$$
where $T > 0$ is fixed independently of $b$. Since $(\sqrt{d-3},0)$ is a stable spiral point
of the truncated equation (\ref{eq-psi-truncated}) for $5 \leq d \leq 12$ with the roots (\ref{complex-roots}),
both $\Theta'(t)$ and $\Xi(t)$ decays to $0$ exponentially fast 
as $t \to +\infty$ such that
$$
|\Theta'(t)| + |\Xi(t)| \leq C_0 e^{-\beta t}, \quad t \geq 0,
$$
for some $b$-independent $C_0 > 0$. The kernel of the integral equation
(\ref{volterra-again}) behaves like $e^{-\beta (t-t')}$ and decays exponentially as $t \to +\infty$.
By the same fixed-point iterations as in the proof of Lemma \ref{lemma-1}, it follows that there exists a sufficiently large $b_{T,a}$ such that for every $b \geq b_{T,a}$
there exists the unique solution $\Upsilon$ to the integral equation (\ref{volterra-again}) in a closed 
subset of Banach space $L^{\infty}(0,T + a \log b)$ satisfying the bound
\begin{equation}
\label{bound-solution-b-extended-too}
\sup_{t \in [0,T + a \log b]} |\Upsilon(t)| \leq C_{T,b} b^{-2(1-a)},
\end{equation}
where $C_{T,b} > 0$ is a suitable chosen constant and $a \in (0,1)$.
Bound (\ref{bound-solution-b-extended-too}) yields (\ref{bound-solution-b-extended}).

Bound (\ref{behavior-Psi-spiral}) follows from (\ref{behavior-spiral}) and (\ref{bound-solution-b-extended})
since $a \log b \to +\infty$ as $b \to \infty$ if $a > 0$ and $b^{-a\beta} \gg b^{-2(1-a)}$ if
$a < \frac{4}{d} < 1$.
\end{proof}

\begin{rem}
The bound (\ref{bound-solution-b}) was used in \cite{Budd1989} without improvement
given by the bound (\ref{bound-solution-b-extended}). The bound (\ref{bound-solution-b}) is not
sufficient for our purpose because if $a = 0$ in Lemma \ref{lemma-spiral-Psi}
then we are not allowed to use the asymptotic behavior (\ref{behavior-spiral})
in order to derive the bound (\ref{behavior-Psi-spiral}).
\end{rem}

\begin{rem}
The constraint $a \in \left(0,\frac{4}{d} \right) \subset (0,1)$ needed to control the small approximation error in the bound (\ref{bound-solution-b-extended})
implies that the oscillatory behavior (\ref{behavior-Psi-spiral}) is observed in $\Psi_b(t)$ for sufficiently large negative $t$, 
yet not in the limit $t \to -\infty$. Indeed, $\Psi_b(t)$ satisfies the asymptotic behavior (\ref{asym-b-solution}) and decays to zero as $t \to -\infty$.
\end{rem}

Let us now turn to the one-parameter solution $\Psi_C(t)$ 
defined by the asymptotic behavior (\ref{asymptotics-infinity-psi}) as $t \to +\infty$. This solution to the differential equation (\ref{eq-psi-singular}) 
is extended globally for every $t \in \mathbb{R}$ by Lemma \ref{lemma-Psi-asymptotics}. For $\lambda = \lambda_{\infty}$, 
there exists a uniquely defined $C = C_{\infty}$ such that $\Psi_{C = C_{\infty}}$ coincides with the unique solution $\Psi_{\infty} := \Psi_{\lambda = \lambda_{\infty}}$ which satisfies the asymptotic behavior (\ref{lambda-solution-behavior}) as $t \to -\infty$. Thus, $\Psi_{\infty} = \Psi_{C = C_{\infty}} = \Psi_{\lambda = \lambda_{\infty}}$ is a bounded function on $\mathbb{R}$. However, the functions $\Psi_{C \neq C_{\infty}}$ and $\Psi_{\lambda \neq \lambda_{\infty}}$ are not globally bounded on $\mathbb{R}$ due to divergence as $t \to -\infty$ and $t \to +\infty$ respectively.

The unique solution $\Psi_C$ is differentiable in $(\lambda,C)$ due to the smooth asymptotic
behavior (\ref{asymptotics-infinity-psi}) and the smoothness of the differential equation (\ref{eq-psi-singular}). Therefore, we can define
\begin{eqnarray}
\label{derivative-C-solution}
\Psi_1 := \partial_{\lambda} \Psi_C |_{(\lambda,C) = (\lambda_{\infty},C_{\infty})}, \quad
\Psi_2 := \partial_C \Psi_C |_{(\lambda,C) = (\lambda_{\infty},C_{\infty})}.
\end{eqnarray}
Functions $\Psi_{1,2}$ satisfy the linear second-order equations written in the form
\begin{equation}
\label{eq-psi-singular-derivative-1}
\mathcal{L}_0 \Psi_1 = f + g \Psi_1, \qquad 
\mathcal{L}_0 \Psi_2 = g \Psi_2,
\end{equation}
where 
\begin{eqnarray*}
(\mathcal{L}_0 \Psi)(t) & := & 
\Psi''(t) + (d-4) \Psi'(t) + 2 (d-3) \Psi(t), \\ 
f(t) & := & - e^{2t} \Psi_{\infty}(t), \\
g(t) & := & 3 (d-3 - \Psi_{\infty}(t)^2) - \lambda_{\infty} e^{2t} + e^{4t}.
\end{eqnarray*}
We add the following technical assumption.
\begin{assumption}
	\label{assumption-1}
Uniquely defined functions $\Psi_{\infty}$ and $\Psi_2$ are assumed to satisfy the following non-degeneracy assumption:
\begin{equation}
\label{degeneracy-1}
\int_{-\infty}^{\infty} e^{(d-2) t} \Psi_{\infty}(t) \Psi_2(t) dt \neq 0.
\end{equation}
\end{assumption}

\begin{rem}
The non-degeneracy assumption (\ref{degeneracy-1}) can be equivalently written as 
\begin{eqnarray*}
\frac{\partial}{\partial C} \int_{-\infty}^{\infty} e^{(d-2)t} \Psi_C(t)^2 dt \biggr|_{\lambda = \lambda_{\infty}, C = C_{\infty}} \neq 0, 
\end{eqnarray*}
or 
\begin{eqnarray*}
\frac{\partial}{\partial C} \int_{0}^{\infty} r^{d-3} F_C(r)^2 dr \biggr|_{\lambda = \lambda_{\infty}, C = C_{\infty}} \neq 0, 
\end{eqnarray*}
or 
\begin{eqnarray*}
\frac{\partial}{\partial C} \int_{0}^{\infty} r^{d-1} f_C(r)^2 dr \biggr|_{\lambda = \lambda_{\infty}, C = C_{\infty}} \neq 0,
\end{eqnarray*}
where $f_C(r) = r^{-1} F_c(r) = r^{-1} \Psi_C(\log r)$.
\end{rem}

\begin{rem}
	One can reformulate the constraint (\ref{degeneracy-1}) from a different point of view. Recall the solution $\Psi_{\lambda}$ to the second-order equation (\ref{eq-psi-singular}) satisfying the asymptotic behavior (\ref{lambda-solution-behavior}) as $t \to -\infty$ and extended globally. Derivative $\partial_{\lambda} \Psi_{\lambda}$ satisfies the same differential equation (\ref{eq-psi-singular-derivative-1}) as $\Psi_1$ 
	but compared to $\Psi_1$, $\partial_{\lambda} \Psi_{\lambda}(t)$ generally diverges as $t \to +\infty$. The condition (\ref{degeneracy-1}) ensures that $\partial_{\lambda} \Psi_{\lambda}$ is not spanned by the derivatives of the solution $\Psi_C(t)$ in $\lambda$ and $C$, which decays to zero as $t \to \infty$.
	Hence, the constraint (\ref{degeneracy-1}) is a transversality condition between the two $C^1$ families of solutions 
	to the differential equation (\ref{eq-psi-singular}) given by $\Psi_{\lambda}$ and $\Psi_C$. Note that this transversality condition was not mentioned in the previous works in \cite{Budd_Norbury1987,Budd1989,DF} on a related subject.
\end{rem}

The following lemma determines the behavior of the solutions 
$\Psi_{1,2}$ for large negative $t$ and the solution $\Psi_{C}$ for parameters $(\lambda,C)$ near the point $(\lambda_{\infty},C_{\infty})$.

\begin{lemma}
\label{lemma-spiral-Psi-C}
Fix $5 \leq d \leq 12$. For fixed $T > 0$ and $a \in (0,1)$,
there exist $b_{T,a} > 0$, $C_{T,a} > 0$, $A_{1,2}$, $B_{1,2}$ 
such that $\Psi_{1,2}$ in (\ref{derivative-C-solution}) 
satisfy for every $t \in (-\infty,(a-1) \log b + T]$: 
\begin{eqnarray}
|\Psi_{1,2}(t) - A_{1,2} e^{-\beta t} \sin(\alpha t) -
B_{1,2} e^{-\beta t} \cos(\alpha t) | \leq C_{T,a}  b^{-2(1-a)} e^{-\beta t}, \quad b \geq b_{T,a},
\label{behavior-C-spiral-derivative}
\end{eqnarray}
where $(\alpha,\beta)$ are given by (\ref{alpha-beta}).
Consequently, there exists $\epsilon_0 > 0$ such that 
for every $\epsilon \in (0,\epsilon_0)$ and for every 
$(\lambda,C) \in \mathbb{R}^2$ satisfying 
\begin{equation}
\label{ball-lambda-C}
(\lambda - \lambda_{\infty})^2 + (C - C_{\infty})^2 \leq \epsilon^2 b^{-2\beta (1-a)},
\end{equation}
it is true for every 
$b \geq b_{T,a}$ and every $t \in [(a-1)\log b,(a-1) \log b + T]$ that 
\begin{eqnarray}
\nonumber
&& |\Psi_C(t) - \sqrt{d-3} 
- [A_1 (\lambda - \lambda_{\infty}) + A_2 (C - C_{\infty})] 
e^{-\beta t} \sin(\alpha t) \\
\nonumber
&& \qquad \qquad \qquad \qquad \qquad
- [B_1 (\lambda - \lambda_{\infty}) + B_2 (C - C_{\infty})] 
e^{-\beta t} \cos(\alpha t) | \\
\nonumber
&& \leq C_{T,a} \left( b^{-2(1-a)}
+ (\lambda - \lambda_{\infty}) b^{-(2-\beta)(1-a)} + (C - C_{\infty}) b^{-(2-\beta)(1-a)} \right. \\
&& \left.\qquad \qquad \qquad \qquad \qquad
+ (\lambda - \lambda_{\infty})^2  b^{2\beta (1-a)} + (C - C_{\infty})^2 b^{2\beta (1-a)}  \right),
\label{behavior-C-spiral}
\end{eqnarray}
where $b_{T,a}$ and $C_{T,a}$ are adjusted appropriately. 
If Assumption \ref{assumption-1} is satisfied, then 
\begin{equation}
\label{coefficients-degeneracy}
A_1 B_2 \neq A_2 B_1.
\end{equation} 
\end{lemma}

\begin{proof}
Since $\Psi_{\infty}(t) = \sqrt{d-3} + \mathcal{O}(e^{2t})$ as $t \to -\infty$, there exist $b_{T,a} > 0$ and $C_{T,a} > 0$ such that
\begin{equation}
\label{bound-f0-f1}
\sup_{t \in (-\infty,(a-1) \log b + T]} (|f(t)| + |g(t)|) \leq C_{T,a} b^{-2(1-a)}, \quad b \geq b_{T,a},
\end{equation}
where $T > 0$ and $a \in (0,1)$ are fixed independently of $b$. 
The left-hand side of linear equations (\ref{eq-psi-singular-derivative-1}) coincides with the linearized equation near the stable spiral point $(\sqrt{d-3},0)$ with two roots (\ref{complex-roots}). By variation of parameters, we can rewrite 
the linear equations for $\Psi_{1,2}$ in the integral form:
\begin{eqnarray}
\nonumber
&& \Psi_{1,2}(t) = A_{1,2} e^{-\beta t} \sin(\alpha t) 
+ B_{1,2} e^{-\beta t} \cos(\alpha t) \\
&& \qquad \qquad + \alpha^{-1} \int_{-\infty}^t e^{-\beta (t-t')} \sin(\alpha (t-t'))
\left[ f(t') e_{1,2} + g(t') \Psi_{1,2}(t') \right] dt',
\label{volterra-last}
\end{eqnarray}
where $A_{1,2}$, $B_{1,2}$ are some constant coefficients 
and $e_1 = 1$, $e_2 = 0$.
The kernel of the integral equations (\ref{volterra-last})  
is bounded in the variable $\tilde{\Psi}_{1,2}(t) = e^{\beta t} \Psi_{1,2}(t)$, for which we can write 
\begin{eqnarray}
\nonumber
&& \tilde{\Psi}_{1,2}(t) = A_{1,2} \sin(\alpha t) 
+ B_{1,2} \cos(\alpha t) \\
&& \qquad \qquad + \alpha^{-1} \int_{-\infty}^t \sin(\alpha (t-t'))
\left[ f(t') e^{\beta t'} e_{1,2} + g(t') \tilde{\Psi}_{1,2}(t') \right] dt'.
\label{volterra-last-tilde}
\end{eqnarray}
By the same fixed-point iterations as in the proof of Lemma \ref{lemma-1}, 
there exists the unique solutions $\tilde{\Psi}_{1,2}$ to the integral equations (\ref{volterra-last-tilde}) in a closed 
subset of Banach space $L^{\infty}(-\infty,T+(a-1) \log b)$ 
satisfying the bounds
\begin{eqnarray*}
&& \sup_{t \in (-\infty,T+(a-1)\log b]} |\tilde{\Psi}_{1,2}(t) - A_{1,2} \sin(\alpha t) - B_{1,2} \cos(\alpha t) | \leq C_{T,a}  b^{-2(1-a)}, \qquad b \geq b_{T,a},
\end{eqnarray*}
due to bounds (\ref{bound-f0-f1}). By using the transformation $\tilde{\Psi}_{1,2}(t) = e^{\beta t} \Psi_{1,2}(t)$, we obtain  (\ref{behavior-C-spiral-derivative}). 

In order to justify (\ref{behavior-C-spiral}), we substitute the decomposition $\Psi_C = \Psi_{\infty} + \Sigma$ into (\ref{eq-psi-singular}) and obtain the following persistence problem:
\begin{equation}
\label{persistence-Psi-C}
\mathcal{L}_{\infty} \Sigma = \mathcal{F},
\end{equation}
where
\begin{eqnarray*}
(\mathcal{L}_{\infty} \Sigma)(t) & := & \Sigma''(t) + (d-4) \Sigma'(t) + (3-d) \Sigma(t) + 3 \Psi_{\infty}(t)^2 \Sigma(t) + \lambda_{\infty} e^{2t} \Sigma(t) - e^{4t} \Sigma(t), \\
\mathcal{F}(t) & := & -(\lambda - \lambda_{\infty}) e^{2t} (\Psi_{\infty}(t) + \Sigma(t)) - 3 \Psi_{\infty}(t) \Sigma(t)^2 - \Sigma(t)^3.
\end{eqnarray*}
Let $\{ \Sigma_1,\Sigma_2\}$ be the fundamental system 
of the homogeneous equation $\mathcal{L}_{\infty} \Sigma = 0$ subject 
to the normalization 
$$
\left\{ \begin{array}{l} \Sigma_1(0) = 1, \\
\Sigma_1'(0) = 0, \end{array} \right. \qquad 
\left\{ \begin{array}{l} \Sigma_2(0) = 0, \\
\Sigma_2'(0) = 1.\end{array} \right.
$$
Since $\mathcal{L}_{\infty} = \mathcal{L}_0 - g$ and $g(t) = \mathcal{O}(e^{2t})$ as $t \to -\infty$, the functions $\Sigma_1(t)$ and $\Sigma_2(t)$ diverge like $\mathcal{O}(e^{-\beta t})$ as $t \to -\infty$, so that there exists a positive constant $C$ such that 
\begin{equation}
\label{bound-Sigma-1-2}
\sup_{t \in (-\infty,0]} e^{\beta t} \left( |\Sigma_1(t)| + |\Sigma_2(t)| \right) \leq C.
\end{equation}
The Wronskian relation from Liouville's theorem yields
\begin{equation}
\label{Wron-Sigma}
W(\Sigma_1,\Sigma_2)(t) = \Sigma_1(t) \Sigma_2'(t) - \Sigma_1'(t) \Sigma_2(t) = e^{-2\beta t},
\end{equation}
where $2\beta = d-4$. By variation of parameters, we can rewrite the differential equation (\ref{persistence-Psi-C}) in the integral form:
\begin{eqnarray}
\nonumber
&& \Sigma(t) = \Sigma(0) \Sigma_1(t) + \Sigma'(0) \Sigma_2(t) \\
&& \qquad \qquad + \int_{t}^0 e^{2\beta t'} 
\left[ \Sigma_1(t) \Sigma_2(t') - \Sigma_1(t') \Sigma_2(t) \right] \mathcal{F}(t') dt',
\label{volterra-Psi-C}
\end{eqnarray}
where $t < 0$.

For simplicity, let us set $C = C_{\infty}$ and consider $\lambda$ satisfying $|\lambda - \lambda_{\infty}| \leq \epsilon e^{-\beta(1-a)}$. The proof 
of the general case under the bound (\ref{ball-lambda-C}) is similar. We set $\tilde{\Sigma}(t) := e^{\beta t} \Sigma(t)$ as before and 
rewrite the integral equation (\ref{volterra-Psi-C}) in the form:
\begin{eqnarray}
\nonumber
&& \tilde{\Sigma}(t) = \Sigma(0) e^{\beta t} \Sigma_1(t) + \Sigma'(0) e^{\beta t}  \Sigma_2(t) \\
\nonumber
&& \qquad \qquad - (\lambda - \lambda_{\infty}) \int_{t}^0 K(t,t') e^{2 t'} \left[ e^{\beta t'} \Psi_{\infty}(t') + \tilde{\Sigma}(t') \right] dt' \\
&& \qquad \qquad - \int_{t}^0 K(t,t') \left[ 3 \Psi_{\infty}(t') e^{-\beta t'} \tilde{\Sigma}(t')^2 - e^{-2\beta t'} \tilde{\Sigma}(t')^3 \right] dt',
\label{volterra-Psi-C-tilde}
\end{eqnarray}
where the kernel $K(t,t') := e^{\beta (t+t')} 
\left[  \Sigma_1(t) \Sigma_2(t') - \Sigma_1(t') \Sigma_2(t) \right]$ 
satisfies the bound
\begin{equation}
\label{bound-Sigma-K}
\sup_{t \in (-\infty,0],t' \in (-\infty,0]} |K(t,t')| \leq C.
\end{equation}
which follows from (\ref{bound-Sigma-1-2}). By the smoothness of $\Psi_C$ in $\lambda$, we have $|\Sigma(0)| + |\Sigma'(0)| \leq C |\lambda - \lambda_{\infty}|$. The nonlinear terms grow as $t \to -\infty$, therefore, the fixed-point arguments cannot be closed in $L^{\infty}(-\infty,0)$. However, they can be closed in the ball $B_{\delta} \subset L^{\infty}((a-1)\log b,0)$ provided that $\delta = C\epsilon e^{-\beta (1-a)}$ with sufficiently small $\epsilon > 0$ and some $C > 0$. 
In particular, the nonlinear terms are contractive if $\epsilon$ is sufficiently small. By using the first-point iterations, there exists 
the unique solution to the integral equation (\ref{volterra-Psi-C-tilde}) satisfying 
\begin{equation}
\sup_{t \in [(a-1)\log b,0]} |\tilde{\Sigma}(t)| \leq C |\lambda - \lambda_{\infty}| \leq C \epsilon b^{-\beta (1-a)}.
\end{equation}
Since $\tilde{\Sigma}$ is smooth in $\lambda$ and $\partial_{\lambda} \tilde{\Sigma} |_{\lambda - \lambda_{\infty}} = \tilde{\Psi}_1$ constructed above, we then conclude that 
\begin{equation}
\label{bound-second-derivative}
\sup_{t \in [(a-1)\log b,0]} |\tilde{\Sigma}(t) - (\lambda - \lambda_{\infty}) \tilde{\Psi}_1 | \leq C (\lambda - \lambda_{\infty})^2 b^{\beta (1-a)} \leq C \epsilon^2 b^{-\beta (1-a)}.
\end{equation} 
Bound (\ref{behavior-C-spiral}) follows from the decomposition $\Psi_C = \Psi_{\infty} + \Sigma$, the expansion 
$\Psi_{\infty}(t) = \sqrt{d-3} + \mathcal{O}(e^{2t})$ as $t \to -\infty$, the bound (\ref{behavior-C-spiral-derivative}) on the first derivatives, 
and the bound (\ref{bound-second-derivative}) on the higher-order terms.

It remains to prove that $A_1 B_2 \neq A_2 B_1$ under Assumption 
\ref{assumption-1}. Since the differential equation (\ref{eq-psi-singular-derivative-1}) is homogeneous and $\Psi_2$ in (\ref{derivative-C-solution}) is nonzero due to the boundary conditions (\ref{asymptotics-infinity-psi}), it follows that $(A_2,B_2) 
\neq (0,0)$ by uniqueness of the zero solution in the integral equation (\ref{volterra-last-tilde}) for $\tilde{\Psi}_2$. If $A_1 B_2 = A_2 B_1$, then there exists $\mu \in \mathbb{R}$ such that $(A_1,B_1) = \mu (A_2,B_2)$ and $\Delta(t) := \Psi_1(t) - \mu \Psi_2(t)$ satisfies 
the integral equation that follows from (\ref{volterra-last}):
\begin{eqnarray}
\Delta(t) = \alpha^{-1} \int_{-\infty}^t e^{-\beta (t-t')}
\sin(\alpha (t-t'))
\left[ f(t') + g(t') \Delta(t') \right] dt'.
\label{volterra-difference}
\end{eqnarray}
By the previous arguments, there exists the unique solution 
to the integral equation (\ref{volterra-difference}) 
for $\tilde{\Delta}(t) = e^{\beta t} \Delta(t)$ in a closed subset of Banach space $L^{\infty}(-\infty,T+(a-1) \log b)$. Moreover, $\tilde{\Delta}(t) \to 0$ as $t \to -\infty$. The Wronskian between $\Delta$ and $\Psi_2$ satisfies the inhomogeneous 
equation 
\begin{equation}
\label{Wronskian-Delta}
\frac{d}{dt} e^{(d-4)t} W(\Delta,\Psi_2) = - e^{(d-4)t} f(t) \Psi_2(t).
\end{equation}
Due to the fast decay of $\Psi_C(t)$ as $t \to +\infty$ in (\ref{asymptotics-infinity-psi}), we integrate the inhomogeneous equation (\ref{Wronskian-Delta}) on $\mathbb{R}$ and obtain the contradiction with the constraint (\ref{degeneracy-1}) in Assumption \ref{assumption-1}:
$$
0 = \lim_{t \to -\infty}  e^{(d-4)t} W(\Delta,\Psi_2) = \int_{-\infty}^{\infty}  e^{(d-4)t} f(t) \Psi_2(t) dt = - 
\int_{-\infty}^{\infty}  e^{(d-2)t} \Psi_{\infty}(t) \Psi_2(t) dt \neq 0,
$$
where for the first equality we have used that $\tilde{\Psi}_2(t)$ is bounded and $\tilde{\Delta}(t)$ is decaying to zero as $t \to -\infty$. The contradiction implies that $A_1 B_2 \neq A_2 B_1$ under Assumption \ref{assumption-1}.
\end{proof}

The proof of Theorem \ref{theorem-2} for $5 \leq d \leq 12$ is developed based on Lemmas \ref{lemma-spiral-Psi} and \ref{lemma-spiral-Psi-C}.\\

{\em Proof of Theorem \ref{theorem-2} for $5 \leq d \leq 12$.}

By Theorem \ref{theorem-1}, the solution $\Psi_b(t)$ exists for
a certain value of $\lambda$ denoted by $\lambda(b)$ for every $b > 0$. 
By Lemma \ref{lemma-0}, it satisfies the asymptotic behavior (\ref{asymptotics-infinity-psi})
for uniquely selected $C = C(b)$. Therefore, for this value of $\lambda = \lambda(b)$, we have
\begin{equation}
\label{relation-b-C}
\Psi_b(t) = \Psi_{C(b)}(t), \quad t \in \mathbb{R}.
\end{equation}
By comparing the bound (\ref{behavior-Psi-spiral}) of Lemma \ref{lemma-spiral-Psi} for any fixed $T \in \mathbb{R}$
and $a \in \left(0,\frac{4}{d} \right)$ with the bound (\ref{behavior-C-spiral}) of Lemma \ref{lemma-spiral-Psi-C} at the time
instance $t = T + (a-1) \log b$, we obtain the system of nonlinear equations:
\begin{eqnarray}
\label{system-A-B}
\left\{ \begin{array}{l}  A_1 (\lambda(b) - \lambda_{\infty}) + A_2 (C(b) - C_{\infty}) = A_0 b^{-\beta} \cos(\delta_0 + \alpha \log b) + E_1 \\
B_1 (\lambda(b) - \lambda_{\infty}) + B_2 (C(b) - C_{\infty}) = A_0 b^{-\beta} \sin(\delta_0 + \alpha \log b) + E_2, \end{array} \right.
\end{eqnarray}
where coefficients $(A_1,A_2)$ and $(B_1,B_2)$ are the same as in (\ref{behavior-C-spiral-derivative}) and $(E_1,E_2)$ are error terms satisfying 
\begin{eqnarray*}
	&& E_{1,2} = \mathcal{O}(b^{-(1+a) \beta}, b^{-(2+\beta)(1-a)},
	(\lambda(b) - \lambda_{\infty})b^{-2(1-a)}, (C(b) - C_{\infty}) b^{-2(1-a)}, \\
	&& \qquad \qquad   (\lambda(b) - \lambda_{\infty})^2 b^{\beta (1-a)},(C(b) - C_{\infty})^2 b^{\beta (1-a)})
\end{eqnarray*}
as $b \to \infty$, provided that $(\lambda(b),C(b))$ satisfy the bound (\ref{ball-lambda-C}) for some $\epsilon > 0$. By Lemma \ref{lemma-spiral-Psi-C}, 
it follows that $A_1 B_2 \neq A_2 B_1$ so that the matrix 
in (\ref{system-A-B}) is invertible. By the implicit function theorem, there exist constants $A_{\infty}$, $B_{\infty}$, $\delta_{\infty}$, and $\nu_{\infty}$ such that
the unique solution to the system (\ref{system-A-B}) is given by
\begin{eqnarray}
\label{system-lambda}
\left\{ \begin{array}{l}
\lambda(b) - \lambda_{\infty} = A_{\infty} b^{-\beta} \sin(\delta_{\infty} + \alpha \log b) +
\mathcal{O}(b^{-(1+a) \beta}, b^{-(2+\beta)(1-a)},b^{-\beta - 2(1-a)}), \\
C(b) - C_{\infty} = B_{\infty} b^{-\beta} \sin(\nu_{\infty} + \alpha \log b) + \mathcal{O}(b^{-(1+a) \beta}, b^{-(2+\beta)(1-a)},b^{-\beta - 2(1-a)}).
\end{array} \right.
\end{eqnarray}
The solution (\ref{system-lambda}) satisfies the bound (\ref{ball-lambda-C}) since $b^{-\beta} \ll b^{-\beta(1-a)}$ for $a > 0$.
On the other hand, if  $a < \frac{4}{d}$, then $b^{-\beta} \gg e^{-(2+\beta)(1-a)}$, and the error terms in (\ref{system-lambda}) are smaller compared to the leading-order terms. Expansion (\ref{system-lambda}) justifies the expansion (\ref{snake}). \hspace{3.5cm} $\Box$

\vspace{0.25cm}

\begin{rem}
	Let $\{b_n\}_{n \in \mathbb{N}}$ be a sequence of roots of $\lambda(b) = \lambda_{\infty}$. It follows from \eqref{system-lambda} that
	\begin{equation}\label{b-large}
	\lim_{n \to \infty} \frac{b_{n+1}}{b_n} = e^{\frac{\pi}{\alpha}}.
	\end{equation}
	We verified the asymptotic limit \eqref{b-large} numerically. The results 
	for $d = 5$ are given in Table \ref{tab2:my_label}, where $e^{\frac{\pi}{\alpha}} \approx 5.06478$.
\end{rem}

\begin{table}[h]
	\centering
	\begin{tabular}{c|cc}
		$n$ & $b_n$ & $b_{n+1}/b_n$ \\ \hline
		1 & 3.7733455 & 5.37388 \\
		2 & 20.277514 & 5.07167\\
		3 & 102.84079 & 5.08211\\
		4 & 522.64782 & 5.06062\\
		5 & 2644.9194 & 5.06744\\
		6 & 13402.960 & 5.06352\\
		7 & 67866.139 & 5.06588\\
		8 & 343801.49 & 5.06317\\
		9 & 1740725.8 & 
	\end{tabular}
\vspace{0.25cm} 
	\caption{Approximate values of $b_n$ such that $\lambda(b_n) = \lambda_{\infty}$ for $d=5$.}
	\label{tab2:my_label}
\end{table}

Figure \ref{fig:psid5} illustrates the solutions $\Psi_{b}$  
of the second-order equation (\ref{eq-psi-singular}) for $d = 5$ 
and $b = b_1, b_3, b_6$, where $\{ b_n \}_{n \in \mathbb{N}}$ 
are defined in Table \ref{tab2:my_label}. The left panel shows that the solutions $\Psi_b$ translated in $t$ 
by $\log b$ in comparison with the solution $\Theta$ of 
the truncated equation (\ref{eq-psi-truncated}). The right panel 
shows the solutions $\Psi_b$ without translation in comparison with the limiting 
singular solution $\Psi_{\infty}$ satisfying (\ref{eq-psi-singular}) and (\ref{bc-singular}). The left panel confirms convergence of 
$\{ \Psi_{b_n}(\cdot - \log b_n)\}_{n \in \mathbb{N}}$ to $\Theta$ on $(-\infty,t_0]$ for a fixed $t_0 > 0$. The right panel confirms convergence of $\{ \Psi_{b_n}\}_{n \in \mathbb{N}}$  to $\Psi_{\infty}$ on $[t_0,\infty)$ for a fixed $t_0 < 0$.

\begin{figure}[h]
	\centering
	\includegraphics[width=0.47\textwidth]{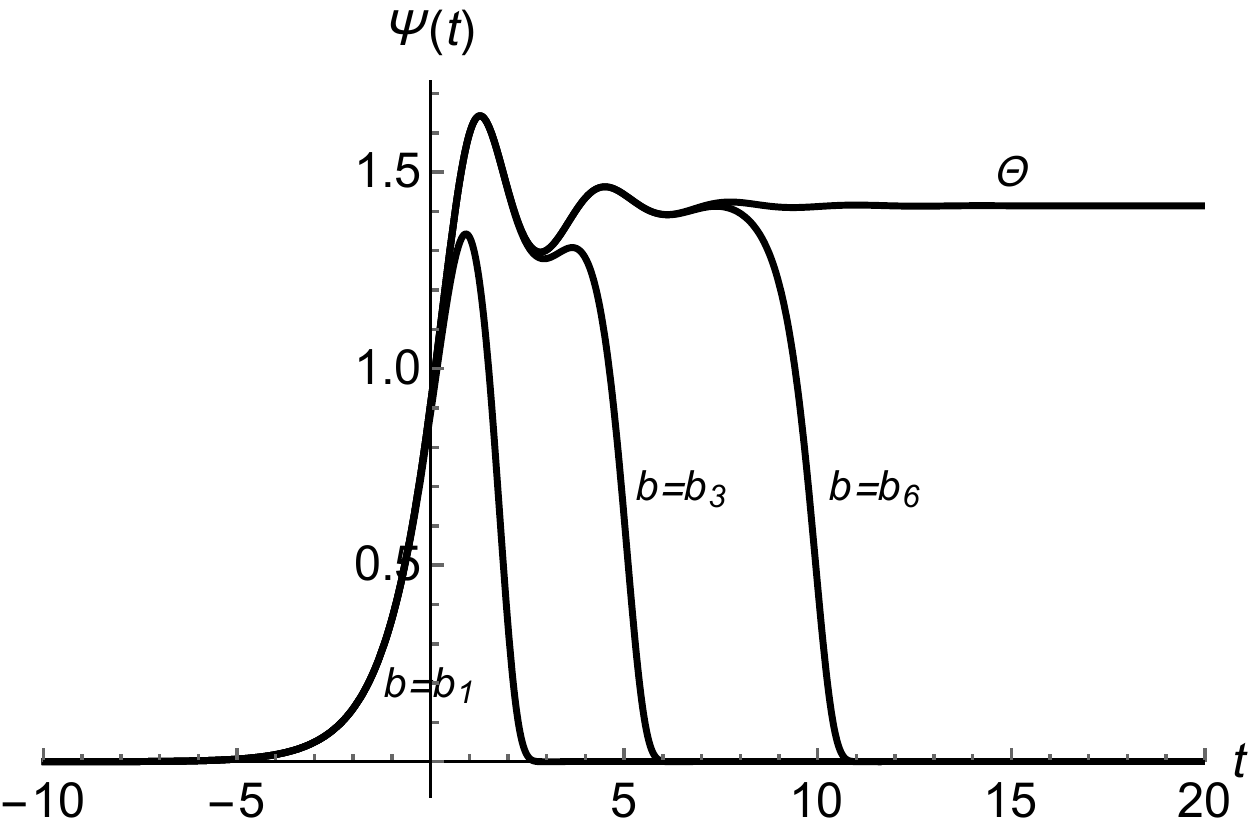}\qquad
	\includegraphics[width=0.47\textwidth]{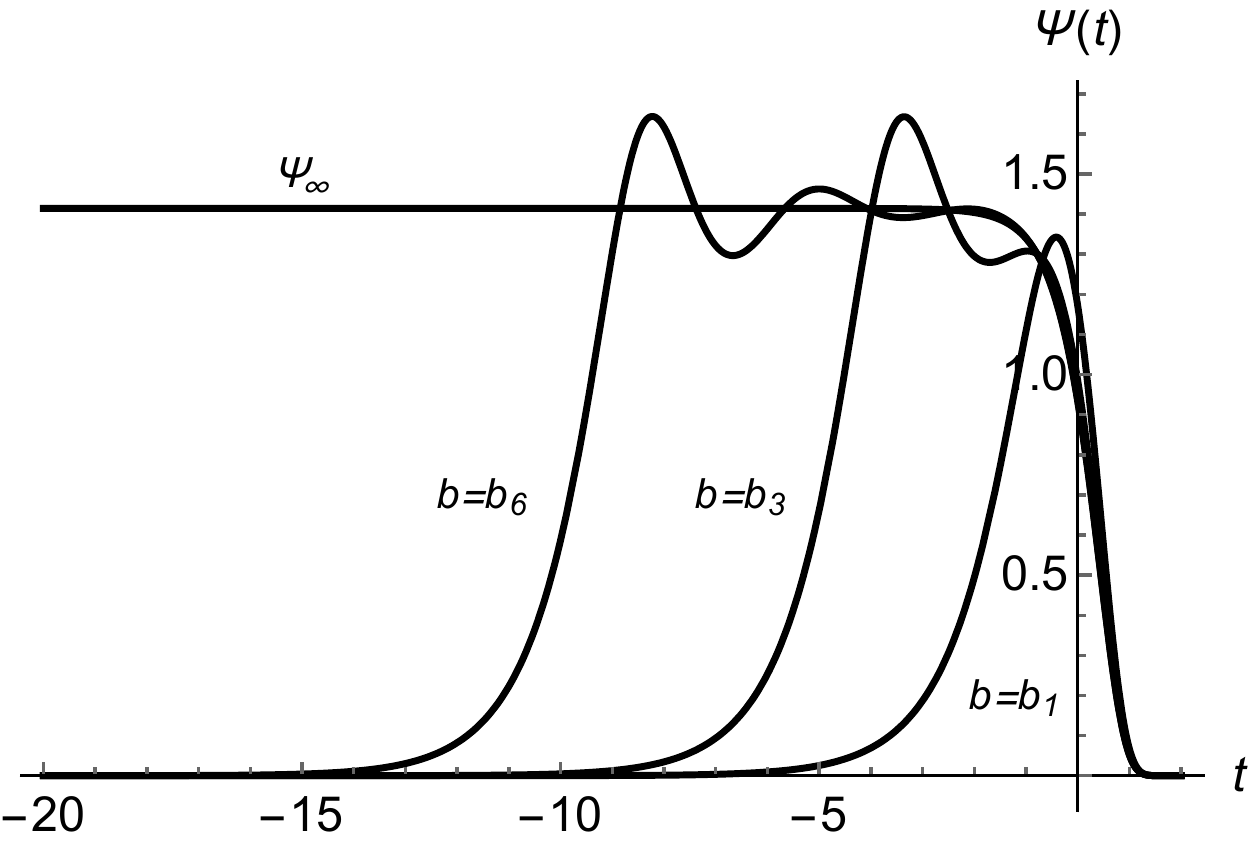}
	\caption{Plots of the solutions $\Psi_{b}$ for $d=5$ 
		and $b = b_1, b_3, b_6$ in comparison with $\Theta$ after translation of $t$ by $\log b$ (left) and with $\Psi_{\infty}$ (right). }
	\label{fig:psid5}
\end{figure}

\begin{figure}[htp]
	\centering
	\includegraphics[width=0.5\textwidth]{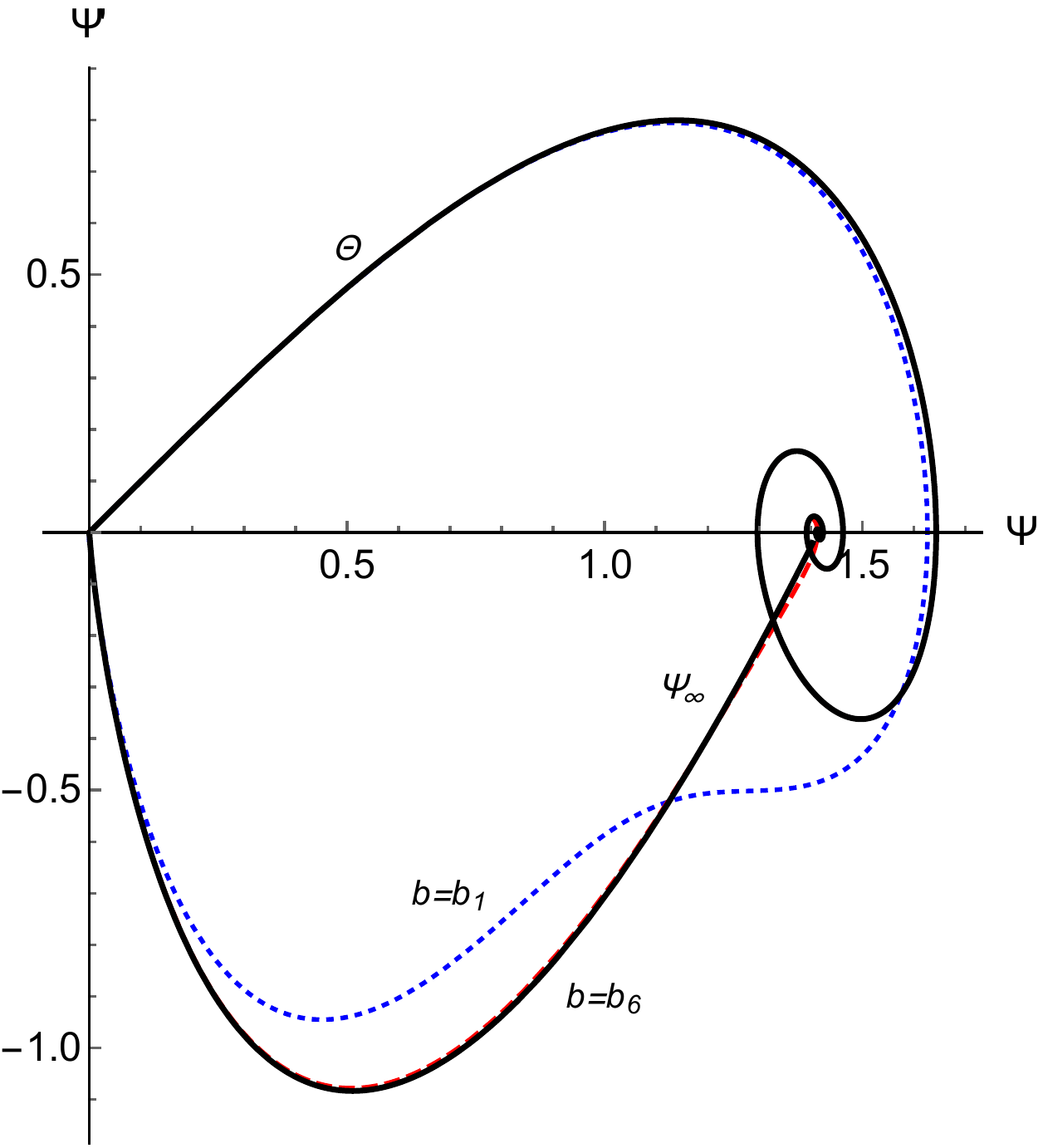}
	\caption{Solutions $\Psi_{b_1}$, $\Psi_{b_6}$, $\Theta$, and $\Psi_{\infty}$ on the phase plane $(\Psi,\Psi')$ for $d = 5$.}
	\label{fig:phaseplane}
\end{figure} 

Figure \ref{fig:phaseplane} shows solutions $\Psi_b$ for $b = b_1$ and $b = b_6$ on the phase plane $(\Psi,\Psi')$ together with the solution 
$\Theta$ of the truncated equation (\ref{eq-psi-truncated}) 
and the limiting singular solution $\Psi_{\infty}$ satisfying (\ref{eq-psi-singular}) and (\ref{bc-singular}). The difference of $\Psi_{b = b_6}$ (red dotted line) from
$\Theta$ and $\Psi_{\infty}$ is almost invisible, 
whereas the difference is large in the case of $\Psi_{b = b_1}$ (blue dotted line).

\subsection{Monotone behavior for $d \geq 13$}

Here we state and prove the corresponding modifications of results of
Lemmas \ref{lemma-spiral-Theta}, \ref{lemma-spiral-Psi}, and \ref{lemma-spiral-Psi-C}
in the case $d \geq 13$. The following lemma described the exponential behavior of $\Theta(t)$ as $t \to +\infty$.

\begin{lemma}
\label{lemma-monotone-Theta}
Fix $d \geq 13$. There exist $t_0 > 0$ (sufficiently large), $A_0 > 0$, $B_0 > 0$, and $C_0 > 0$
such that the unique solution $\Theta$ of Lemma \ref{lemma-Psi-inf} satisfies the following behavior:
\begin{equation}
\label{behavior-monotone}
\sup_{t \in [t_0,\infty)} |\Theta(t) - \sqrt{d-3} - A_0 e^{\kappa_+ t} - B_0 e^{\kappa_- t}| \leq C_0 e^{2 \kappa_+ t_0},
\end{equation}
where $(\kappa_+,\kappa_-)$ are given by (\ref{linearization-nonzero}).
\end{lemma}

\begin{proof}
The equilibrium point $(\sqrt{d-3},0)$ is a stable sink of the truncated equation (\ref{eq-psi-truncated}) for $d \geq 13$ due to the roots (\ref{linearization-nonzero}) of the characteristic equation (\ref{char-eq-roots}) satisfying (\ref{real-roots}). Quadratic terms
beyond the linearization at $(\sqrt{d-3},0)$ can be removed by a near-identity transformation under the non-resonance condition
$\kappa_- \neq 2 \kappa_+$ which is satisfied since there are no integer solutions of the quadratic equation $d^2 - 17 d + 43 = 0$.
By the Hartman--Grobman theorem, there exists a $C^2$-diffeomorphism, under which the dynamics
of the truncated equation (\ref{eq-psi-truncated}) near $(\sqrt{d-3},0)$ is conjugate to the dynamics
of the linearized equation. The asymptotic behavior (\ref{behavior-monotone}) follows from the solution of
the linearized equation and the existence of the $C^2$-diffeomorphism.
\end{proof}

\begin{rem}
Because $\kappa_- < \kappa_+ < 0$, the function $\Theta(t)$ approaches $\sqrt{d-3}$ monotonically 
and the bound (\ref{behavior-monotone}) can be rewritten in a simpler way:
\begin{equation}
\label{behavior-monotone-explicit}
\sup_{t \in [t_0,\infty)} |\Theta(t) - \sqrt{d-3} - A_0 e^{\kappa_+ t} | \leq C_0 \max\{e^{\kappa_- t_0},e^{2 \kappa_+ t_0}\},
\end{equation}
from which the monotone behavior of $\Theta(t)$ as $t \to +\infty$ is obvious.
\end{rem}

Using Lemma \ref{lemma-monotone-Theta}, the statement of Lemma \ref{lemma-spiral-Psi}
is modified to yield the exponential behavior of the solution $\Psi_b(t)$ 
at the intermediate values of $t$ as $b \to \infty$.

\begin{lemma}
\label{lemma-monotone-Psi}
Fix $d \geq 13$ and $\lambda \in \mathbb{R}$. For fixed $T > 0$ and $a \in(0,a_0)$, 
where $a_0 \in (0,1)$ is defined by (\ref{a-0}), there exists $b_{T,a} > 0$ and $C_{T,a} > 0$ such that 
the unique solution $\Psi_b$ to the second-order equation (\ref{eq-psi-singular})
with the asymptotic behavior (\ref{asym-b-solution}) satisfies for $b \geq b_{T,a}$:
\begin{eqnarray}
|\Psi_b(T + (a-1) \log b) - \sqrt{d-3} - A_0 b^{a \kappa_+} e^{\kappa_+ T} | 
\leq C_{T,a} \max\{b^{a \kappa_-},b^{2a \kappa_+},b^{-2(1-a)}\}, 
\label{behavior-Psi-monotone}
\end{eqnarray}
where $(\kappa_+,\kappa_-)$ are given by (\ref{linearization-nonzero})
and $A_0$ is defined in (\ref{behavior-monotone-explicit}).
\end{lemma}

\begin{proof}
The proof of the bound (\ref{bound-solution-b-extended}) remains the same for every $d \geq 5$.
Bound (\ref{behavior-Psi-monotone}) follows from (\ref{bound-solution-b-extended}) and (\ref{behavior-monotone-explicit})
since $a \log b \to +\infty$ as $b \to \infty$ if $a > 0$ and $b^{a \kappa_+} \gg b^{-2(1-a)}$ if $a < a_0$, 
where 
\begin{equation}
\label{a-0}
a_0 := \frac{2}{2+|\kappa_+|} = \frac{4}{d - \sqrt{d^2 - 16 d + 40}} = \frac{d + \sqrt{d^2 - 16 d + 40}}{2 (2d - 5)}.
\end{equation}
Note that $a_0  < \frac{1}{2}$ for every $d \geq 13$.
\end{proof}

Finally, we recall again that $\Psi_C$ coincides with $\Psi_{\infty}$ for $(\lambda,C) = (\lambda_{\infty},C_\infty)$ and define $\Psi_{1,2}$ as in 
(\ref{derivative-C-solution}). 
We add the following technical assumption.
\begin{assumption}
	\label{assumption-2}
	Uniquely defined functions $\Psi_{\infty}$ and $\Psi_2$ are assumed to satisfy the following non-degeneracy assumptions:
	\begin{equation}
	\label{degeneracy-2}
	\int_{-\infty}^{\infty} e^{(d-2) t} \Psi_{\infty}(t) \Psi_2(t) dt \neq 0
	\end{equation}
and 
		\begin{equation}
	\label{degeneracy-3}
	\lim\limits_{t \to -\infty} e^{-\kappa_- t} \Psi_2(t) \neq 0.
	\end{equation} 
\end{assumption}

\begin{rem}
	Compared to Assumption \ref{assumption-1}, we have an additional 
	assumption (\ref{degeneracy-3}) in Assumption \ref{assumption-2}. 
	This additional condition excludes solutions of the homogeneous 
	equation $\mathcal{L}_{\infty} \Psi_2 = 0$ decaying to zero as $t \to +\infty$ to grow slowly as $\mathcal{O}(e^{\kappa_+ t})$ as $t \to -\infty$. 
\end{rem}

The following lemma determines the exponential behavior of the solution $\Psi_{C}$ for $(\lambda,C)$ near the point
$(\lambda_{\infty},C_{\infty})$.

\begin{lemma}
\label{lemma-monotone-Psi-C}
Fix $d \geq 13$. There exist $L_1, L_2 \in \mathbb{R}$ such that 
\begin{equation}
\label{limits-B}
L_{1,2} = \lim_{t \to -\infty} e^{-\kappa_- t} \Psi_{1,2}(t).
\end{equation}
If Assumption \ref{assumption-2} is satisfied, then $L_2 \neq 0$ 
and for fixed $T > 0$ and $a \in (0,1)$,
there exist $b_{T,a} > 0$, $C_{T,a} > 0$, and $\Delta_0 \neq 0$  
such that $\Delta(t) := \Psi_1(t) - L_2^{-1} L_1 \Psi_2(t)$ 
satisfy for every $t \in (-\infty,(a-1) \log b + T]$: 
\begin{equation}
\label{behavior-C-monotone-derivative}
|\Delta(t) - \Delta_0 e^{\kappa_+ t} | \leq 
C_{T,a} b^{-2(1-a)} e^{\kappa_+ t}, \qquad b \geq b_{T,a},
\end{equation}
where $(\kappa_+,\kappa_-)$ are given by (\ref{linearization-nonzero}). Consequently, there exists $\epsilon_0 > 0$ such that 
for every $\epsilon \in (0,\epsilon_0)$ and for every 
$\lambda \in \mathbb{R}$ satisfying 
\begin{equation}
\label{ball-lambda-C-monotone}
|\lambda - \lambda_{\infty}| \leq \epsilon b^{\kappa_+ (1-a)},
\end{equation}
it is true for every 
$b \geq b_{T,a}$ and every $t \in [(a-1)\log b,(a-1) \log b + T]$ that 
\begin{eqnarray}
\nonumber
&& | \Psi_{C_{\infty}-L_2^{-1} L_1(\lambda - \lambda_{\infty})}(t) - \sqrt{d-3} 
- \Delta_0 (\lambda - \lambda_{\infty}) e^{\kappa_+ t} | \\
\label{behavior-C-monotone}
&& \qquad \leq C_{T,a} \left( 
b^{-2(1-a)} +   (\lambda(b) - \lambda_{\infty}) b^{-(2+\kappa_+)(1-a)} +  (\lambda(b) - \lambda_{\infty})^2 b^{-2\kappa_+(1-a)} \right),
\end{eqnarray}
where $b_{T,a}$ and $C_{T,a}$ are adjusted appropriately. 
\end{lemma}

\begin{proof}
The proof of Lemma \ref{lemma-monotone-Psi-C} follows the same steps as the proof of Lemma \ref{lemma-spiral-Psi-C} but incorporates the different exponential behavior of the solutions $\Psi_{1,2}(t)$ 
in (\ref{derivative-C-solution}) as $t \to -\infty$.
By variation of parameters, the linear equations for $\Psi_{1,2}$ can be rewritten in the integral form:
\begin{eqnarray}
\nonumber
&& \Psi_{1,2}(t) = A_{1,2} e^{\kappa_+ t} + 
B_{1,2} e^{\kappa_- t} \\
&& \qquad \qquad + \frac{1}{\kappa_+ - \kappa_-} 
\int_{t}^{+\infty} \left[ e^{\kappa_-(t-t')} - e^{\kappa_+(t-t')} \right]
\left[ f(t') e_{1,2} + g(t') \Psi_{1,2}(t') \right] dt',
\label{volterra-last-monotone-plus}
\end{eqnarray}
where $A_{1,2}$, $B_{1,2}$ are some constant coefficients and 
$e_1 = 1$, $e_2 = 0$. Since $\kappa_- < \kappa_+ < 0$, 
whereas $f(t)$ and $g(t) \Psi_{1,2}(t)$ decays to zero fast as $t \to +\infty$,
the integral kernel in (\ref{volterra-last-monotone-plus}) becomes bounded in the variable 
$\tilde{\Psi}_{1,2}(t) := e^{-\kappa_- t} \Psi_{1,2}(t)$ on $[t_0,+\infty)$ for every $t_0 \in \mathbb{R}$. 
The existence of $\tilde{\Psi}_{1,2}$ in $L^{\infty}(t_0,\infty)$ is guaranteed by the Banach fixed-point theorem and 
the solutions $\tilde{\Psi}_{1,2}$ are extended globally on $\mathbb{R}$. 
In the limit $t \to -\infty$, we obtain
\begin{eqnarray*}
L_{1,2} & := & \lim_{t \to -\infty} e^{-\kappa_- t} \Psi_{1,2}(t)  \\
& = & B_{1,2} +  \frac{1}{\kappa_+ - \kappa_-}  
\int_{-\infty}^{+\infty} e^{-\kappa_- t'}
\left[ f(t')  e_{1,2} + g(t') \Psi_{1,2}(t') \right] dt'.
\end{eqnarray*}
Hence, $L_{1,2}$ are bounded. Since $L_2 \neq 0$ due to the constraint (\ref{degeneracy-3}) in Assumption \ref{assumption-2}, we can define 
$$
\Delta(t) := \Psi_1(t) - L_2^{-1} L_1 \Psi_2(t),
$$
so that $\lim\limits_{t \to -\infty}  e^{-\kappa_- t} \Delta(t) = 0$.
By variation of parameters, the linear equation for $\Delta$ can be rewritten in the integral form:
\begin{eqnarray}
\Delta(t) = \Delta_0 e^{\kappa_+ t} + 
\frac{1}{\kappa_+ - \kappa_-} 
\int_{-\infty}^t \left[ e^{\kappa_+(t-t')} - e^{\kappa_-(t-t')} \right]
\left[ f(t') + g(t') \Delta(t') \right] dt',
\label{volterra-last-monotone}
\end{eqnarray}
where $\Delta_0$ is some constant coefficient. 

\begin{rem}
The integral equation (\ref{volterra-last-monotone}) is 
different from the one which would follow from the integral equation (\ref{volterra-last-monotone-plus}) in the variable $\Delta(t)$ so that 
$\Delta_0 \neq A_1 - L_2^{-1} L_1 A_2$ generally.
While (\ref{volterra-last-monotone-plus}) is useful in the limit $t \to +\infty$, (\ref{volterra-last-monotone}) is useful in the limit $t \to -\infty$.
\end{rem}

The integral kernel in (\ref{volterra-last-monotone}) becomes bounded in the variable  $\tilde{\Delta}(t) := e^{-\kappa_+ t} \Delta(t)$, for which it 
can written in the form
\begin{eqnarray}
\tilde{\Delta}(t) = \Delta_0 + 
\frac{1}{\kappa_+ - \kappa_-} 
\int_{-\infty}^t \left[ 1 - e^{-(\kappa_+-\kappa_-)(t-t')} \right]
\left[ f(t') e^{-\kappa_+ t'} + g(t') \tilde{\Delta}(t') \right] dt',
\label{volterra-last-monotone-tilde}
\end{eqnarray}
By the same fixed-point iterations as in the proof of Lemma \ref{lemma-1}, 
there exist the unique solutions $\tilde{\Delta}$ to the integral equation (\ref{volterra-last-monotone-tilde}) in a closed 
subset of Banach space $L^{\infty}(-\infty,T+(a-1) \log b)$ 
satisfying the bounds
\begin{eqnarray*}
	&& \sup_{t \in (-\infty,T+(a-1)\log b)} |\tilde{\Delta}(t) - \Delta_0 | \leq C_{T,a}  b^{-2(1-a)}, \qquad b \geq b_{T,a},
\end{eqnarray*}
due to bounds (\ref{bound-f0-f1}). Since $\tilde{\Delta}(t) = e^{-\kappa_+ t} \Delta(t)$, we obtain the bounds (\ref{behavior-C-monotone-derivative}). 

The linear combination in $\Delta = \Psi_1 - L_2^{-1} L_1 \Psi_2$ corresponds to the choice of 
$$
C - C_{\infty} = -L_2^{-1} L_1 (\lambda - \lambda_{\infty}).
$$ 
The second derivatives of $\Psi_C$ in $(\lambda,C)$ 
grow like $\mathcal{O}(e^{2\kappa_+t})$ as $t \to -\infty$. Similarly to 
the bound (\ref{bound-second-derivative}) in the proof of Lemma \ref{lemma-spiral-Psi-C}, one can justify the bound 
\begin{equation}
\label{bound-second-derivative-monotone}
\sup_{t \in [(a-1) \log b, 0]} |\Psi_{C_{\infty}-L_2^{-1} L_1 (\lambda - \lambda_{\infty})}(t) - \Psi_{\infty}(t) - (\lambda - \lambda_{\infty}) \Delta(t) | \leq C (\lambda - \lambda_{\infty})^2 e^{-2\kappa_+ (1-a)}
\end{equation}
if $|\lambda - \lambda_{\infty}| \leq \epsilon e^{\kappa_+(1-a)}$ with sufficiently small $\epsilon > 0$. 
Bound (\ref{behavior-C-monotone}) follows from 
the expansion $\Psi_{\infty}(t) = \sqrt{d-3} + \mathcal{O}(e^{2t})$ as $t \to -\infty$ and the bounds (\ref{behavior-C-monotone-derivative}) and (\ref{bound-second-derivative-monotone}).

Finally, it is proven similarly to the proof of Lemma \ref{lemma-spiral-Psi-C} that $\Delta_0 = 0$ is in contradiction with the condition (\ref{degeneracy-2}) of Assumption \ref{assumption-2}. Hence, $\Delta_0 \neq 0$.
\end{proof}

We end this section with the formal proof of Theorem \ref{theorem-2} for $d \geq 13$.\\

{\em Proof of Theorem \ref{theorem-2} for $d \geq 13$.}

We match again the solutions $\Psi_b(t)$ and $\Psi_C(t)$ as in (\ref{relation-b-C}). 
By comparing the bound (\ref{behavior-Psi-monotone}) of Lemma \ref{lemma-monotone-Psi} for any fixed $T \in \mathbb{R}$
and $a \in (0,a_0)$ with the bound (\ref{behavior-C-monotone}) of Lemma \ref{lemma-monotone-Psi-C} at the time
instance $t = T + (a-1) \log b$, we obtain the nonlinear equation:
\begin{eqnarray}
\label{system-A-B-monotone}
\Delta_0 (\lambda(b) - \lambda_{\infty}) = A_0 b^{\kappa_+}  
+ E, 
\end{eqnarray}
where coefficients $A_0$ and $\Delta_0$ are the same as in 
(\ref{behavior-Psi-monotone}) and (\ref{behavior-C-monotone-derivative}) 
respectively and the error term $E$ satisfies 
$$
E = \mathcal{O}(b^{(1-a) \kappa_+ + a \kappa_-},b^{(1+a) \kappa_+},
b^{-(2-\kappa_+)(1-a)},  (\lambda(b) - \lambda_{\infty}) b^{-2(1-a)}, 
 (\lambda(b) - \lambda_{\infty})^2 b^{-\kappa_+(1-a)})
$$ 
as $b \to \infty$, provided that $\lambda(b)$ satisfies the bound 
(\ref{ball-lambda-C-monotone}) for some $\epsilon > 0$.
Since $(1-a) \kappa_+ + a \kappa_- < \kappa_+< 0$, 
it follows that $b^{(1-a) \kappa_+ + a \kappa_-} \ll b^{\kappa_+}$. 
Similarly, we have already checked that 
$b^{-(2-\kappa_+)(1-a)} \ll b^{\kappa_+}$ if $a < a_0$, where $a_0$ is given by (\ref{a-0}).

Since $\Delta_0 \neq 0$ by Lemma \ref{lemma-monotone-Psi-C}, 
there exists the unique solution to the nonlinear equation (\ref{system-A-B-monotone}) by the implicit function theorem 
and the unique solution for $\lambda(b)$ satisfies 
\begin{eqnarray}
\label{system-lambda-monotone}
\lambda(b) - \lambda_{\infty} = \Delta_0^{-1} A_0  b^{\kappa_+} + \mathcal{O}(b^{(1-a) \kappa_+ + a \kappa_-},b^{(1+a) \kappa_+},b^{-(2-\kappa_+)(1-a)},b^{-2(1-a)+\kappa_+}).
\end{eqnarray}
Since $a > 0$ and $\kappa_+ < 0$, it follows that $b^{\kappa_+} \ll b^{\kappa_+(1-a)}$ so that $\lambda(b)$ in (\ref{system-lambda-monotone}) belongs to the bound (\ref{ball-lambda-C-monotone}). 
The expansion (\ref{system-lambda-monotone}) justifies the expansion 
(\ref{monotone}). \hspace{5.8cm} $\Box$

\begin{figure}[h]
	\centering
	\includegraphics[width=0.45\textwidth]{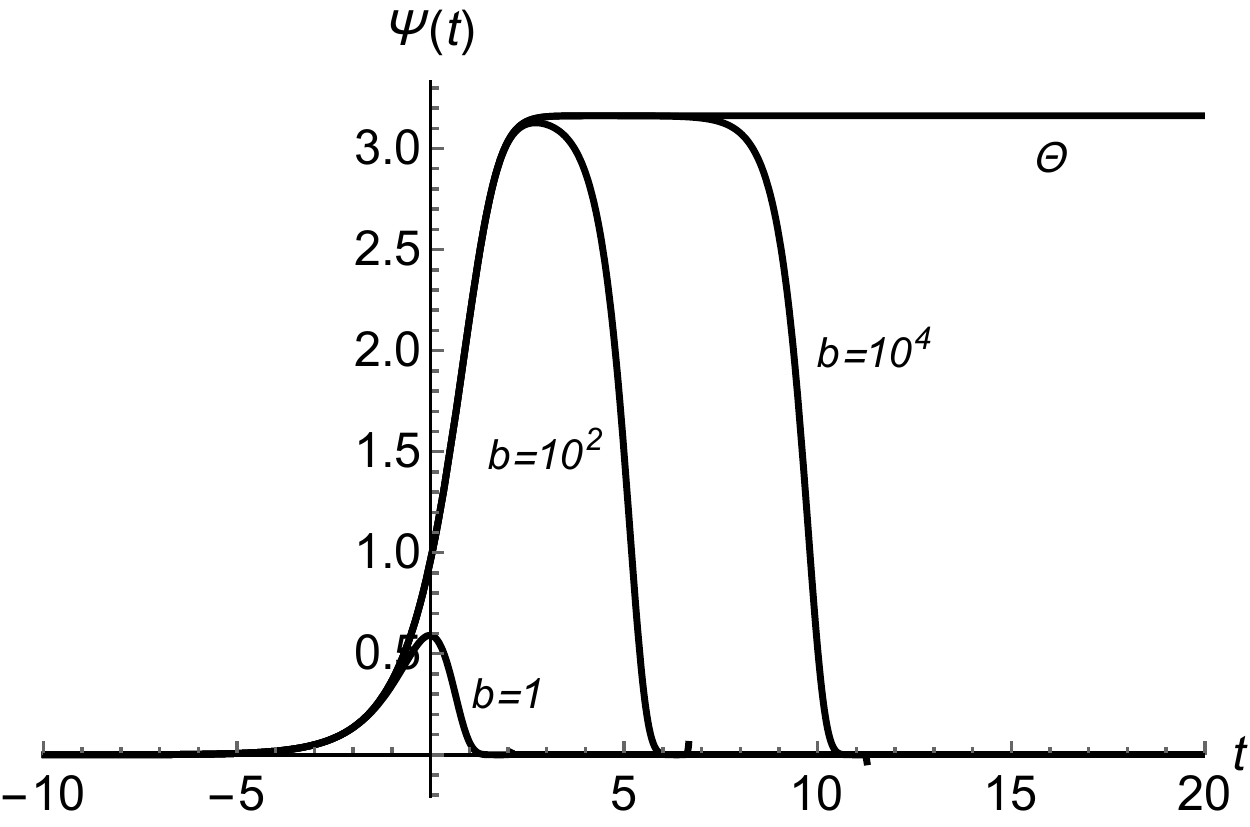}\qquad
	\includegraphics[width=0.45\textwidth]{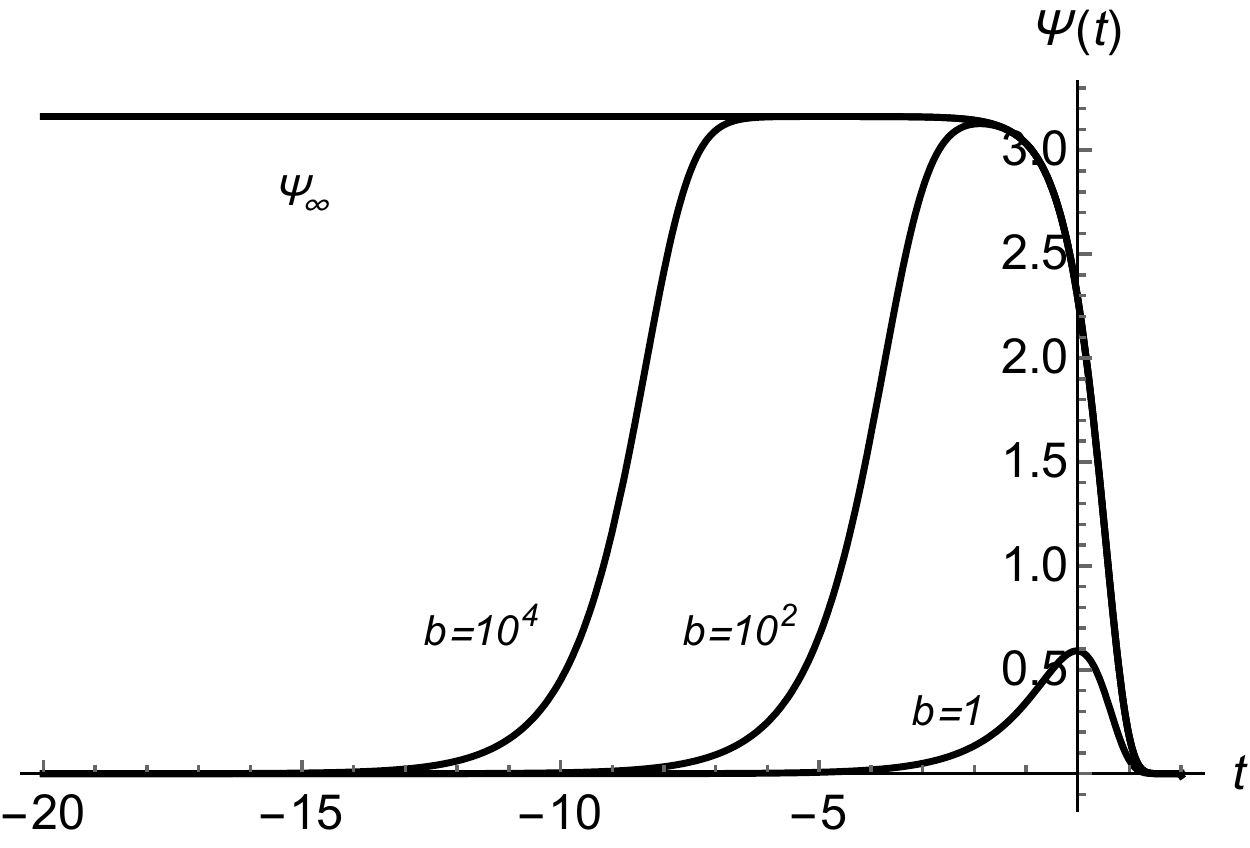}
	\caption{Plots of the solutions $\Psi_b$ for $\lambda = \lambda(b)$ and specific values of $b$ for $d=13$ in comparison with 
		$\Theta$ after translation of $t$ by $\log b$ (left) and with $\Psi_{\infty}$ (right). }
	\label{fig:psid13}
\end{figure}

Figure \ref{fig:psid13} illustrates the solutions $\Psi_b$ 
of the second-order equation (\ref{eq-psi-singular}) with $\lambda = \lambda(b)$ for $d = 13$ and $b = 1, 10^2, 10^4$. 
The left panel shows that the solutions $\Psi_b$ translated in $t$ 
by $\log b$ in comparison with the solution $\Theta$ of 
the truncated equation (\ref{eq-psi-truncated}). The right panel 
shows the solutions $\Psi_b$ without translation in comparison with the limiting 
singular solution $\Psi_{\infty}$ satisfying (\ref{eq-psi-singular}) and (\ref{bc-singular}). Convergence $\Psi_b(\cdot - \log b) \to \Theta$ as $b \to \infty$ on $(-\infty,t_0]$ for a fixed $t_0 > 0$ 
is obvious from the left panel, whereas convergence $\Psi_b \to \Psi_{\infty}$ as $b \to \infty$ on $[t_0,\infty)$ for a fixed $t_0 < 0$ 
is obvious from the right panel.

Figure \ref{fig:phaseplane13} shows two solutions $\Psi_b$ with $b = 1$ and $b = 10^2$ on the phase plane $(\Psi,\Psi')$ together with the solution 
$\Theta$ of the truncated equation (\ref{eq-psi-truncated}) 
and the limiting singular solution $\Psi_{\infty}$ satisfying (\ref{eq-psi-singular}) and (\ref{bc-singular}). The difference of $\Psi_{b = 10^2}$ (red dotted line) from
$\Theta$ and $\Psi_{\infty}$ is almost invisible, 
whereas the difference is large in the case of $\Psi_{b = 1}$ (blue dotted line).

\begin{figure}[h]
	\centering
	\includegraphics[width=0.38\textwidth]{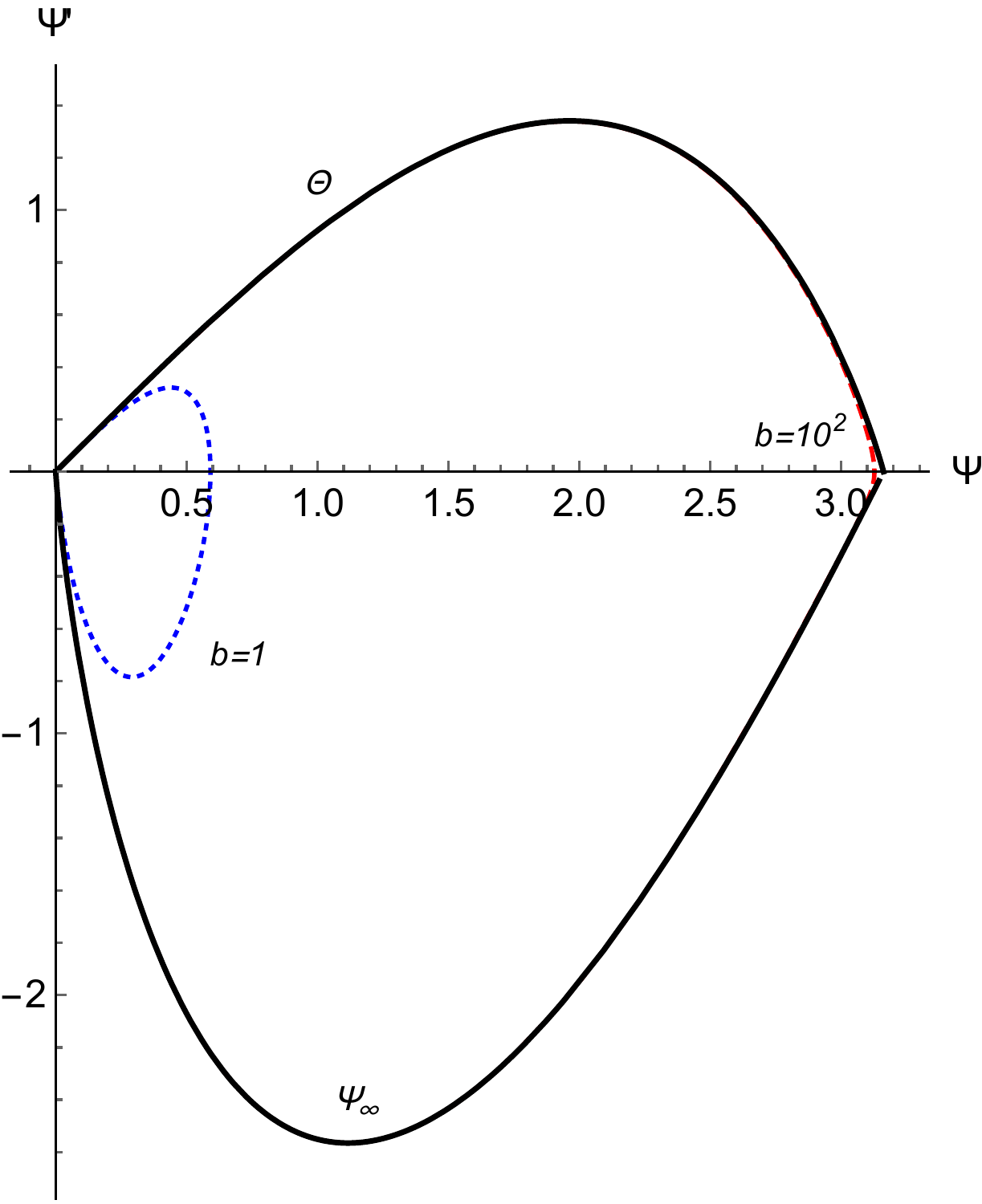}
	\caption{The solutions $\Psi_{b=1}$, $\Psi_{b=10^2}$, $\Theta$, and $\Psi_{\infty}$ on the phase plane $(\Psi,\Psi')$ for $d = 13$.}
	\label{fig:phaseplane13}
\end{figure}

\section{Conclusion}
\label{sec-conclusion}

We have considered the existence of the ground state in the energy-supercritical cubic Gross--Pitaevskii equation with a harmonic 
potential in dimension $d \geq 5$. The solution curve is unbounded 
and the supremum norm of the ground state given by the parameter $b$ diverges along the solution curve. 
In this limit, the eigenvalue parameter $\lambda = \lambda(b)$ displays the oscillatory (snaking) behavior in $b$ if $5 \leq d \leq 12$ and the monotone behavior if $d \geq 13$. 
This resembles the behavior of the ground states in the Liouville--Bratu--Gelfand problem \cite{JS,JL} and in the Dirichlet 
problem for the stationary Gross--Pitaevskii equation in a ball and without the harmonic potential \cite{Budd_Norbury1987,Budd1989,DF}. 
Our results extend and clarify the previous works \cite{Selem2011,SK2012,Selem2013}. 

It remains open to study spectral and orbital stability of the ground state by computing the Morse 
index of the Jacobian operator associated with the stationary Gross--Pitaevskii equation. See \cite{GuoWei,KikuchiWei} for computations of the Morse index of the limiting singular solutions for the Dirichlet problem. While the ground state is expected to be spectrally unstable 
near the limiting singular solution due to the oscillatory behavior for $5 \leq d \leq 12$, there is a possibility that the ground state is spectrally stable for $d \geq 13$ if the Morse index is one and the map from the eigenvalue parameter $\lambda$ to the mass $\mu := M(\mathfrak{u})$ is monotonically decreasing. The latter condition is known as the Vakhitov--Kolokolov stability criterion (see Chapter 4 in \cite{Pel}).

It seems interesting that higher dimensions $d \geq 13$ may re-enforce stability of the ground state, which is lost in $5 \leq d \leq 12$ past the first turning point along the solution curve. The stability problem will be studied in our forthcoming work.

\vspace{0.25cm}

{\bf Acknowledgements.} This research was supported by the Polish National Science Center grant no. 2017/26/A/ST2/00530.

\bibliographystyle{siam}
\bibliography{references}

\end{document}